\documentclass[UKenglish]{lmcsNologo}


\usepackage{microtype}
\usepackage[utf8]{inputenc}
\usepackage{amsmath, amsthm, amssymb, mathtools, stmaryrd}
\usepackage{tikz}
\usetikzlibrary{decorations.pathreplacing,patterns}
\usepackage{relsize}
\tikzset{fontscale/.style = {font=\relsize{#1}}
    }
\usepackage{xspace}
\usepackage{booktabs}
\usepackage{bbm}
\usepackage{graphicx}
\usepackage{enumitem}
\usepackage{lmodern}
\usepackage{todonotes}
\usepackage{hyperref}

 \DeclareMathOperator{\ifp}{ifp}

\DeclareMathOperator{\im}{im}
\DeclareMathOperator{\dom}{dom}
\DeclareMathOperator{\Aut}{Aut}
\DeclareMathOperator{\characteristic}{char}

\DeclarePairedDelimiter\floor{\lfloor}{\rfloor}

\makeatletter
\newcommand{\@abbrev}[3]{
  \def\c@a@def##1{
      \if ##1.
        \relax
      \else
        \@ifdefinable{\@nameuse{#1##1}}{\@namedef{#1##1}{#2##1}}
        \expandafter\c@a@def
      \fi
    }
  \c@a@def #3.
}
\@abbrev{bb}{\mathbb}{ABCDEFGHIJKLMNOPQRSTUVWXYZ}
\@abbrev{b}{\bar}{abcdefghijklmnopqrstuvwxyz}
\@abbrev{bf}{\mathbf}{ABCDEFGHIJKLMNOPQRSTUVWXYZabcdefghijklmnopqrstuvwxyz}
\@abbrev{bit}{\boldsymbol}{
ABCDEFGHIJKLMNOPQRSTUVWXYZabcdefghijklmnopqrstuvwxyz}
\@abbrev{mc}{\mathcal}{ABCDEFGHIJKLMNOPQRSTUVWXYZ}
\@abbrev{mb}{\mathbb}{ABCDEFGHIJKLMNOPQRSTUVWXYZ}
\@abbrev{mf}{\mathfrak}{
ABCDEFGHIJKLMNOPQRSTUVWXYZabcdefghijklmnopqrstuvwxyz}
\@abbrev{rm}{\mathrm}{ABCDEFGHIJKLMNOPQRSTUVWXYZabcdefghijklmnopqrstuvwxyz}
\@abbrev{scr}{\mathscr}{
ABCDEFGHIJKLMNOPQRSTUVWXYZabcdefghijklmnopqrstuvwxyz}
\@abbrev{sf}{\mathsf}{ABCDEFGHIJKLMNOPQRSTUVWXYZabcdefghijklmnopqrstuvwxyz}
\makeatother

\newcommand{\N}{\mbN}
\newcommand{\Inv}{\text{Inv}}
\newcommand{\field}{\mathbb}

\renewcommand{\phi}{\varphi}
\renewcommand{\theta}{\vartheta}
\newcommand{\tup}{\bar}

\newcommand{\Str}{\textup{Str}}

\newcommand{\CFIgraph}[3]{\text{\sf CFI}\,[#1; #2; #3]}
\newcommand{\CFIclass}[2]{\text{\sf CFI}\,[#1; #2]}

\newcommand{\Primes}{\mbP}
\newcommand{\taucfi}{\tau_{\text{\sf CFI}}}

\renewcommand{\mod}{\ensuremath{\text{mod }}}

\newcommand{\I}{\mcI}

\newcommand{\FPC}{\textup{FPC}\xspace}
\newcommand{\FO}{\textup{FO}\xspace}
\newcommand{\FOC}{\textup{FOC}\xspace}
\newcommand{\LFP}{\textup{LFP}\xspace}
\newcommand{\IFP}{\textup{IFP}\xspace}
\newcommand{\LL}{\textup{L}\xspace}
\newcommand{\LC}{\textup{C}\xspace}

\newcommand{\cequivx}[1]{\ensuremath{\equiv^{#1}}}
\newcommand{\cequivk}{\cequivx k}

\newcommand{\simequivx}[2]{\ensuremath{\equiv^\text{IM}_{#1, #2}}}

\newcommand{\IMequiv}{\simequivx}

\newcommand{\sms}{\ensuremath{\text{SimMatSim}}\xspace}
\newcommand{\HomMat}[2]{\ensuremath{\text{H}_{#1,#2}}}
\newcommand{\StabMat}[1]{\ensuremath{\text{C}_{#1}}}
\newcommand{\StabMatDiag}[1]{\ensuremath{\text{C}_{#1}^{\text{D}}}}
\newcommand{\HomMatDiag}[2]{\ensuremath{\text{H}_{#1,#2}^{\text{D}}}}

\newcommand{\idmat}[1]{\ensuremath{\text{Id}_{#1}}}

\newcommand{\MatAlg}[2]{\ensuremath{\text{Mat}_{#1 \times #1}({#2})}}
\newcommand{\MatMod}[3]{\ensuremath{\text{Mat}_{#1 \times #2}({#3})}}

\newcommand{\cip}{\ensuremath{\text{cip}}\xspace}

\newcommand{\nmo}{\ensuremath{{n-1}}}

\newcommand{\mmo}{\ensuremath{{m-1}}}

\newcommand{\smo}{\ensuremath{{s-1}}}

\newcommand{\lmo}{\ensuremath{{\ell-1}}}

\newcommand{\inseg}[1]{\ensuremath{[#1]}}

\newcommand{\Diag}{\ensuremath{\text{Diag}}}

\newcommand{\fbg}{\ensuremath{\text{f-block generated}}\xspace}

\newcommand{\locsimsim}{\ensuremath{\text{loc-sim similar}}\xspace}

\newcommand{\ctype}{\ensuremath{\textsc{ct}}}

\newcommand{\clalgebra}[3]{\text{\sf Alg}\,[#1; #2; #3]}

\newcommand{\CA}[3]{\text{\sf CA}\,[#1; #2; #3]}

\newcommand{\linisom}[3]{\ensuremath{({#1};{#2};{#3})\text{-isomorphic}}}
 
\newcommand{\Stab}{\ensuremath{\text{Stab}}\xspace}

\newcommand{\CBMat}{\ensuremath{\text{Basis}}\xspace}
 
\newcommand{\PrimeF}{\ensuremath{\text{Prim}}}

\newcommand{\FPR}{\mathrm{FPR}}
\newcommand{\LALogic}{\ensuremath{\mathrm{LA}^{\omega}}}
\newcommand{\LAkLogic}{\ensuremath{\mathrm{LA}^{k}}}
\newcommand{\iso}{\ensuremath{\cong}}

\newcommand{\nats}{\mathbb{N}}
\newcommand{\ra}{\rightarrow}

\newcommand{\lary}[1]{\ensuremath{{\ell}\text{-}{#1}}}
\newcommand{\commentout}[1]{}
 
\begin{document}

\title{Approximations of Isomorphism and Logics with Linear-Algebraic Operators}
\titlecomment{{\lsuper*}This article is an extended version  of the 
conference paper~\cite{DawarGraPak19a}.}

\author[A.~Dawar]{Anuj Dawar}	
\address{University of Cambridge, UK}
\email{anuj.dawar@cl.cam.ac.uk}  
\thanks{The research of the first author is supported by EPSRC grant EP/S03238X/1}

\author[E.~Gr\"adel]{Erich Gr\"adel}	
\address{RWTH Aachen University, Germany}
\email{graedel@logic.rwth-aachen.de}

\author[W.~Pakusa]{Wied Pakusa}
\address{RWTH Aachen University, Germany}
\email{pakusa@logic.rwth-aachen.de}

\begin{abstract}
Invertible map equivalences are approximations of graph isomorphism that refine the well-known Weisfeiler-Leman method.
They are parametrised by a number $k$ and a set $Q$ of primes. 
The intuition is that two graphs $G\IMequiv{k}{Q}H$ cannot be distinguished by a refinement of $k$-tuples given 
by linear operators acting on vector spaces over fields of characteristic $p$, for any $p \in Q$.
These equivalences have first appeared in the study of rank logic, but in fact they can be 
used to delimit the expressive power of any extension of fixed-point logic with linear-algebraic operators. We define $\LAkLogic(Q)$, an infinitary logic with $k$ variables and all linear-algebraic operators over finite vector spaces of characteristic $p \in Q$ and show
that $\IMequiv{k}{Q}$ is the natural notion of elementary equivalence for this logic.  The logic  $\LALogic(Q) = \bigcup_{k \in \omega} \LAkLogic(Q)$ is then a natural upper bound on the expressive power of any extension of fixed-point logics by means of $Q$-linear-algebraic operators.

By means of a new and much deeper algebraic analysis of a generalized variant, for any prime~$p$,
of the CFI-structures due to Cai, Fürer, and Immerman,  we prove that, as long as $Q$ is not the set of \emph{all} 
primes, there is no $k$ such that $\IMequiv{k}{Q}$ is the same as isomorphism. It follows that there are polynomial-time properties
of graphs which are not definable in $\LALogic(Q)$, which implies that no extension of fixed-point logic with linear-algebraic operators can capture PTIME, 
unless it includes such operators for all prime characteristics. Our analysis requires substantial algebraic machinery, 
including a homogeneity property of CFI-structures and Maschke's Theorem, an important result  
from the representation theory of finite groups. 
\end{abstract}

\maketitle


\section{Introduction}

The graph isomorphism problem (or more generally, the structure isomorphism problem) is an important computational problem which is also
very interesting from the point of view of complexity theory.   It is not known to be in P nor known to be NP-complete.  It is known to be solvable in quasi-polynomial time by Babai's algorithm~\cite{Babai15}.

An important theoretical approach to understanding the nature of the graph isomorphism problem is the Weisfeiler-Leman method.  For each positive integer $k$, the $k$-dimensional Weisfeiler-Leman method ($k$-WL method for short) defines an equivalence relation $\cequivk$ which over-approximates isomorphism in the sense that if $G\iso H$ for a pair of graphs $G$ and $H$, then $G \cequivk H$ for any $k$.  The relations form a refining family in the sense that if $G \not\cequivk H$ then $G \not\cequivx{k'} H$ for all $k'> k$.  Thus, the equivalence relation gets finer with increasing $k$ and approaches isomorphism in the limit.  Moreover, if $G$ and $H$ are $n$-vertex graphs then $G \cequivx{n} H$ if, and only if, $G \iso H$.  For each fixed $k$, the equivalence relation $\cequivk$ is decidable in polynomial time, indeed in time $n^{O(k)}$.  Thus, if there were a fixed $k$ such that $\cequivk$ were the same as isomorphism, we would have a polynomial-time algorithm for graph isomorphism.  However, we know this is not the case.  Cai, F\"urer and Immerman~\cite{CFI92} showed that there are pairs of non-isomorphic  graphs $G$ and $H$ with $O(k)$ vertices such that $G \cequivk H$.  We call the construction of such graphs the CFI construction.

The Weisfeiler-Leman equivalences arise naturally in the study of graphs in many different guises.  We have definitions based on combinatorics (such as Babai's original definition, see~\cite{CFI92}); in logic as the equivalences induced by bounded variable fragments of first-order logic with counting; linear programming (see~\cite{AtseriasManeva, GroheOtto}); and algebra (as in the original definition of Weisfeiler and Leman, extended to dimension $k$ in~\cite{DawarHolm}).  The equivalences have proved to be of central importance in the area of descriptive complexity theory.  In particular, they delimit the power of fixed-point logic with counting ($\FPC$), an important logic in the study of symmetric polynomial-time computation.  On many important classes of structures, it turns out that there is a fixed $k$ for which $k$-WL suffices to distinguish all non-isomorphic graphs.  Most significantly, Grohe~\cite{Grohe17} has shown that for any proper minor-closed class $\mcC$ of graphs, there is a $k$ such that $\cequivk$ coincides with isomorphism on graphs in $\mcC$.

Despite its importance in the interplay of graph structure theory and
logic, and its theoretical significance in understanding the graph
isomorphism problem, the Weisfeiler-Leman method does not give the
most efficient algorithms for solving the isomorphism problem.  The CFI construction demonstrates that using the WL method to decide isomorphism would yield an algorithm of complexity $n^{\Omega(n)}$ which is asymptotically no better than trying all permutations and far removed from the quasi-polynomial time algorithms known.  This has inspired the search for other structured families of equivalences  (see for example~\cite{Ponomarenko,Derksen}).  One particularly interesting such family are the \emph{invertible-map equivalences} defined in~\cite{DawarHol17}.  This gives, for each $k$ and each set $Q$ of prime numbers an equivalence relation $\IMequiv{k}{Q}$.  The precise definition is given in Section~\ref{sec:IMEquiv} but the intuition is that if $G \IMequiv{k}{Q} H$, then $G$ and $H$ are not distinguishable by a refinement of $k$-tuples given by linear operators acting on vector spaces over fields of characteristic $p$, for any $p \in Q$.  The reason for considering such equivalences stems from the realisation that the CFI-construction codes in graph form the problem of solving equations over $\field{F}_2$---the 2-element field (see~\cite{AtseriasBulDaw09}).  It can then be shown that the family of equivalences  $\IMequiv{k}{\{2\}}$ properly refine the Weisfeiler-Leman equivalences in that $G \IMequiv{k'}{\{2\}} H$ for sufficiently large $k'$ implies $G \cequivk H$ for all $G$ and $H$ and yet $G \not\IMequiv{3}{\{2\}} H$ for the pairs $G, H$ obtained in the CFI construction.

Furthermore, for any finite $Q$, the relation $\IMequiv{k}{Q}$ is decidable in time $n^{O(k)}$.  We can also vary $Q$ with $n$.  For instance, we could let $Q_s$ be the collection of all primes up to $s(n)$ for some growing function $s$.  In this case $\IMequiv{k}{Q_s}$ is decidable in time $s(n)n^{O(k)}$.  It is therefore an interesting question whether the family of equivalence relations is (like the Weisfeiler-Leman equivalences) infinitely refining.  Do increasing values of $k$ yield ever finer equivalence relations?  The r\^{o}le of the parameter $Q$ is also worth investigating.  If there were a fixed polynomial $s$ and constant $k$ for which $\IMequiv{k}{Q_s}$ was the same as isomorphism, we would have a polynomial-time test for isomorphism.  Even if we could prove this for $k$ growing poly-logarithmically, and $s$ quasi-polynomial, this would yield a new (and more systematic) quasi-polynomial algorithm for isomorphism.  We have no reason to conjecture that either of these upper bounds holds, but they have not been ruled out.

One reason for the interest in the invertible-map equivalences is the connection with logic.  In the long-running quest for a logic for PTIME (see~\cite{Grohe08}), an important direction is the study of extensions of fixed-point logic with rank operators (FPR)~\cite{DawarGroHolLau09} or other algebraic operators (see~\cite{DawarGraHolKopPak13}).  The relations $\IMequiv{k}{Q}$ were introduced first as a tool to study the expressive power of FPR.  It was shown in~\cite{DawarHol17} that for every formula $\phi$ of FPR (as originally defined in~\cite{DawarGroHolLau09}) there is a $k$ and a finite $Q$ such that the class of models of $\phi$ is closed under $\IMequiv{k}{Q}$.  For the more powerful rank logic $\FPR^*$ defined in~\cite{GraedelPak19}, we can show that for any formula $\phi$, there is a $k$ and a polynomial $s$ such that $\phi$ is invariant under $\IMequiv{k}{Q_s}$.  This implies, in particular that, if we could show that there is no fixed $k$ such that $\IMequiv{k}{Q}$ is the same as isomorphism when $Q$ is the set of all primes, we could, by means of padding, separate $\FPR^*$ from PTIME.  In short, any advance in understanding the structure of these equivalence relations is a significant step for resolving important questions.

The equivalence relations tell us about more than just rank logic.  They can be used to delimit the expressive power of any extension of fixed-point logic with linear-algebraic operators.  In this paper we define $\LAkLogic(Q)$, an infinitary logic with $k$ variables and all linear-algebraic operators (which we define formally below) over finite vector spaces of characteristic $p \in Q$.  This is the logic for which $\IMequiv{k}{Q}$ is the natural notion of elementary equivalence.  Then, $\LALogic(Q) = \bigcup_{k \in \omega} \LAkLogic(Q)$ is a natural upper bound on the expressive power of any extension of fixed-point logics by means of $Q$-linear-algebraic operators.

Our main results can now be stated as follows.  As long as $Q$ is not the set of all primes, there is no $k$ such that $\IMequiv{k}{Q}$ is the same as isomorphism.  From this, it follows that there are classes of graphs which are not definable in $\LALogic(Q)$.  Moreover, we can construct polynomial-time decidable such classes.   This implies that any logic with linear-algebraic operators, unless it includes such operators for all prime characteristics, does not capture PTIME.  Note, this does not separate $\FPR^*$ from PTIME, due to the restriction on $Q$, but it shows that if $\FPR^*$ is to capture PTIME, we need to use the set of all primes.

Establishing the result requires significant technical innovation.  In particular, we develop novel algebraic machinery that has not previously been deployed in the field of finite model theory.  As noted above, the CFI construction codes, in graph form, the problem of solving systems of linear equations over ${\field F}_2$.  We can give a similar construction that codes linear equations over the ${\field F}_p$ for any prime $p$.  Such a construction was given in~\cite{Holm10}, where it was used to establish that the resulting non-isomorphic graphs were not distinguished by a  variant of $\IMequiv{k}{\{q\}}$ for any $q \neq p$, where the matrix operations are restricted to a particularly simple form.  A more refined analysis of the construction was used in~\cite{GraedelPak19} to separate the expressive power of $\FPR$ from that of $\FPR^*$.  To be precise, they showed that the formulas of $\FPR$ that do not use an operator with the prime $p$ are no more expressive than formulas of $\FPC$ over these graphs.  Our result uses the same graph construction but brings significant new algebraic machinery to its analysis.

We are able to show, in this paper, that, on graphs obtained by the
CFI construction for ${\field F}_p$, the distinguishing power of
$\IMequiv{k}{Q}$, where $p \not\in Q$, is no greater than
$\cequivx{k'}$ for some fixed $k'$.  Note that the graphs are
definitely distinguished in  $\IMequiv{k}{Q}$ when $p\in Q$.  We
establish the result by showing that on these graphs, the equivalence
relation $\IMequiv{k}{\{q\}}$ is itself definable in $\FPC$ when $q
\neq p$.  This is done by implementing a matrix similarity test in
$\FPC$, based on the module isomorphism algorithm of Chistov et
al.~\cite{ChistovIK97}.  There are two key ingredients by which this
yields an $\FPC$ definition.  The first is that, on the graphs
obtained in the construction, the equivalence relation $\cequivk$ (now
understood as an equivalence relation on $k$-tuples of vertices rather
than on graphs) coincides with the partition into automorphism orbits,
for sufficiently large but constantly bounded $k$.  We say that the
graphs are $\LC^k$-homogeneous for large enough $k$.  The second
ingredient is that, because the automorphism groups of the graphs are
Abelian  $p$-groups,  this partition induces a matrix algebra over
${\field F}_q$, when $q \neq p$, which is semisimple and so admits a
nice decomposition, by Maschke's theorem.  
Maschke's
theorem, formally given as Theorem~\ref{thm:maschke} below is a
central result in the representation theory of finite groups, which
states conditions under which a linear-algebraic representation of a
finite group admits a decomposition into irreducible representations.
It is a powerful tool and we hope that its use opens the door to
further applications of representation theory in the context of finite
model theory.  Indeed, we see a major contribution of the present
work as being the introduction of Maschke's theorem and related tools
into the subject.


\section{Preliminaries}
\label{sec:preliminaries}
\label{sec:counting_logic}

We denote by $\Primes$ the set of prime numbers.
For a prime power $q$ we denote by $\field F_q$ the finite field with $q$ 
elements.
This is a paper in finite model theory and, if not stated otherwise, 
all relational structures, such a graphs, are implicitly assumed to be 
finite.
We denote relational structures by
$\mfA, \mfB, \mfC, \dots$ and we use corresponding latin letters 
$A, B, C, \dots$ to denote their universes. If $\mfA$ is a relational 
structure over the vocabulary $\tau = \{ R_1, \dots, R_k \}$, then we write 
$\mfA = (A, R_1^\mfA, \dots, R_k^\mfA)$ if $R_i^\mfA \subseteq A^{r_i}$ 
is the interpretation of relational symbol $R_i$ in $\mfA$.
We assume that the reader has a solid background in finite model theory 
and we refer to the texbooks~\cite{ebbflu99,lib04} for details.
Moreover, in order to follow our definability results in all detail, a 
good understanding of fixed-point logic with counting is necessary 
(see~\cite{Dawar15} for a survey).

\subsection*{Counting Logic}
The extension of first-order logic, denoted $\FO$, by \emph{counting quantifiers} 
$\exists^{\geq i} x \phi$, $i \geq 1$, which express the existence of at least $i$ many 
elements that satisfy $\phi$, is called \emph{counting logic} and it is 
denoted by $\LC$.
The fragments of $\FO$ and $\LC$ consisting of all formulae that contain at most 
$k$ variables (without loss of generality $x_1, \dots, x_k$) are denoted by $\LL^k$
and $\LC^k$, respectively. 
Note that $\LC$ is only a syntactic extension of $\FO$, because we can 
rewrite counting quantifiers $\exists^{\geq i} x$ using standard 
first-order quantifiers only. However, in general this translation 
will increase the number of variables. 
Hence, while $\FO \equiv \LC$ (the two logics are semantically 
equivalent), $k$-variable counting logic $\LC^k$ is strictly stronger than 
pure $k$-variable logic $\LL^k$.

\subsection*{Fixed-Point Logics}
We assume that the reader is familiar with \emph{least fixed-point logic} 
($\LFP$) and inflationary fixed-point logic ($\IFP$).
In a nutshell, \emph{fixed-point logic with counting} ($\FPC$) is the 
extension of $\IFP$ by operators for the cardinality of 
definable sets.
Formally, formulae of $\FPC$ are evaluated over the \emph{two-sorted 
extension} of an input structure $\mfA$ by a copy of the natural numbers.  
We denote by  $\mfA^{\#}$ the two-sorted extension 
of a $\tau$-structure $\mfA=(A,R_1,\dots,R_k)$ by the structure 
$\mfN=(\N,+,\cdot,0,1)$; that is $\mfA^{\#} = ( A, R_1,\dots , R_k , \N, 
+, \cdot, 0, 1)$ and the universe of the first sort (the 
\emph{vertex sort}) is $A$ and the universe of the second sort (the 
\emph{number sort} or \emph{counting sort}) is $\N$.
For both sorts, we have a collection of 
\emph{typed first-order variables}, that is the domain of any variable $x$ 
(over the input structure $\mfA$) is either $A$ or $\N$.
Similarly, for second-order variables $R$ we allow mixed types, that is a 
relation symbol $R$ of type $(k , \ell) \in \N \times \N$ stands for a 
relation $R \subseteq A^k \times \N^\ell$. 

Clearly, if we would allow unbounded first-order quantification over the 
second sort, then already $\FO$ over structures $\mfA^{\#}$ has
an undecidable model-checking problem. To obtain a logic with 
polynomial-time data 
complexity, we restrict the range of 
quantifiers over the numeric sort by fixed polynomials.
More precisely, $\FPC$-formulas can use quantifiers over the numeric sort 
only in the form $Qx \leq n^q . \phi$ where $Q \in \{ \exists, \forall 
\}$ and where $q \geq 1$ is a fixed constant.
The range of the quantifier $Q$ is $\{0, \dots, n^q\}$ where $n$ denotes 
the 
size of the input structure~$\mfA$. To simplify notation, we henceforth 
assume that each numeric variable $x$ comes with a built-in restricted 
range 
polynomial, that is $x = (x \leq n^q)$. For better readability, we 
usually omit this range polynomial in our notation.
By this convention, each variable $x$ has a 
predefined range in any input structure $\mfA^{\#}$ of polynomial size 
(which is either $A$ or $\{0, \dots, n^q\}$ for a fixed $q \geq 1$).
We denote this range by $\dom(\mfA,x)$ (or just by $\dom(x)$ if 
$\mfA$ is clear from the context). Analogously, for a tuple of variables 
$\tup x=(x_1, \dots, x_k)$ we set $\dom(\tup x) = \dom(x_1) \times \cdots 
\times \dom(x_k)$.
By this, we also obtain polynomial bounds for numeric components in 
fixed-point definitions $\left[ \ifp \, R\bx \,.\, \phi(R,\bx) 
\right](\bx)$.
Indeed, the inflationary fixed-point defined by this formula is of the 
form $R \subseteq \dom(\tup x)$. 

Crucial ingredients of $\FPC$ are \emph{counting terms} which 
allow to define cardinalities of sets. Starting with an arbitrary
$\FPC$-formula $\phi(x)$ we can form a new \emph{counting term} $s = [\# 
x : \phi]$ whose value in $\mfA$ is the size of the set defined 
by $\phi$ in $\mfA$.
In particular, the term $s$ is a \emph{numeric term}, that is $s$ takes 
its 
value in the number sort. 
One can also allow counting terms of a more general form 
without increasing the expressive power of \FPC. In particular, counting 
terms $[\# \bx : \phi]$
over mixed tuples of variables can be simulated with unary counting 
terms  and fixed-point operators; we refer to~\cite{Otto97} for more 
details and background on fixed-point logic with counting.

\subsection*{Counting Equivalence}

Let $k \geq 1$, let $\mfA$ and $\mfB$ be two structures of the same 
signature, and let $a_1, \dots, a_\ell \in A$ and $b_1, \dots, b_\ell \in 
B$ for some $0 \leq \ell \leq k$.
Then the structures $(\mfA, \tup a)$ and $(\mfB, \tup b)$ are called 
\emph{$k$-counting equivalent}, where $\tup a = (a_1, \dots, a_\ell)$ and 
$\tup b = (b_1, \dots, b_\ell)$, if for every formula $\phi(x_1, 
\dots, x_\ell) \in \LC^k$ we have $\mfA \models \phi(a_1, \dots, a_\ell)$ 
if, and only if, $\mfB \models \phi(b_1, \dots, b_\ell)$. In this case, we write $(\mfA, \tup a) \cequivk (\mfA, \tup b)$.
Obviously, for each fixed signature $\tau$ and each $\ell \leq k$, the 
relation $\cequivk$ forms an equivalence relation on the class of all 
pairs 
$(\mfC, \tup c)$ where $\mfC$ is a $\tau$-structure and where $\tup c \in 
C^\ell$ is a tuple of $\ell \leq k$ distinguished elements.
Moreover, if we fix a concrete $\tau$-structure $\mfC$, then 
$\cequivk$ induces an equivalence relation on $C^\ell$ which identifies 
$\ell$-tuples in $\mfC$ that cannot be distinguished from another by any 
$\LC^k$ formula.

A key property of the counting equivalence relation that we use is
that it is a congruence with respect to disjoint union.  So, if we
write $(\mfA,\mfB)$ for the structure that is the disjoint union of
$\mfA$ and $\mfB$, then $\mfA \cequivk \mfC$ and $\mfB \cequivk \mfD$
implies $(\mfA,\mfB) \cequivk (\mfC,\mfD)$.

\subsection*{Counting-Type Formulas}
One of the beautiful properties of the relations $\cequivk$ is 
that we can linearly order $\cequivk$-equivalence classes by means 
of a (uniform) family of $\FPC$-formulae that only use a linear number of 
variables. For technical reasons, we use a variant of these $\FPC$-formulas 
in which we can specify parameters for the equivalence relations 
$\cequivk$ (but this variant readily reduces to the standard version).

Formally, let $k \geq 1$, let $r \geq 0$, and let $\ell \leq k$.  
Fix an $r$-tuple of variables $\tup z = (z_1, \dots, z_r)$ and two 
$\ell$-tuples of variables $\tup x = (x_1, \dots, x_\ell)$ and $\tup y 
=(y_1, \dots, y_\ell)$ where all variables are pairwise distinct.
Then there exists an $\FPC$-formula
$\ctype_k[\tup z](\tup x, \tup y)$ with $\mcO(k + r)$ many variables
such that for every structure $\mfA$ and every parameter tuple $\tup c \in 
A^r$ we have that $\ctype_k[\tup c]$ defines a linear preorder $\preceq$ 
on $A^\ell$ which linearly orders the tuples in $A^\ell$ up to $k$-counting 
equivalence in the structure $(\mfA, \tup c)$ that is:
\begin{itemize}
 \item $\preceq \,\, =\,  \{ (\tup a, \tup b) \in A^\ell \times A^\ell : 
\mfA \models \ctype_k[\tup c](\tup a, \tup b) \}$ is a linear preorder on 
$A^\ell$, and
 \item for $\tup a, \tup b \in A^\ell$ we have that $\tup a \preceq 
\tup b$ and $\tup b \preceq \tup a$, that is $\tup a$ and $\tup b$ 
are incomparable (or equivalent) with respect to $\preceq$, if, and only 
if, $(\mfA, \tup c, \tup a) \cequivx{r + k} (\mfA, \tup c, \tup b)$.
\end{itemize}
In the special case where we do not have parameters, that is if $r = 0$, 
we write $\ctype_k(\tup x, \tup y)$ instead of $\ctype_k[](\tup x, \tup 
y)$. Note that for this parameter-free setting we obtain an $\FPC$-formula 
with $\mcO(k)$ many variables.
Moreover, we abuse notation and write $(\tup x \cequivk_{\tup z} \tup y)$ 
to abbreviate the formula $\ctype_k[\tup z](\tup x, \tup y) \wedge 
\ctype_k[\tup z](\tup y, \tup x)$ that is the formula which defines the 
$(r+k)$-counting equivalence with respect to the parameter tuple $\tup z$ 
of length $r \geq 0$.

Another useful fact is that for each $(\mfA, \tup c)$ and each $k$,
there is a formula $T_{\tup c}(\tup{x})$ of $\LC^k$ such that $\mfB
\models T_{\tup c}[\tup b]$ if, and only if, $(\mfB,\tup b) \cequivk (\mfA,\tup
c)$.  In particular, interpreted in $\mfA$, $T_{\tup c}$ defines
exactly the equivalence class of $\tup c$ under the relation $\cequivk$.

\subsection*{Logical Interpretations and Lindstr\"om Quantifiers}
The logical counterpart of an (algorithmic) reduction is 
the notion of a \emph{logical interpretation}. A logical interpretation $\mcI$ 
transforms an input structure $\mfA$ into a new structure $\mfB = 
\mcI(\mfA)$ 
and 
this transformation is defined by formulae of some logic $L$.
We further introduce \emph{Lindström 
quantifiers}, also known as \emph{generalised quantifiers}, which capture 
the notion of \emph{oracles} in the realm of finite model theory.

Let $\sigma, \tau$ be signatures 
with $\tau = \{S_1,...,S_{\ell}\}$, where $s_i$ denotes the arity of $S_i$. An 
$L[\sigma,\tau]$-\emph{interpretation} is a tuple
\[ I(\tup{z}) = 
(\varphi_{\delta}(\tup{x},\tup{z}),\varphi_{\approx}(\tup{x}_1,\tup{x}_2,
\tup{z}),\varphi_{S_1}(\tup{x}_1,...,\tup{x}_{s_1},\tup{z}),...,\varphi_{S_
{\ell}}(\tup{x}_1,...,\tup{x}_{s_{\ell}},\tup{z}))\]
where 
$\varphi_{\delta},\varphi_{\approx},\varphi_{S_1},...,\varphi_{S_{\ell}} 
\in L[\sigma]$ and $\tup{x},\tup{x}_1,...,\tup{x}_{s_{\ell}}$ are tuples 
of pairwise distinct variables of the same length $d$ and $\tup{z}$ is a 
tuple of variables pairwise distinct from the $x$-variables. We call $d$ 
the \emph{dimension} and $\tup{z}$ the \emph{parameters} of $\I(\tup{z})$.

A $d$-dimensional $L[\sigma,\tau]$-interpretation $\I(\tup{z})$  defines 
a partial mapping $\I \colon \Str(\sigma,\tup{z}) \rightarrow \Str(\tau)$ 
in the 
following way:  For $(\mfA,\tup{z} \mapsto \tup{a}) \in 
\Str(\sigma,\tup{z})$ we obtain a $\tau$-structure $\mfB$ over the universe 
$\{\tup{b} \in A^d \ | \ \mfA \models \varphi_{\delta}(\tup{b},\tup{a})\}$, 
setting $S_i^{\mfB} = \{(\tup{b}_1,..,\tup{b}_{s_i}) \in B^{s_i} \ | \ \mfA 
\models \varphi_{S_i}(\tup{b}_1,...,\tup{b}_{s_i},\tup{a})\}$ for each 
$S_i 
\in \tau$.
Moreover let $\mcE = \{(\tup{b}_1,\tup{b}_2) \in A^d \times A^d \ | \ \mfA 
\models \varphi_{\approx}(\tup{b}_1,\tup{b}_2,\tup{a})\}$. Now we define
\begin{equation*}
\I(\mfA,\tup{z} \mapsto \tup{a}) :=
\begin{cases}
\mfB/ \mcE &\text{if $\mcE$ is a congruence relation on } \mfB\\
\text{undefined} & \text{otherwise.}  
\end{cases}
\end{equation*}	 
We say that $\I$ interprets $\mfB/\mcE$ in $\mfA$. 

Next, we introduce Lindström quantifiers. Let $L$ be a logic and $\mcK 
\subseteq Str(\tau)$ a class of 
$\tau$-structures with $\tau = \{S_1,...,S_{\ell}\}$. The \emph{Lindström 
extension} $L(\mcQ_{\mcK})$ of $L$ by \emph{Lindström quantifiers} for 
the class $\mcK$ is obtained by extending the syntax of $L$ by the 
following formula creation rule:
\begin{quote}
	Let 
$\varphi_{\delta},\varphi_{\approx},\varphi_{S_1},...,\varphi_{S_{\ell}}$ 
be formulas in $L(\mcQ_{\mcK})$ that form an 
$L[\sigma,\tau]$-interpretation $\I(\tup{z})$. Then $\psi(\tup{z}) = 
\mcQ_{\mcK}\I(\tup{z})$
	is a formula in $L(\mcQ_{\mcK})$ over the signature $\sigma$, with
$(\mfA,\tup{z} \mapsto \tup{a}) \models \mcQ_{\mcK}\I(\tup{z})$, if, and only 
if,  $\mfB := 
\I(\mfA,\tup{z} \mapsto \tup{a})$ is defined and $\mfB \in \mcK$. 
\end{quote}  

Thus, adding the Lindström quantifier $\mcQ$ to the logic $L$ is 
the most direct way to make the class $\mcK$ definable in $L$.
Formally, if $L$ is a regular logic in the sense of~\cite{Ebb85}, then
its extension by $\mcQ$ is the minimal regular logic that can also
define $\mcK$.

\section{The Invertible Map Equivalence and Linear-Algebraic Logics}\label{sec:IMEquiv}

The invertible map equivalence relation was introduced by Dawar and
Holm~\cite{DawarHol17,Holm10} as a family of approximations of
isomorphism.  It was shown that it is at least as fine an
approximation as that induced by the infinitary logic with rank
quantifiers, introduced in~\cite{DawarGroHolLau09}.  Dawar and Holm
posed the question whether there is a logic which corresponds to the
invertible map equivalences.  In this section we answer the question by
showing that these equivalence relations are the right notions of
elementary equivalence for an infinitary logic extended with
\emph{all} linear algebraic operations.  We first review the
definition of invertible map equivalence in
Section~\ref{sec:IMEquivdef}.  We then introduce the infinitary logic,
and its various parameters, in Section~\ref{sec:LALogic}.  Finally, in
Section~\ref{sec:Equiv} we establish the relationship between the two.

\subsection{Invertible Map Equivalence}\label{sec:IMEquivdef}

We begin by defining the equivalence relations $\IMequiv{k}{Q}$ for $k
\in \nats$ and $Q$ a set of prime numbers.  To understand the
definition, it is worth reviewing the definition of the counting-logic
equivalence $\cequivx{k}$.  This is not only an equivalence relation
among finite structures, which serves as an approximation to the
isomorphism relation, it also induces a relation on the tuples in
$A^k$ for any structure $\mfA$ that serves as an approximation to the
partition into orbits of the automorphism group of $\mfA$.

On a structure $\mfA$, the relation $\cequivk$ can be obtained by an
iterative refinement process.  Suppose we are given a partition
$\mcP = \{P_i\}_{i\in I}$ of $A^k$ indexed by a set $I$.  Now, we say
that a pair of tuples ${\tup a}_1$ and ${\tup a}_2$ are
\emph{$\mcP$-similar} if they are in the same part of $\mcP$
\emph{and} for each $i \in I$ and each $j \in [k]$ the sets
$\{ b \in A \mid \tup a_1[b/j] \in P_i \}$ and
$\{ b \in A \mid \tup a_2[b/j] \in P_i \}$ have the same number of
elements.  The equivalence relation $\cequivk$ can then be
characterised as the coarsest partition $\mcP$ of $A^k$ that refines
the partition into atomic types, such that any two tuples in the same
part of $\mcP$ are $\mcP$-similar.  This means that we can arrive at
this partition by starting with the partition of $A^k$ into atomic types
and repeatedly refine it until we get a partition $\mcP$ for which the
notions of $\mcP$-equivalence and $\mcP$-similarity are the same.

We now modify this in two ways to obtain the definition of
$\IMequiv{k}{Q}$.  First we define similarity not in terms of the
substitution of  a single
element $b$ into a tuple $\tup{a} \in A^k$ but of an $\ell$-tuple
$\tup{b} \in A^{\ell}$ for some $\ell < k$.  So, for each injective function 
$\gamma: [\ell] \ra [k]$, let $\tup{a}[\tup{b}/\gamma]$ denote the tuple
in $A^k$ obtained from $\tup{a}$ by simultaneously substituting $b_i$ in position
$\gamma(i)$ for all $i \in [\ell]$.   If $\Gamma$ denotes the set of
all injective functions from $[\ell]$ to $[k]$, we say  tuples
${\tup a}_1$ and ${\tup a}_2$ are \emph{$\mcP$-similar} if they are in the
same part of $\mcP$ \emph{and} for each $\gamma \in \Gamma$ and each
$i \in I$, the sets
$\{ \tup{b} \in A^{\ell} \mid \tup a_1[\tup{b}/\gamma] \in P_i \}$ and
$\{ \tup{b} \in A^{\ell} \mid \tup a_2[\tup{b}/\gamma] \in P_i \}$
have the same size.  Taking the coarsest relation that is stable in
this sense still gives us $\equiv^k$ (though see~\cite{DawarVagnozzi}
for some nuances when comparing with the Weisfeiler-Leman equivalences).

For our purposes, we want a different notion of similarity.  Assume
that $\ell = 2m$ for some $m$.  We can then view any set $C \subseteq
A^{\ell}$ as giving us an $A^m \times A^m$ $0$-$1$ matrix, which we
denote $M$.  So the entry in row $\tup{b}_1 \in A^m$ and column
$\tup{b}_2 \in A^m$ of $M$ is $1$ if, and only if, the $\ell$-tuple
$\tup{b}_1\tup{b}_2$ is in $C$.  Hence, given, as before, a partition
$\mcP = \{P_i\}_{i\in I}$ of $A^k$, and an injective function 
$\gamma: [\ell] \ra [k]$, each tuple $\tup{a}$ induces a
partition of tuples $\tup b$ in $A^{\ell}$ according to which part $P_i$ contains
$\tup{a}[\tup{b}/\gamma]$.  We think of this as a collection
$(M^{\tup{a}}_i)_{i \in I}$ of $0$-$1$ matrices.  For a prime
number $p$, we say that two tuples ${\tup a}_1$ and ${\tup a}_2$ are
\emph{$\mcP$-$p$-$m$-similar} if they are in the same part of
$\mcP$ \emph{and} for every $\gamma$ there is an invertible matrix $S \in
\field{F}_p^{A^m \times A^m}$ such that for each type $i \in I$ we have $S M^{\tup{a}_1}_{i} S^{-1} = M^{\tup{a}_2}_{i}$.
In other words, the sequences of matrices $(M^{\tup{a}_1}_{i})_{i \in I}$ and $(M^{\tup{a}_2}_{i})_{i \in I}$ are
simultaneously similar, witnessed by $S$.  We say the tuples are
$\mcP$-$p$-similar if they are $\mcP$-$p$-$m$-similar for all $m
\leq k/2$.  The equivalence relation $\IMequiv{k}{p}$ is then the
coarsest partition $\mcP$ that refines the partition into atomic types
and such that any two tuples in the same part of $\mcP$ are
$\mcP$-$p$-similar.  Finally, for a set $Q$ of prime numbers,
$\tup{a}_1 \IMequiv{k}{Q} \tup{a}_2$ if, and only if, $\tup{a}_1
\IMequiv{k}{p} \tup{a}_2$ for each $p \in Q$.  So, $\IMequiv{k}{Q}$ is
the coarsest common refinement of the relations $(\IMequiv{k}{p})_{p
\in Q}$.

Given a fixed set $Q$ of primes with $|Q| = s$, it is possible to
compute, for a structure $\mfA$ with $n$ elements, the partition of
$A^k$ into $\IMequiv{k}{Q}$ equivalence classes in time $sn^{O(k)}$.
To see this, we note that the equivalence relation can be obtained by
an iterated refinement process.  First, let $\mcP_0$ be the partition
of $A^k$ into atomic types.  Then, for each $i$, let $\mcP_{i+1}$ be
the partition which places two tuples in the same class if, and only
if, they are $\mcP_i$-$p$-similar for all $p \in Q$.  This refinement
process converges in at most $n^k$ steps to the partition into
$\IMequiv{k}{Q}$-equivalence classes.  At 
each stage we compute, for each tuple $\tup{a} \in A^k$ and each
injective function $\gamma: [2m] \ra [n]$, the
partition of $A^{2m}$ into types, where $m = \floor{k/2}$.  This
suffices because $\mcP$-$p$-$m$-similarity implies
$\mcP$-$p$-$m'$-similarity for all $m' < m$.  Having computed the
partition, we need to check for each pair of tuples and for each $p$
in $Q$, whether the induced partitions are simultaneously similar.
For this, we use the simultaneous matrix similarity test of Chistov et
al.~\cite{ChistovIK97}.  Since this runs in polynomial time, it
follows that the whole procedure can be completed in time
$sn^{O(k)}$.

Finally, we want to make a remark about the connection with graph
isomorphism.  The partition of the tuples $A^k$ in a structure $\mfA$
into $\IMequiv{k}{Q}$ classes can be understood as approximating the
partition into orbits of the automorphism group.  Indeed, if two
tuples are in the same orbit then necessarily they are
$\IMequiv{k}{Q}$-equivalent, for all $Q$.  The relation to isomorphism
comes from the fact that computationally, the problem of partitioning
a structure into the orbits of its automorphism group and the problem
of testing a pair of structures for isomorphism are easily
inter-reducible.  For instance, given a pair of structures $\mfA$ and
$\mfB$, we define the structure $\mfA^+ \oplus \mfB^+$.  This is the
disjoint union of $\mfA^+$, the extension of $\mfA$ by a new element related by a
binary relation to every element of $\mfA$ and $\mfB^+$, a similar
extension of $\mfB$.  Then, $\mfA$ and $\mfB$ are isomorphic if, and
only if, there is some tuple of elements of $\mfA$ that is in the same
orbit as a tuple of elements of $\mfB$ in this new structure.  Hence,
any approximation of the partition into orbits of the automorphism
group gives us an approximation to the isomorphism problem.  It is in
this sense that  $\IMequiv{k}{Q}$ yields an approximation to
isomorphism.  It should be noted however that it is possible to have
a structure $\mfA$ such that for some fixed $k$ and $Q$,
$\IMequiv{k}{Q}$ does partition $A^k$ into the orbits of the
automorphism group but there is still a structure $\mfB$ that is not
isomorphic to $\mfA$ but $\IMequiv{k}{Q}$ does not distinguish between
$\mfA$ and $\mfB$.  Indeed, our key example (see
Section~\ref{sec:CFI}) has this property.

\subsection{Linear-Algebraic Logic}\label{sec:LALogic}
The study of logics with linear-algebraic operators over finite fields
was initiated in~\cite{DawarGroHolLau09}, where $\FPR$, the
fixed-point logic with rank operators, was first introduced.  As with
fixed-point logics generally, the expressive power of $\FPR$ is
naturally analysed by seeing it as a fragment of an infinitary logic, in this case
with rank \emph{quantifiers}.  The notion of elementary equivalence
that corresponds to this logic was given in terms of a game characterisation
in~\cite{DawarHol17}, where the invertible map equivalences were also
introduced.  Here, we define, for any set $Q$ of primes, 
an infinitary logic $\LALogic(Q)$ with quantifiers for \emph{all} linear-algebraic operators over finite fields
of characteristics in $Q$.
This logic is not really intended for practical use.
Instead it is designed to be strong enough so that inexpressibility results
for $\LALogic(Q)$ carry over to any well-defined logic that
extends first-order or fixed-point logic by any kind of linear-algebraic operators
over $Q$. 

We begin with a precise definition of what constitutes a
linear-algebraic operator.
Let $\field F$ be a field and let $\mcB$ be a (non-empty, finite) set that 
serves as a supply of abstract basis elements.
We consider the $\field F$-vector space $\field F^\mcB$.
For each subset $K \subseteq \mcB$ we identify the vector space $\field F^K$ 
with a subspace of $\field F^\mcB$ in the natural way: since $\field F^\mcB = 
\field F^{K} \oplus \field F^{\mcB \setminus K}$ we can (implicitly) set $\field 
F^K = \field F^{K} \oplus \{ 0 \}$.

Let $m \geq 1$. Then an \emph{$m$-ary linear-algebraic operator} $f$ is a 
function that defines a  linear-algebraic property 
$f(M_1, \dots, M_m)$ of an $m$-tuple of $\field F$-linear transformations  
$M_i$ on (subspaces of) $\field F^\mcB$.
To make things more precise, let $K_i, L_i \subseteq \mcB$, for $i \in [m]$, denote
pairs of (non-empty) subsets of basis elements.
We set $V_i = \field  F^{K_i}$ and  $W_i  = \field F^{L_i}$.
We consider $m$-tuples $(M_1, \dots, M_m)$ consisting of $\field F$-linear 
mappings $M_i\colon V_i \to W_i$ which are represented 
succinctly in terms of $m$-tuples $(M_1, \dots, M_m)$ of $L_i \times 
K_i$-matrices with entries in $\field F$.
Then an \emph{$m$-ary linear-algebraic operator} over~$\field F$ is a 
function $f$ that takes such sequences $(M_1, \dots, M_m)$ to some kind of linear-algebraic information $f(M_1, \dots, M_m)$
about the sequence.
This information is, without loss of generality, determined by a natural number 
$f(M_1, \dots, M_m) \in \N$.

Now, to say that $f$ outputs a ``linear-algebraic information'' means that the 
output   of $f$  is invariant under $\field F$-vector space isomorphisms. 
Formally, let $\mcC$ be another (abstract) set of basis elements, where $|\mcB| 
= |\mcC|$, let  $K_i', L_i'\subseteq \mcC$ where $|K_i| = |K_i'|$ and $|L_i| = 
|L_i'|$ for $i \in [m]$, and let $(N_1, \dots, N_m)$ be a sequence of 
matrices $N_i\colon L_i'\times K_i' \to  \field F$, $i \in [m]$, analogously to the 
above.
Moreover, let $V_i' = \field F^{K_i'}$ and $W_i' = \field F^{L_i'}$ 
for $i \in [m]$.
Then we say that $(N_1, \dots, N_m)$ results from $(M_1, \dots, M_m)$ by 
means of an $\field F$-vector space isomorphism if we can find an invertible 
$\field F$-linear mapping $S\colon \field  F^{\mcB} \to \field F^{\mcC}$ such 
that the following holds:
\begin{itemize}
 \item For all $i \in [m]$, $S$ maps each of the subspaces $V_i$ and 
 $W_i$ in $\field F^\mcB$ to the respective subspaces $V_i'$ and $W_i'$ in 
$\field F^{\mcC}$.
That is, if we represent  $S$ in terms of a $\mcC \times \mcB$-matrix with 
entries in $\field F$, then we have that for each of the subblocks $K_i' \times 
K_i$, $i \in [m]$, the restriction 
$S\upharpoonleft_{(K_i' \times K_i)} : K_i' \times K_i \to \field F$ of the 
matrix $S$ to this block is invertible and we have that $S(a,b) = 0$ for all $a 
\in \mcC \setminus K_i'$ and $b \in K_i$ (and the analogous holds for all 
subblocks $L_i' \times L_i$ and the corresponding restrictions
$S\upharpoonleft_{(L_i' \times L_i)} : L_i' \times L_i \to \field F$  of $S$ 
to the subblocks $L_i' \times L_i$).
 \item For each $i \in [m]$, the $\field F$-vector space isomorphism $S$ 
simultaneously transforms all linear operators $M_i \colon V_i \to W_i$ to the 
corresponding operators 
$N_i \colon V_i' \to W_i'$, that is for all $i \in [m]$ we have:
$ N_i \cdot S = S \cdot M_i$.
Note that if we want to read this as a matrix equation, then we formally have 
to replace the matrix $S$ by its restrictions to the subblocks $K_i' 
\times K_i$ and $L_i'\times L_i$ as we described above, that is 

\[ S\upharpoonleft_{(L_i' \times L_i)} \cdot M_i =  N_i \cdot 
S\upharpoonleft_{(K_i' \times K_i)}\]
\end{itemize}

Then we require that a linear algebraic operator $f$ outputs the same result 
for all  pairs of matrix sequences $(M_1, \dots, M_m)$ and $(N_1, \dots, 
N_m)$ that are related via an $\field F$-vector space isomorphism $S$ (as 
above), that is
\[ f(M_1, \dots, M_m) = f( N_1, \dots, N_m). \]
This condition guarantees that $f$ is not
able to distinguish between isomorphic objects and here, in the realm of linear 
algebra, isomorphisms are $\field F$-vector space isomorphisms. 
Besides this basic invariance condition, we do not put any kind of 
additional restrictions onto $f$.  For instance, $f$ may not even be a  
computable function.  Note that, though in introducing the function
$f$, we considered a fixed set $\mcB$, really $f$ defines, for any
$\mcB$, a function on $m$-tuples of linear operators over subspaces of
$\field{F}^{\mcB}$.   Without this, the notion of invariance would not
make sense.

Now, we can associate with $f$ a family of \emph{Lindstr\"om
quantifiers}.  For simplicity, we restrict our attention to operators
of a specific form and we explain later why this is no loss of
generality.  Specifically, we assume that $K_i = L_i = \mcB$ for all
$i$ in the above definition. In other words, $f$ is defined for a
tuple of \emph{square} matrices all with the same index set.

Let $\tau_m$ denote a vocabulary with $m$ distinct binary relations.
Given an operator that defines such an $f$ for each finite $\mcB$, for each $t \in \nats$ we define a class
of structures $\mcK^t_f$ in the vocabulary $\tau_m$.
We can think of an index set $\mcB$ with a collection $M_1,\ldots,M_m$
of $0$-$1$ $\mcB \times \mcB$ matrices as a $\tau_m$-structure
$(\mcB,M_1,\ldots,M_m)$.  The class  $\mcK^t_f$ is then the collection
of those $\tau_m$-structures where $f(M_1,\ldots,M_m)\geq t$.  For
each $\ell \geq 1$ we then have a quantifier $\mcQ^{t,\ell}_f$ such
that if $I(\tup{x})$ is an $L[\sigma,\tau_m]$-interpretation of
dimension $\ell$, then $\mcQ^{t,\ell}_f I(\tup x)$ is a
formula true in a $\sigma$ structure $\mcA$ if $I(\mcA) \in \mcK^t_f$.

The infinitary logic $\text{LA}$ is defined as the closure of
first-order logic under \emph{infinitary} disjunction and
conjunction, along with quantification $\mcQ^{t,\ell}_f$ for any linear
algebraic operator $f$ over any finite field.  That is, if $\Phi$ is
any set of formulas of $\text{LA}$, then $\bigvee \Phi$ and $\bigwedge
\Phi$ are both formulas of $\text{LA}$.  And, if $f$ is an $m$-ary
linear algebraic operator over a finite field, and $\Theta(\tup{x})$
is an $\ell$-ary $\text{LA}$-interpretation of $\sigma_m$ in $\tau$,
then $\mcQ^{t,\ell}_f \tup{x} \Theta$ is an $\text{LA}$
$\tau$-formula.  We are interested in various fragments of the logic
$\text{LA}$ for which we introduce notation in the following definition.

\begin{defi}\label{def:logic}
  $\LAkLogic$ is the collection of formulas of $\text{LA}$ that contain at
  most $k$ distinct variables.
 
  $\LALogic = \bigcup_{k \in \omega}\LAkLogic$ is the  collection of
  formulas of $\text{LA}$ containing a finite number of variables.

  For any set $Q$ of primes, we write $\text{LA}(Q)$, $\LAkLogic(Q)$
  and $\LALogic(Q)$ to denote the restrictions of these logics to
  using only linear-algebraic operators over fields of characteristic
  $p \in Q$.

  If $\mcL$ is any of the logics $\text{LA}$, $\LALogic$, $\LAkLogic$,
  $\text{LA}(Q)$, $\LALogic(Q)$ or $\LAkLogic(Q)$, and $\ell \in
  \nats$ we write
  $\lary{\mcL}$ to denote the fragment of $\mcL$ where all
  algebraic quantifiers are  $\mcQ^{t,\ell}_f$ for some $t$ and $f$.
  In other words, interpretations are restricted to be of dimension $\ell$.
\end{defi}

There are a few observations we would like to make before we go on to
analyse these logics.

The first is that, as long as $k \geq 2$, we do not need the usual
quantifiers of first-order logic.  Indeed, the formula $\exists x \phi$
is equivalent to $\mcQ^{1,1}_r xy (x=y \land \phi(x))$ where $r$ is
the unary matrix rank function.  Thus, in the inductive arguments
about the logic below, we will dispense with the case of the
existential quantifier.  More generally, the counting formula
$\exists^{\geq t}  x \phi$
is equivalent to $\mcQ^{t,1}_r xy (x=y \land \phi(x))$, so the logic
$\LAkLogic(Q)$ subsumes $\LC^k$.

The second point is that in identifying matrices with binary
relations, we have restricted ourselves to $0$-$1$-matrices.  But,
this is no loss of generality as our operators are over fixed finite
fields.  To be precise, if $f$ is an $m$-ary linear algebraic operator
over a finite field $\field{F}_q$ with $q$ elements, let $\hat{f}$ be
the $m(q-1)$-ary operator defined by 
$$ \hat{f}(M^t_i)_{i\in [m], t \in \field{F}_q\setminus\{0\}} =
f(\sum_{t\in\field{F}_q\setminus \{0\}}tM^t_i)_{i \in [m]}.$$
Then, for any $m$-tuple of matrices $(M_i)_{i\in [m]}$, the value of
$f(M_i)_{i\in[m]}$ is given by $\hat{f}(M^t_i)_{i\in [m], t \in
  \field{F}_q\setminus\{0\}}$ where $M^t_i$ is the $0$-$1$ matrix defined by
$(M^t_i)_{x,y} = 1$ if, and only if, $(M_i)_{x,y} = t$.

This has another consequence.  If $M$ is a $0$-$1$ matrix over a field
$\field{F}_q$, it is also a matrix over the prime subfield $\field
F_p$ of $\field{F}_q$, where $p$ is the characteristic of
$\field{F}_q$.  And, any linear algebraic operator over $0$-$1$
matrices over $\field{F}_q$ is completely determined by its action on
$\field{F}_p$.  For this reason, from now on, we will assume that all
linear-algebraic operators used in the logic are over prime fields.

Finally, we would like to explain why the restriction to square
matrices involves no loss of generality.  Again, this is because we
can replace an arbitrary linear-algebraic operator by one which is
defined on a tuple of square matrices all over the same index set.
Again, this involves an increase in the arity of the operator, this
time by a factor of three.

Let us start with a sequence $(M_1, \dots, M_m)$ of linear mappings $M_i\colon \field F^{K_i} \to \field F^{L_i}$ as above.
Our strategy is to encode each $M_i$ by a $3$-tuple of endomorphisms 
$(M_i^{\text{dom}},M_i^{\text{im}},M_i^{\star})$.
First of all, $M_i^{\text{dom}}\colon V \to V$ is used to encode the domain 
$\field F^{K_i}$ of $M_i$.
To this end we set
\[ M_i^{\text{dom}}(v) = v_0, \text{ where } v = v_0 \oplus v_1 \in \field F^{K_i} \oplus 
\field F^{\mcB \setminus K_i}. \]
In other words, $M_i^{\text{dom}}$ is projection of $V$ onto the subspace $\field F^{K_i}$, that is $M_i^{\text{dom}}$ the identity function on the space generated by the basis vectors in $K_i$ and it is the constant $0$ on the space generated by $\mcB \setminus K_i$. In particular, the image of $M_i^{\text{dom}}$ is 
$\im(M)=\field F^{K_i}$. 
Hence, given $M_i^{\text{dom}}$, we can easily reconstruct 
the space $\field F^{K_i}$, which corresponds to the domain of $M_i$
(in matrix representation, $M_i^{\text{dom}}$ is the identity matrix on the block
$K_i \times K_i$ and $0$ for all remaining position).

In the same way, we define a mapping $M_i^{\text{im}}$ in order to encode the image $\im(M_i) = \field F^{L_i}$ of $M_i$ via an endomorphism on $V$.
Finally, we lift $M_i$ to an endomorphism $M_i^\star\colon V \to V$ by setting
$M_i^\star = M \circ M^{\text{dom}}$. 
It is easy to translate from $M_i\colon \field F^{K_i} \to \field F^{L_i}$ 
to the encoding $(M_i^{\text{dom}},M_i^{\text{im}},M_i^{\star})$ and vice versa.
In particular, this encoding allows us to simulate a $k$-ary linear-algebraic operator $f$ by a $3k$-ary linear-algebraic operator $f'$ that only takes square $\mcB \times \mcB$-matrices as input. Indeed, a suitable operator would first decode a given $3k$-tuple 
$(M_1^{\text{dom}},M_1^{\text{im}},M_1^{\star}, \dots, M_m^{\text{dom}},M_m^{\text{im}},M_m^{\star})$  to get the original $k$-tuple $(M_1, \dots, M_m)$ and would then simulate $f$ on the input $(M_1, \dots, M_m)$.
As we claimed, this reduction shows that the set of all linear-algebraic operators has the same expressiveness as the set of all linear-algebraic operators that only accept
square matrices over the same index set.

\subsection{Relating Logic to Equivalence}\label{sec:Equiv}

Having introduced the linear algebraic logic $\LALogic$ and the
invertible-map equivalences $\IMequiv{k}{Q}$, we are now in a position
to prove that the latter is the right notion of equivalence for the
former.  Here we prove it only for equivalence within a structure,
since this is how we defined the equivalence relations.  The results
are true more generally, but this suffices for our purposes, with it
being lifted to equivalence between structures by
Lemma~\ref{lem:elem-equiv} below.

At the end of Section~\ref{sec:LALogic}, we identified three
simplifying assumptions that were made in the definition of the logic
and argued that they resulted in no loss of expressive power.  We now
make another simplifying assumption, though without restricting the
definition of the language.  We assume that  in any use of a
linear-algebraic quantifier $\mcQ^{t,\ell}_f I(\tup{x})$, the
interpretation $I$ is one without relativisation and without congruences.  This means that
the formulae $\varphi_{\delta}(\tup{x})$ and $\varphi_{\approx}(\tup{x}_1,\tup{x}_2)$ defining the
universe and the congruence relation are trivial: the former is true
of all $\ell$-tuples and the latter just defines the equality
$\tup{x}_1 = \tup{x}_2$.  To see that this involves no loss of
generality, we need to show that any use of a quantifier with an
interpretation that involves a non-trivial relativisation and
congruence can be replaced by one that does not.  So, fix an $m$-ary
linear-algebraic function $f$ and let $I(\tup{x},\tup{y}) =
(\varphi_{\delta}(\tup{x}),\varphi_{\approx}(\tup{x},\tup{y}),\varphi_1(\tup{x},\tup{y}),\ldots,\varphi_m(\tup{x},\tup{y}))$
be an $\LALogic(Q)[\sigma,\tau_m]$-interpretation.  Now, define
$\hat{f}$ to be the $(m+2)$-ary function such that
$\hat{f}(M_d,M_e,M_1,\ldots,M_m) = f(M_1,\ldots,M_m)+1$ if the
following three conditions are satisfied
\begin{enumerate}
\item $M_d$ is a $0$-$1$ matrix with non-zero entries only on the
  diagonal;
\item $M_e$ is the matrix of an equivalence relation, i.e.\ it can be
  put in block-diagonal form by a row-column permutation with each
  block being an all $1$s matrix; and
\item each of the matrices $M_1,\ldots,M_m$ is invariant under the
  equivalence relation given by $M_e$ 
\end{enumerate}
and $\hat{f}(M_d,M_e,M_1,\ldots,M_m) =0$ otherwise.  It is easily
checked that this is a linear-algebraic operator.  Now, any formula 
$\mcQ_f{\ell,t}I(\tup{x},\tup{y},\tup{z})$ is equivalent to 
$$\mcQ_{f'}{\ell,t}(x=x,\tup{x}=\tup{y},\varphi_{\delta}(\tup{x})\land
\tup{x}=\tup{y},\varphi_{\approx}(\tup{x},\tup{y}),\varphi_1(\tup{x},\tup{y}),\ldots,\varphi_m(\tup{x},\tup{y})).$$
Thus, since we only deal with interpretations without
relativisation and congruence, we will not explicitly mention the domain and congruence formulas $\varphi_{\delta}$ and
$\varphi_{\approx}$ and just write the interpretation as $(\varphi_1(\tup{x},\tup{y}),\ldots,\varphi_m(\tup{x},\tup{y}))$.

With this simplification in hand, we next proceed to establish a basic
property of the relationship between the logic $\LAkLogic(Q)$ and the
equivalence relation  $\IMequiv{k}{Q}$, namely that, in any finite
structure, this equivalence relation corresponds to the partition into
types that can be defined by formulas of the logic.  
This is
similar to the remark in Section~\ref{sec:preliminaries} to the effect
that equivalence classes with respect to $\cequivk$ are definable by
formulas of $\LC^k$.  Note however that  we do not have a counterpart to the formulas $\ctype_k$ which order
the equivalence classes.

\begin{thm}
   Let $k \geq 2$ be a positive integer and $Q$ a set of prime
   numbers.  For any finite structure $\mfA$ and $\tup{a},\tup{b} \in
   A^k$, the following are equivalent:
   \begin{enumerate}
  \item  $\tup a  \IMequiv{k}{Q}  \tup b$; and
  \item  for every formula
  $\phi$ of $\LAkLogic(Q)$, $\mfA \models \phi[\tup a]$ if, and only if, $\mfA \models \phi[\tup b]$.
   \end{enumerate}
\end{thm}
\begin{proof}
  First suppose that $\tup a  \IMequiv{k}{Q}  \tup b$ and let $\phi$
  be a formula of $\LAkLogic(Q)$.  We show by induction on the
  structure of $\phi$ that it does not distinguish the two tuples.
  Clearly if $\phi$ is an atomic formula it does not distinguish them
  by the requirement that $\IMequiv{k}{Q}$ is a refinement of the
  partition into atomic types.  The case of Boolean connectives is
  straightforward.  So, let us assume that $\phi(\tup{z})$ is
  $\mcQ_f^{t,\ell}$ for some linear-algebraic
  quantifier $\mcQ_f^{t,\ell} I(\tup{x},\tup{y})$.  Here $I$ is an
  interpretation
  $(\varphi_1(\tup{x},\tup{y}),\ldots,\varphi_m(\tup{x},\tup{y}))$ of
  dimension $\ell$, so $\tup{x}$ and $\tup{y}$ are $\ell$-tuples of
  variables and furthermore each of the formulas $\varphi_i$ may have
  parameters from $\tup{z}$.  The total number of variables is at most
  $k$ so we can assume, without loss of generality that
  $\tup{x}\tup{y}\tup{z}$ is a $k$-tuple and let $\gamma:
  [2\ell]\ra[k]$ denote the injective function that picks out the
  $\ell$-tuple $\tup{x}\tup{y}$.  Further, let $(P_j)_{j\in[t]}$
  be an enumeration of the $\IMequiv{k}{Q}$-equivalence classes.
  Then, as we noted in defining $\IMequiv{k}{Q}$, a tuple $\tup{a}$,
  along with $\gamma$ induces a partition of $A^{2\ell}$ into sets
  $P^{\tup{a},\gamma}_j = \{\tup{c} \mid \tup{a}[\tup{c}/\gamma]\in P_j\}$.  
  By induction hypothesis, each $\varphi_i$ defines a relation closed
  under $\IMequiv{k}{Q}$.  So, when $\tup{z}$ is interpreted by
  $\tup{a}$, each $\varphi_i$ defines a union of classes from among
  $(P^{\tup{a},\gamma}_j)_{j\in[t]}$.  By the assumption that  $\tup a
  \IMequiv{k}{Q}  \tup b$, we have that
  $(P^{\tup{a},\gamma}_j)_{j\in[t]}$, seen as a sequence of
  $A^{\ell}\times A^{\ell}$ matrices is simultaneously similar to
  $(P^{\tup{b},\gamma}_j)_{j\in[t]}$ over $\field F_p$ for each $p \in
  Q$.  Hence, no linear-algebraic operator can distinguish them and
  the result follows.

  In the other direction, we show that for each $\tup{a} \in A^k$ we
  can construct a formula $\Theta_{\tup a}$ that defines exactly the
  $\IMequiv{k}{Q}$-class of $\tup a$ and the result immediately
  follows.  We construct $\Theta_{\tup a}$ by induction on the iterative process of
  refinement that defines the equivalence relation $\IMequiv{k}{Q}$.
  As we noted, if $A$ has $n$ elements, there is a refining sequence of
  partitions $(\mcP_m)_{m< n^k}$ of $A^k$  that
  converges into the partition into $\IMequiv{k}{Q}$-classes, where
  $\mcP_0$ is the partition of $k$-tuples into atomic types.  We show,
  by induction on $m$,
  that for each $m$ and each part $P$ of $\mcP_m$ there is a formula
  $\Theta^m_P$ that defines exactly that part.  This is immediate for
  $\mcP_0$ as every atomic type is defined by a quantifier-free
  formula.  Now, suppose we have formulas $\Theta^m_P$ for all parts
  $P$ in the partition $\mcP_m = (P_j)_{j\in [t]}$.  Now, if $\tup a$
  and $\tup b$ are tuples that are in the same part of $\mcP_m$ but in
  distinct parts of $\mcP_{m+1}$, then there is some $\ell$,  an
  injective function $\gamma: [2\ell]\ra [k]$ and some $p\in Q$ such
  that the partitions $(P^{\tup{a},\gamma}_j)_{j\in[t]}$ and
  $(P^{\tup{b},\gamma}_j)_{j\in[t]}$ of $A^{2\ell}$ induced by $\tup
  a$ and $\tup b$ respectively are not simultaneously similar over
  $\field F_p$.  There is then some linear-algebraic function that
  distinguishes these two partitions seen as tuples of matrices.
  Indeed, we could define a $t$-ary function $f(M_1,\ldots,M_t)$ which
  is $1$ exactly when $(M_1,\ldots,M_t)$ are simultaneously similar to
  $(P^{\tup{a},\gamma}_j)_{j\in[t]}$ and $0$ otherwise.  Thus, for
  this $f$, the formula $\theta(\tup x) =
  \mcQ^{1,\ell}_f(\Theta^m_{P_1},\ldots,\Theta^m_{P_t})(\tup
  x^{\gamma})$ distinguishes $\tup a$ from $\tup b$.  Here $\tup
  x^{\gamma}$ denotes the subtuple of $k$-tuple of the variables $\tup x$
  that is picked out by $\gamma$.  Thus, we can take $\Theta^m_{[\tup
    a]}$ to be the conjunction of all formulas of this form that are
  true of $\tup a$ along with the negation of all formulas that are
  false of $\tup a$.
\end{proof}

Thus, we can treat the equivalence relation $\IMequiv{k}{Q}$, at least
in a fixed structure, as the notion of indistinguishability with
respect to the logic $\LAkLogic(Q)$.  This can be extended in the
natural way to talk of indistinguishability between structures.  So,
we use it sometimes in the form $(\mfA,\tup a) \IMequiv{k}{Q}
(\mfB,\tup b)$.  Also, by extension we allow the tuples $\tup a$ and
$\tup b$ to be of length less than $k$.  In particular, they may have
length $0$ and we can write $\mfA \IMequiv{k}{Q} \mfB$ to mean that
the two structures cannot be distinguished.  See
Lemma~\ref{lem:elem-equiv} for further treatment of this.

\section{Cai-Fürer-Immerman Structures and Logic}\label{sec:CFI}

In this section we describe a generalised variant of the \emph{CFI-construction} 
due to Cai, Fürer, and Immerman \cite{CFI92}. 
It provides a family of pairs of non-isomorphic 
graphs $(G_n, H_n)$, $n \geq 1$, such that $\Omega(n)$
many variables are required in first-order formulae that distinguish 
between $G_n$ and $H_n$, even if we allow the use of \emph{counting 
quantifiers} $\exists^{\geq i} x$.
Moreover, the construction ensures that the graphs $G_n$ and $H_n$ contain 
$\mcO(n)$-many vertices only, so that $\mcO(n)$-many variables are sufficient to
identify $G_n$ and $H_n$ up to isomorphism. Hence, the CFI-construction 
provides an optimal (linear) lower bound on the number of variables that are 
required to distinguish pairs of $n$-vertex graphs in first-order logic 
with counting (\FOC). To put it in other words, the isomorphism problem on the 
class of graphs $\{ G_n, H_n : n \geq 1 \}$ is as hard as 
possible when we measure the logical resources required for 
\FOC-definability.

On the other hand, the CFI-construction ensures that the non-isomorphic graphs 
$G_n$ and $H_n$ can quite easily be distinguished by solving a linear equation 
system over $\field F_2$.
In particular, the isomorphism problem on the family of graphs
$\{ G_n, H_n : n \geq 1 \}$ is algorithmically easy since we can solve linear equation systems over 
$\field F_2$ efficiently.
In contrast, and in addition to the lower 
bound on \FOC-definability mentioned above, we prove in this 
paper that the graphs $G_n$ and $H_n$ cannot be distinguished by any 
linear-algebraic property over any field $\field F$ of characteristic 
$\characteristic(\field F) \neq 2$. 
Hence, although linear algebra over $\field F_2$ easily separates $G_n$ from 
$H_n$, it is of no help over any field of different characteristic.

It has been observed in different contexts that the 
CFI-construction can be adapted to other algebraic structures
than the field $\field F_2$. A very general version 
due to Holm \cite{Holm10} is based on arbitrary finite Abelian groups.
For the applications in this paper it suffices to consider a less 
general version  which works over prime fields 
$\field F_p$. 
We introduce this variant here  and establish a key property 
of the automorphism group and orbits in CFI-structures that allows us to
describe the automorphism-type of $k$-tuples in counting logic by using 
$\mcO(k)$ variables only. We refer to this property as \emph{homogeneity}.

\subsection{A Generalised CFI-Construction}
\label{sec:generalised_cfi}

Our variant of the CFI-construction associates with 
every 
\begin{itemize}
 \item connected, $3$-regular, and ordered (undirected) graph 
$G=(V,E,\leq)$, and
 \item every prime field $\field F_p$, $p \in \Primes$,
\end{itemize}
a set of CFI-graphs $\CFIgraph{G}{p}{\lambda}$, where the 
role of the parameter $\lambda$ will become clear in the following.
We briefly comment on our choice of assumptions on the underlying graph $G$.
First of all, the requirement that $G$ is a connected (undirected) graph 
is standard and it guarantees that the set $\{ \CFIgraph{G}{p}{\lambda} : 
\lambda \}$ of CFI-graphs over $G$ and $\field F_p$ can be partitioned
into precisely $p$ distinct isomorphism types.
The assumption that $G$ is $3$-regular is not important for our results 
and it would be sufficient to require that the maximal degree of $G$ is 
bounded by a constant $d \geq 1$. However, assuming that each 
vertex has precisely three neighbours makes the technical presentation 
slightly simpler.
Finally, requiring that the graph $G=(V,E,\leq)$ is \emph{ordered}, that 
is that $G$ contains besides the (symmetric) edge relation $E$ also a 
linear order $\leq$ on the set of vertices $V$, is crucial for many of 
our proofs and, more specifically, in most of our definability results.
The fact that $G$ is an ordered graph ensures that no 
symmetries of the 
underlying graph $G$ carry over to the CFI-graphs 
$\CFIgraph{G}{p}{\lambda}$ and thus the only symmetries of the 
CFI-graphs result from the CFI-construction itself.
This assumption of starting with ordered graphs is crucial 
for our later definability considerations.

We now go through the construction.
Let $p \in \Primes$ be a prime. For every vector $\lambda \in \mbF_p^V$ we 
construct the \emph{CFI-structure} $\CFIgraph{G}{p}{\lambda}$ over the 
(connected, $3$-regular, and ordered) graph $G=(V, E, \leq)$, the prime field 
$\field F_p$, and 
with \emph{load} $\lambda$ as the following relational structure. The 
signature of $\CFIgraph{G}{p}{\lambda}$ is $\taucfi = \{ \preceq, R, C, I \}$ 
where $R$ is a ternary 
relation symbol and where $\preceq, I, C$ are binary relation symbols.
The universe $A$ of the CFI-structure $\mfA = \CFIgraph{G}{p}{\lambda}$ is 
$A = E \times \mbF_p$.
The linear order $\leq$ on the vertex set $V$  
extends to a linear order on the edge set $E$ (as the lexicographic order, for 
example).
We use this linear order on $E$ to define the following 
\emph{total preorder} $\preceq$ on $A$: $(e,x) \preceq (f,y)$ if $e \leq 
f$. Note that $\preceq$ induces a linear order on the corresponding 
equivalence classes $e^p = e \times \mbF_p$. Clearly, each of these 
classes $e^p$ is of size $p$. 
Since $G$ is undirected every edge $e=(v,w) \in E$ comes with its
corresponding \emph{dual edge} $f=(w,v)\in E$. In what follows, we use
the notation $e^{-1} = f$ to denote the dual of the edge $e \in E$.
The relations $I$ and $C$ are defined as follows.
\begin{itemize}
 \item The \emph{cycle relation} $C$ defines the cyclic structure 
 of the additive group of $\field  F_p$ on each of the edge classes 
$e^p$.
 More precisely, 
 \[ C = \bigcup_{e \in E} \{ ((e,x), (e, x+1 \text{ mod } p) )
 : x \in \field F_p\}.\]
 \item The \emph{inverse relation} $I$ relates additive inverses for 
dual edges. Formally,
 \[ I = \bigcup_{e \in E} \{ ( (e,x), (e^{-1}, -x) : x \in \field F_p 
\}.\]
\end{itemize}
Note that while the cycle relation $C$ defines a directed graph, the 
inverse relation $I$ is symmetric. Furthermore, observe that the 
relations $\preceq, C$ and $I$ are defined independently of the load 
vector $\lambda$ and so only depend on the underlying graph $G$ and 
the prime field $\field F_p$.
In contrast, the \emph{CFI-relation} $R=R^\lambda$ is defined using the 
load vector $\lambda$ as follows.
Since $G$ is $3$-regular we can write the set of edges outgoing from $v$
as $E(v) = \{ e_1, e_2, e_3 \}$ where $e_1 < e_2 < e_3$. The 
\emph{CFI-relation} $R^\lambda(v)$ at vertex $v$ is defined as follows:
\[ R^\lambda(v) = \{ ((e_1, x_1), (e_2, x_2), (e_3,x_3)) : x_1 + x_2 + x_3 
= \lambda(v) \,\, \mod p \}. \]
The full CFI-relation $R^\lambda$ of  
$\CFIgraph{G}{p}{\lambda}$ is given as $R^\lambda = \bigcup_{v \in 
V} R^\lambda(v)$.

\subsection{Symmetries of CFI-Structures}
\label{sec:symmetriesCFI}

The automorphism group $\Gamma$ of a CFI-structure
$\CFIgraph{G}{p}{\lambda}$ only depends on $G$ and $p$, but not on 
$\lambda$. To see this, first observe that every automorphism $\pi \in 
\Gamma$ has to maintain the linear preorder $\preceq$. This means that each 
$\pi \in \Gamma$ has to fix each edge class, that is
$\pi(e^p) = e^p$ for all $e \in E$. Moreover, $\pi$ has to maintain the 
cycle relation~$C$. This means that the action of $\pi$ on an 
edge class $e^p$ is a cyclic shift in~$\field F_p$. 
Indeed, if $\pi(e,0) = (e,x)$ for $x \in \field F_p$, then the cycle 
relation $C$ enforces that $\pi(e,i) = (e,j)$ where $j = i+x \,\,\mod p$.
Let us write $\pi(e) 
\in \field F_p$ to denote the length $x \in \field F_p$ of the cyclic shift of 
$\pi$ on $e^p$ for $e \in E$.
Then, because of the inverse relation $I$, we have $\pi(e) + \pi(e^{-1}) = 
0$ for all $\pi \in \Gamma$.
Altogether this shows that 
\[ \Gamma \leq \{ \pi \in \field F^E_p : \pi(e) + 
\pi(e^{-1}) = 0 \text{ for } e \in E \} \leq \field F^E_p.\] 
So far we have not taken the CFI-relation $R^\lambda$ into 
account.
Since $\pi(e^p) = e^p$ for all $e \in E$ it follows that $\pi(R^\lambda(v)) = 
R^\lambda(v)$ for all $v \in V$.
Let  $v \in V$ and $vE = \{ w_1, w_2, w_3 \}$ and let 
$((w_1, x_1), (w_2, x_2) , (w_3, x_3)) \in R^\lambda(v)$, that is 
$x_1 + x_2 + x_3 = \lambda(v) \,\,\mod p$.
From our earlier observations we know that
\[\pi( (w_i, x_i)) =(w_i, 
x_i+\pi(v,w_i)).\]
Hence, the condition $\pi(R^\lambda(v))=R^\lambda(v)$ implies that
\[ x_1 + \pi(v,w_1) +  x_2 + \pi(v, w_2) + x_3 + \pi(v,w_3) = \lambda(v).\]
This, in turn, means that $\pi(v,w_1) + \pi(v,w_2) + \pi(v,w_3) = 
\sum_{e \in E(v)} \pi(e) = 0$. In fact, this last condition is not only 
necessary, but also sufficient for $\pi$ to preserve the relation
$R^\lambda(v)$, as one can verify easily. 
Moreover, this condition on $\pi$ is independent of the specific 
load
vector $\lambda$. The only requirement is that, for each vertex $v$, 
the three cyclic shifts $\pi(e)$ for $e \in E(v)$ sum up to $0 \,\, \mod 
p$.
Altogether this gives us the following characterisation of the automorphism 
group $\Gamma$ of $\CFIgraph{G}{p}{\lambda}$ as a subspace of the vector space 
$\field F_p^E$ that is determined by the following set of linear equations 
in variables $\pi(e)$ for $e \in E$:
\begin{align*}
 \tag{\text{Inv}}\pi(e) + \pi(e^{-1}) &=0 && \text{ for } e \in 
E \label{aut-constr-inv}
\\
 \tag{\text{CFI}}\pi(v) := \sum_{e \in E(v)} \pi(e) 
&=0 && \text{ for 
} v \in V.\label{aut-constr-CFI}
\end{align*}

More generally, we can apply each vector $\pi \in \field F_p^E$, that
satisfies the constraints (Inv), to a CFI-structure 
$\CFIgraph{G}{p}{\lambda}$ and obtain a new CFI-structure over the same 
underlying graph~$G$.
As it turns out the resulting structure is $\CFIgraph{G}{p}{\lambda+\pi}$ 
where $(\lambda + \pi) (v) = \lambda(v) + \pi(v)$ for all $v \in V$.
Let us denote by $\Inv(\field F_p^E) \leq \field F_p^E$ the set of all 
vectors $\pi$ that satisfy the $(\Inv)$-constraints.

\begin{rem}
The group $\Delta = \Inv(\field F_p^{E}) \leq 
\field F_p^{E(G)}$ acts on the set of all CFI-structures over 
$G$ 
that is on $\CFIgraph{G}{p}{\star} := \{ \CFIgraph{G}{p}{\lambda} : \lambda 
 \in\field F_p^V\}$ (and partitions this set into $p$ orbits, as we will see 
below).
\end{rem}

Clearly, the set $\CFIgraph{G}{p}{\star}$ has size $p^n$ where $n = 
|V|$. However, if we consider this set up to isomorphisms, then it turns out 
that there are only $p$ different types of CFI-structures over a fixed graph 
$G$ \cite{CFI92, Holm10, Pakusa16}. 
To put it differently, the action of $\Inv(\field F_p^E)$ on 
$\CFIgraph{G}{p}{\star}$ has $p$ orbits.

\begin{thm}\label{thm:cfi-isomorphism}
Two CFI-structures $\CFIgraph{G}{p}{\lambda}, \CFIgraph{G}{p}{\sigma}$ 
over the same graph $G$ are isomorphic if, and only if, 
\[\sum \lambda = \sum_{v \in V} \lambda (v) = \sum_{v \in V} \sigma(v) = 
\sum \sigma.\]
\end{thm}

For technical convenience, we have introduced 
CFI-structures as relational 
structures. However, it is easy to encode them as usual (unordered) 
graphs, and, in fact, this is the way in which they were originally 
defined in \cite{CFI92} (for $p=2$). The main step is to introduce for each 
CFI-constraint $i = ((e_1, x_1), (e_2, x_2), (e_3,x_3)) \in 
R^{\lambda(v)}$, $e_i \in vE$, $x_i \in \field F_p$, a new node 
$i^{\lambda(v)}$ and to connect it to the edge nodes $(e_i,x_i) \in e_i^p$ 
accordingly
(these additional constraint nodes $i^{\lambda(v)}$ are called \emph{inner 
nodes} in 
the original 
construction in~\cite{CFI92}). 
Furthermore, we can replace the linear preorder by a path of the
appropriate length and connect vertices in the edge classes to positions 
on this path accordingly.
All of these simple transformation steps are clearly definable in \FPC.

\begin{lem}
\label{lem:encodeCFI:asgraph}
 There exist \FPC-interpretations 
 $\mcJ$ and $\mcJ^{-1}$ 
 such that $\mcJ$ maps CFI-structures $\mfA = \CFIgraph{G}{p}{\lambda} 
\in \CFIclass{\mcF}{p}$ to  graphs $\mcJ(\mfA)$ of degree $\mcO(p^2)$ and 
with $\mcO(p^2 \cdot n)$ many  vertices, where $n = |V(G)|$, and such that 
$\mcJ^{-1}$, which maps graphs to CFI-structures, is the inverse of $\mcJ$ 
in the sense that for all $\mfA \in \CFIgraph{G}{p}{\lambda}$ we have that
$\mcJ^{-1} ( \mcJ(\mfA) )$ is isomorphic to $\mfA$, that is $ \mcJ^{-1} 
( \mcJ(\mfA) ) \cong \mfA$.
\end{lem}

\subsection{CFI-Structures over Expander Graphs}
\label{sec:homogeneity}

The CFI-construction unfolds its full power when it is based on a 
family of underlying graphs that is highly connected. 
A good choice is to take $3$-regular \emph{expander 
graphs} with $\mcO(n)$ vertices, as such graphs have a linear lower bound on
the size of their separators (which means that we cannot disconnect the graphs into components of size 
$\leq n/2$ by removing fewer than $\Omega(n)$ vertices).
We briefly recall some basic facts on expander graphs 
from~\cite{HoLiWi06}.
Let $G=(V,E)$ be an undirected $d$-regular 
graph (in this paper we have $d = 3$). 
For two subsets of vertices $S, T \subseteq V$ in $G$ we denote the set of 
directed edges from $S$ to $T$ by  $E[S;T] = E \cap (S \times T)$.
The \emph{edge boundary} of a set $S \subseteq V$ is $\partial S = E[S;V 
\setminus S]$ and the \emph{expansion ratio} $h(G)$ is defined as:
\[ h(G) = \min\limits_{\{ S : |S| \leq |V|/2 \}} \frac{|\partial S|}{|S|}. 
\]

\begin{defi}[Expander graphs]
A family $\mcF = \{ G_n = (V_n, E_n) : n \geq 1\}$ of undirected 
$d$-regular 
graphs is called a \emph{family of $d$-regular expander graphs} if 
\begin{itemize}
 \item $\mcF$ is \emph{increasing}, that is $|V_n|$ is monotone and 
unbounded, and
 \item $\mcF$ is \emph{expanding}, that is there exists a constant
$\varepsilon > 0$ such that $h(G_n) \geq \varepsilon$ for all $n \geq 1$.
\end{itemize}
\end{defi}

For our applications we fix a family $\mcF$ of $3$-regular, 
connected expander graphs as provided by the following theorem.

\begin{thm}[see e.g.\ Example~2.2 in~\cite{HoLiWi06}]
There exists a family of 3-regular, connected expander graphs $\mcF = \{ 
G_n : n \in \mathbb N\}$ such that each graph $G_n$, $n \in \mathbb N$, has 
$\mcO(n)$ vertices. 
\end{thm}

Of course, we can also assume that the graphs in $\mcF$ are
\emph{ordered} just by adding to each graph $G_n=(V_n, E_n) \in \mcF$ an 
arbitrary linear order on $V_n$. 
From this family $\mcF$ of $3$-regular, connected, ordered expander graphs 
$G_n$ with $\mcO(n)$ many vertices we construct, for every $p \in 
\Primes$, the CFI-class $\CFIclass{\mcF}{p}$ consisting of 
all CFI-structures over graphs from $\mcF$ that is
\[ \CFIclass{\mcF}{p} = \bigcup_{n \in \mathbb N} 
\CFIgraph{G_n}{p}{\star}.\]

The \emph{CFI-problem} (over $\mcF$ and $p \in \Primes$) is to 
decide, given a structure 
$\CFIgraph{G}{p}{\lambda} \in \CFIclass{\mcF}{p}$ whether $\sum\lambda = 
0$. 
For the original form of the CFI-construction, it was shown in \cite{CFI92} that this problem is 
undefinable in counting logic with sublinearly many variables. Also the generalization to
more powerful variants, and in particular to our class $\CFIclass{\mcF}{p}$
is well-known.

\begin{thm}
\label{thm:cfi92}
 For any two structures $\CFIgraph{G_n}{p}{\lambda}, \CFIgraph{G_n}{p}{\sigma}  \in 
\CFIclass{\mcF}{p}$ we have
\[ \CFIgraph{G_n}{p}{\lambda} \equiv^{\Omega(n)} 
\CFIgraph{G_n}{p}{\sigma}.\]
\end{thm}
Thus, from the perspective of counting logic (with 
$\Omega(n)$ many variables) CFI-structures over the same underlying graph $G_n$ look the same although, for  load vectors $\lambda$ and $\sigma$ with $\sum \lambda \neq \sum \sigma$, we know that $\CFIgraph{G_n}{p}{\lambda}$ and $\CFIgraph{G_n}{p}{\sigma}$ are not isomorphic. 

\subsection{Homogeneity}
\label{sec:cfi:homogeneous}

We have seen that the generalised CFI-construction starts with a family $\mcF$ of 
ordered, connected, three-regular expander graphs and 
generates a family of non-isomorphic structures that are hard to 
distinguish from the perspective of counting logic.
We now discuss a further useful property of the resulting 
structures. Despite the fact that counting logic cannot determine the 
(full) isomorphism type of a CFI-structure, it turns out that it can  
control the ``automorphism types'' of $k$-tuples inside a given CFI-structure. 
That is to say that counting logic with $\mcO(k)$ many variables can distinguish between all pairs of $k$-tuples which are not related via an automorphism of the CFI-structure.
This property is known as \emph{homogeneity}.

\begin{defi}
\label{def:homogeneity}
 Let $\ell \geq 1$. We say that a structure $\mfA$ with automorphism group 
$\Gamma$ is \emph{$\ell$-homogeneous} if for all $k \geq 1$ and all 
$k$-tuples $\tup a, \tup b \in A^k$ we have that
\[ (\mfA, \tup a) \equiv^{\ell \cdot k} (\mfA, \tup b) \text{ if, and only 
if, } \Gamma(\tup a) = \Gamma(\tup b).\]
In other words, the equivalence relation $\equiv^{\ell \cdot k}$ refines $k$-tuples in $\mfA$ up to orbits. 
Moreover, we say that a class $\mcK$ of structures is \emph{homogeneous} 
if each structure $\mfA \in \mcK$ is $\ell$-homogeneous 
for some fixed constant $\ell \geq 1$.
\end{defi}

\begin{thm} \label{thm:homogeneity}
For every prime $p$, the class $\CFIclass{\mcF}{p}$ is homogeneous.
\end{thm}

This theorem has been established very recently in~\cite{GraedelGroPagPak19}, and we refer to that paper for the full proof.
To give the reader some intuition, we briefly outline the proof strategy.
Assume that a CFI-structure $\mfA$ with a distinguished $k$-tuple $\ba \in A^k$ 
of elements is given.  Consider an element $b \in A$ 
that cannot be moved by any automorphism that fixes the tuple $\ba$, that is an 
element $b \in A$ such that the stabiliser group of the tuple $\ba$ 
is contained in the stabiliser group of the element $b$, 
formally: $ \Stab(\ba) \leq \Stab(b).$
In this situation the orbit of the element $b$ is trivial (given the elements 
$\ba$) and we need to show that the element $b \in A$ itself is definable in 
counting logic, using the elements in $\ba$ as parameters, with at most $\ell 
\cdot k$ many variables (the constant $\ell \geq 1$ depends on the underlying 
class $\mcF$ of expander graphs, more precisely on the expander constant 
$\varepsilon$).
The key insight is that if the tuple $\ba$ obstructs any automorphism that
moves $b$, then in the underlying expander graph the removal of the edges 
corresponding to the elements in $\ba$ and $b$ disconnects the graph.
Because of the expansion property it follows that the edges must be connected 
to some component which is small, where small means linearly bounded in $k$ 
(the constant for the linear bound depends on the expansion constant of the 
class $\mcF$). Since the component is small, its isomorphism type can be 
described in counting logic with $\mcO(k)$ many variables and we 
conclude that $b$ is indeed definable.

Homogeneity of CFI-structures is very useful because it 
implies that counting logic  (indeed, $\FPC$) can 
order $k$-tuples up to orbits using formulas with only a 
linear  number of variables. Indeed, by the above result, the counting-type formula 
$\ctype_{\ell \cdot k}(\tup x, \tup y) \in \FPC$ (see 
Section~\ref{sec:counting_logic}) 
defines a linear preorder on $k$-tuples which distinguishes between all pairs 
of $k$-tuples in different orbits, and it uses only $\mcO(\ell \cdot k)$ 
many variables.

One key consequence of homogeneity is that on the class of CFI
structures, the relations $\cequivk$ and $\IMequiv{k}{Q}$ coincide for
$k$ above some constant threshold.  Indeed,  $\IMequiv{k}{Q}$ is
always at least as fine as  $\cequivk$ and no finer than the
equivalence given by the partition into automorphism orbits.  When the
former and the latter are the same,  $\IMequiv{k}{Q}$ must be the
same.  In particular, this means that the counting-type formulas
$\ctype_{\ell \cdot k}(\tup x, \tup y) \in \FPC$  define a pre-order
on the  $\IMequiv{k}{Q}$ equivalence classes.

\section{Background on Associative Algebra}

In this section we present the required background on the structure theory 
of semisimple algebras and modules, following the monograph~\cite{Pierce82}.
The definitions and results  are certainly well-known in the field of 
associative algebra. However, since this is a paper in finite 
model theory, some readers may appreciate a detailed 
presentation of the algebraic background. 

Let us start with the central definition of an \emph{algebra}.
Although algebras are in general defined and studied  over commutative 
rings with unity, we consider here only algebras over fields.

\begin{defi}[Algebra]
 Let $\field{F}$ be a field. An $\field{F}$-algebra $A$ is a (non-trivial) ring with unity 
that is also an $\field{F}$-vector space and which additionally satisfies the 
identity $a(xy) = (ax)y = x(ay)$ for all $a \in \field{F}$ and $x,y \in A$ 
(intuitively, we require that the $\field{F}$-scalar multiplication of the 
vector space structure and the inner multiplication of the algebra are 
compatible).
\end{defi}

By definition, we only consider \emph{associative and unital} 
algebras, that is we require the algebra to be a ring with unity.
If  one defines algebras over commutative rings $R$ instead of fields $\field{F}$, then one needs
to replace the requirement that $A$ is an $\field{F}$-vector space by the 
requirement that $A$ is an $R$-module.
However, $\field{F}$-algebras provide much more structure than 
general $R$-algebras. Most importantly, $A$ is an $\field{F}$-vector space, 
rather than only an $R$-module, which means that powerful 
linear-algebraic machinery becomes available to us.
In particular, we can speak of the dimension $\dim(A)$ of an 
$\field{F}$-algebra $A$. In this paper all algebras will be \emph{$\field{F}$-algebras of 
finite dimension}. 
Note that while the dimension describes the structure of the underlying 
$\field{F}$-vector space up to isomorphism, due the presence of 
the inner multiplication operation on the elements of $A$, the 
dimension does certainly not characterise the whole algebra $A$ up to 
isomorphism.


\begin{defi}[Group algebra]
\label{def:groupalgebra}
 Let $G$ be a finite group and let $\field{F}$ be a field.
 Then the \emph{group algebra} $\field{F}[G]$ is the $\field{F}$-algebra whose elements 
are formal sums of the form $\sum_{g \in G} r_g g$ with coefficients $r_g 
\in \field{F}$ and such that 
\begin{itemize}
 \item addition and scalar multiplication are defined component-wise, and
 \item multiplication is defined by convolution based on the group 
operation in $G$, that is for $x = \sum_{g \in G} r_g g$ and $y = \sum_{g 
\in G} s_g g$ we have 
\[ x \cdot y = \sum_{g \in G} \big( \sum_{h_1 \cdot h_2 = g} r_{h_1} 
\cdot r_{h_2} \, \big) \, g.\]
\end{itemize}
\end{defi}
We remark that this definition can be generalised to cover the case
of infinite groups $G$ and even infinite monoids $G$. However, in this 
paper we will not require this more general form of group algebras.
Note that since we assume that the group $G$ is finite, all 
group algebras $\field{F}[G]$ that we consider are finite-dimensional 
$\field{F}$-algebras.

\begin{defi}[Matrix algebra]
 Let $A$ be an $\field{F}$-algebra and let $I$ be a non-empty (finite) set.
 Then we denote by $\MatAlg{I}{A}$ the \emph{$\field{F}$-matrix algebra} which 
consists of all $(I \times I)$-matrices with entries in $A$ and for which
(matrix) addition and multiplication and scalar multiplication are defined 
in the usual way.
\end{defi}
Again, we will not need this definition in its full generality. 
In fact, we will only encounter the special case of $\field{F}$-matrix algebras 
$\MatAlg{I}{\field{F}}$ where the entries of the matrices lie in some field $\field{F}$ 
(note that each field $\field{F}$ is an $\field{F}$-algebra over itself).
Such algebras are again finite-dimensional $\field{F}$-algebras.

\subsection{Simple and Semisimple Modules}

We now go a step further 
and consider \emph{modules over algebras}.
Our  goal is to characterise the structure of semisimple modules 
over finite-dimensional algebras and to formulate \emph{Maschke's Theorem} 
which gives a sufficient condition for modules over group algebras 
to be semisimple, see~\cite[Section 2]{Pierce82}.

Before we proceed with more definitions, let us discuss the 
prototype setting for algebras and modules that we are interested in.
Let $\field{F}$ be a field and let $I$ and $J$ be two non-empty (finite) sets.
We have introduced the $\field{F}$-algebra $\MatAlg{I}{\field{F}}$ 
consisting of all $(I \times I)$-matrices with entries in $\field{F}$ above.
Now consider the set $\MatMod{I}{J}{\field{F}}$ consisting of all $(I \times 
J)$-matrices with entries in $\field{F}$. Clearly this set forms an $\field{F}$-vector 
space as well, but, in contrast to $\MatAlg{I}{\field{F}}$, the standard matrix 
multiplication operation is not defined for pairs of $(I \times 
J)$-matrices. Hence we do not obtain an $\field{F}$-algebra structure on 
$\MatMod{I}{J}{\field{F}}$, since we are missing a multiplication operation. 
However, we can clearly multiply matrices of the $\field{F}$-algebra 
$\MatAlg{I}{\field{F}}$ from the left to matrices in $\MatMod{I}{J}{\field{F}}$.
This means that the structure of $\MatMod{I}{J}{\field{F}}$ is not only that of an 
$\field{F}$-vector space, but it obtains, with the additional (left) 
multiplication by elements from the $\field{F}$-algebra $\MatAlg{I}{\field{F}}$, the 
structure of a $\MatAlg{I}{\field{F}}$-module.
The algebras and modules that we consider in this paper arise as 
subalgebras and submodules of these prototype matrix algebras and modules.
Since each $\field{F}$-algebra $A$ is also a ring with unity, the notion of an 
$A$-module coincides with the usual definition of modules over rings.
For completeness we give this definition here from the viewpoint of 
algebras.

\begin{defi}[Module]
 Let $\field{F}$ be a field and let $A$ be an $\field{F}$-algebra.
 Then a (left) $A$-module $M$ is an Abelian group $(M,+)$ together with a
 multiplication operation $A \times M \to M$ which satisfies the following 
 for $a, b \in A$ and $x, y \in M$:
 \begin{itemize}
  \item $a (x + y) = ax + ay$
  \item $(a + b) x = ax + bx$
  \item $(ab) x = a(bx)$
  \item $1 \cdot x = x$ (where $1$ is the neutral element for 
multiplication in $A$).
 \end{itemize}
\end{defi}
As there is a natural embedding of the field $\field{F}$ into the $\field{F}$-algebra $A$, 
via $x \mapsto x \cdot 1$, it follows that every $A$-module is also an 
$\field{F}$-vector space.
Note that whenever we speak of a module in this paper, we implicitly 
refer to a left module. We refrain from introducing further notions such as submodules, 
module homomorphisms, direct sums of 
modules, and so on, as these are straightforward adaptations 
of the related notions for, say, vector spaces.
We next consider the important classes of simple and, more 
generally, semisimple modules. 
\begin{defi}
 An $A$-module $M$ is \emph{simple} if every submodule $N$ of $M$ is 
trivial, i.e.\ $N = 0$ or $N = M$. 
Moreover, an $A$-module $M$ is \emph{semisimple} if it is a direct sum of 
simple modules. The corresponding notions for an $\field{F}$-algebra $A$ are 
defined by considering the algebra as an $A$-module over itself.
\end{defi}
Intuitively, a module is simple if it is a basic building block 
that cannot be refined any further. More formally, we say that an 
$A$-module $M$ is \emph{indecomposable} if whenever $M = S \oplus T$ for 
submodules $S, T$, then $S = 0$ or $T = 0$.
\begin{thm}
 A semisimple $A$-module $M$ is simple if, and only if, $M$ is 
indecomposable.
\end{thm}

A key property of semisimple modules is that submodules have 
complements. More precisely, let $M$ be an $A$-module and let $N$ be a 
submodule of $M$. Then a \emph{complement of $N$ in $M$} is a submodule 
$P$ of $M$ such that
$M = N \oplus P$, i.e.\  $M = N + P$ and $N \cap P = 0$.
As it turns out, in a semisimple module each submodule has a complement. 
If we think of vector spaces, then this should sound quite familiar. 
Indeed, also in a vector space each subspace has a complement.
However, in contrast to vector spaces, this property is not shared by 
every module. In fact, it rather leads to an alternative characterisation 
of the notion of a semisimple module.
\begin{thm}[Complements in semisimple modules]
 An $A$-module $M$ is semisimple if, and only if, every submodule of $M$ 
has a complement in $M$.
\end{thm}
Although complements in semisimple modules always exist, 
they are clearly not unique (not even in the case 
of vector spaces). 

We can now describe the structure of semisimple 
modules as follows. For a semisimple $A$-module $M$ let $S(M)$ denote a set of representatives 
for the simple submodules of $M$ (up to isomorphism).
Then $M = \bigoplus_{N \in S(M)}  N^{\alpha(N)}$ where the 
multiplicities $\alpha(N) \geq 1$ of the simple submodules $N$ are cardinal 
numbers (but since we are here only dealing with finite modules the $\alpha(N)$ are just natural numbers).
Moreover, if we consider another $A$-module $M'$ with the same set 
$S(M')=S(M)$ of representatives of simple submodules, 
then $M$ and $M'$ are isomorphic if, and only if, $M' = \bigoplus_{N \in 
S(M)} N^{\beta(N)}$ and $\alpha(N) = \beta(N)$ for all $N \in S(M)$.
Thus the multiset of simple submodules that occur in (any) decomposition 
of the module (up to isomorphism) characterises its isomorphism 
class uniquely.

\subsection{Semisimple Algebras and Maschke's Theorem}
\label{sec:maschke}
So far we considered simple and semisimple $A$-modules over $\field{F}$-algebras 
$A$. We now turn our attention to the algebras $A$ themselves.
As pointed out above, any $\field F$-algebra $A$ can naturally be considered as 
an $A$-module over itself. We follow~\cite{Pierce82} and denote this 
$A$-module by ${}_A A$. Hence, we can use 
the same terminology that we established for modules also in the realm of 
algebras.
Understanding the structure of a semisimple algebra $A$ is quite valuable. 
Most importantly, it suffices in order to understand the structure of 
\emph{any} $A$-module $M$:
\begin{thm}[Modules over semisimple algebras]
\label{thm:semisimple}
 Let $A$ be a semisimple algebra. 
 Then every $A$-module is semisimple. 
 Moreover, we can decompose the algebra~$A$, again considered as an 
$A$-module, into a  finite direct sum of (some of its) simple submodules 
$N_1 \oplus \cdots \oplus N_m$.
 It then holds that \emph{every} simple $A$-module is isomorphic to one of 
 the $A$-submodules $N_i$ of $A$.
 As a result, the number of isomorphism types of simple $A$-modules is 
 finite.
\end{thm}

The structure of semisimple algebras is characterised by 
\emph{Wedderburn's Theorem}. It states that a semisimple algebra 
can be expressed as a finite sum of matrix algebras over 
appropriate division algebras in a \emph{unique} way.
We do not need this structure theorem in our paper and the 
interested reader is referred to~\cite[Section 3.4]{Pierce82} for more 
details. Instead, the our most important tool will be \emph{Maschke's Theorem}
which tells us that semisimple algebras occur naturally in the context of 
algebras over finite groups.

\begin{thm}[Maschke]
\label{thm:maschke}
 Let $G$ be a finite group and let $\field{F}$ be a field. The 
group algebra $\field{F}[G]$ is semisimple if, and only if, the characteristic 
of $\field{F}$ does not divide the order of $G$.
\end{thm}

\section{The Simultaneous Matrix Similarity Problem}
\label{sec:simmatsim}
We argued in Section~~\ref{sec:IMEquiv} that the equivalence relation
$\IMequiv{k}{Q}$ is decidable in time $|Q|n^{O(k)}$.  This is based on
the fact that the relation can be obtained by an iterated refinement
process that takes $n^k$ steps where, at each step, we have to perform
at most $n^{2k}$ tests for simultaneous similarity over $\field{F}_p$
for each $p \in Q$.  Crucially, checking for simultaneous matrix
similarity is itself in polynomial time.  Indeed, Chistov et
al.~\cite{ChistovIK97} describe a polynomial-time algorithm that
achieves this for all $p$.

The algorithm of~\cite{ChistovIK97} works by reducing simultaneous
matrix similarity to module isomorphism and this is the reason for our
interest in semisimple algebras and modules.  As we show, the
structure of the modules of interest is particularly simple when we
are considering the CFI-structures $\CFIgraph{G}{p}{\lambda}$ and
simultaneous similarity of matrices with respect to $\field{F}_q$
where $q$ is co-prime with $p$.  In this case, we are able to show
how a module isomorphism test can be implemented in counting logic. 
Towards this end, in this section, we develop the algebraic machinery
behind the algorithm of Chistov et al.~\cite{ChistovIK97}.

\subsection{Matrix Similarity and Modules}\label{sec:modules}
Let $\field{F}$ be a field and let $K$ be a (non-empty and finite) set.
Consider two families of $K$-indexed matrices $\mcM = \{ M_k : k \in K \}$ 
and $\mcN = \{ N_k : k \in K \}$ where the matrices $M_k$ are $I \times 
I$-matrices over~$\field{F}$ and the matrices $N_k$ are $J \times J$-matrices 
over~$\field{F}$ and such that $I$ and $J$ are index sets of the same size.
For the \emph{Simultaneous Matrix Similarity Problem (over the field 
$\field{F}$)}, or \emph{\sms} for short, we ask whether there exists an
\emph{invertible} $I \times J$-matrix $S$ over $\field{F}$ such that 
\emph{simultaneously} for all $k \in K$ it holds that
$M_k S = S N_k$. 
In other words we are asking for a similarity 
transformation which simultaneously maps the matrices $A_k$ to the 
matrices $B_k$ with corresponding indices $k \in K$.
If such a matrix $S$ exists, then we say that the matrix families $\mcM$ 
and $\mcN$ are \emph{simultaneously similar} over $\field{F}$.

A small remark is in place about our choice of working with two different 
index sets $I$ and $J$. In fact, note that $I$ and $J$ need to have 
the same size, as otherwise the problem would be ill-posed. Hence, 
without changing the problem as such, we could identify the sets $I$ 
and $J$ by fixing any bijection between $I$ and $J$ beforehand. This would 
not only simplify our notation, but it would also turn the similarity 
transformation $S$ into a square matrix. The advantage of the latter would 
be that we didn't have to deal with two-sided inverses for example.
A presentation with a single index set would be more compatible 
with the (algebraic) literature as well, as in~\cite{ChistovIK97}.
However, we stick to the setting of having families of matrices with two 
different index sets $I$ and $J$. The reason is that in our finite-model 
theoretic framework, considering definability in $\FPC$, we have no  means of 
selecting a bijection between the sets $I$ and $J$. Indeed, in general 
there is no canonical, that means isomorphism invariant, 
bijection between the sets $I$ and $J$.  If we had access to 
any (non-canonical) bijection between $I$ and $J$ in our logics, this 
would trivialise most of the problems that we study in this paper.

Let us see how the \sms-problem is connected to the structure of algebras 
and modules. The following exposition is based on~\cite{ChistovIK97}.
We define the set $\HomMat \mcM \mcN$ of $I \times J$-matrices 
$X$ over $\field{F}$ which satisfy $M_k X = X N_k$ for all $k \in K$.
Note that $\HomMat \mcM \mcN$ is an $\field{F}$-vector space.
Next, we turn this vector space into a module over an $\field{F}$-algebra.
To this end, consider the set $\StabMat \mcM$ of  $I \times I$-square 
matrices $Z$ over $\field{F}$ such that $M_k Z = Z M_k$ for all $k \in K$. The set 
$\StabMat \mcM$ is called the \emph{centraliser} of the matrix family 
$\mcM$.
It is easy to verify that $\StabMat \mcM$ forms an $\field{F}$-algebra.
Moreover, by considering matrix multiplication (from the left) by 
elements from $\StabMat \mcM$, the $\field{F}$-vector space $\HomMat 
\mcM \mcN$ turns into a $\StabMat \mcM$-module indeed.
The next observation from~\cite{ChistovIK97} establishes a necessary 
condition for matrix families to be simultaneously similar.
To state the criterion we restrict ourselves to the context of matrix 
algebras, but the result remains valid in general algebras and 
modules, see~\cite{ChistovIK97}.
To state the result we first need to introduce the following notion.

\begin{defi}
 Let $A$ be an $\field{F}$-algebra and $M$ be an $A$-module. The module $M$ 
is called \emph{cyclic} if it is generated by a single element, that is 
if $Am = \{ am : a \in A\} = M$ for some $m \in M$.
\end{defi}

\begin{lem}[\cite{ChistovIK97}]
\label{lemma:suff:sms-solution}
 If $\HomMat{\mcM}{\mcN}$ contains an invertible matrix, that is if $\mcM$ 
and $\mcN$ are simultaneously similar, then $\HomMat{\mcM}{\mcN}$ is 
cyclic (as a $\StabMat{\mcM}$-module) and every generator is an invertible 
matrix.
\end{lem}
\begin{proof}
 Fix an invertible matrix $X \in \HomMat{\mcM}{\mcN}$. We show that 
$\StabMat{\mcM} \cdot X = \HomMat{\mcM}{\mcN}$.
 Let us denote by $\idmat I$ the $I \times I$-identity matrix and by 
$\idmat J$ the $J \times J$-identity matrix.
Fix an $J \times I$-matrix $X^{-1}$ such that $X X^{-1} = \idmat I$ 
and $X^{-1} X = \idmat J$.
Let $Y \in \HomMat{\mcM}{\mcN}$.
We claim that $YX^{-1} \in \StabMat{\mcM}$.
First, note that since $M_k X = X N_k$ we have $X^{-1} M_k X X^{-1} = 
X^{-1} X N_k X^{-1}$, and hence $X^{-1} M_k = N_k X^{-1}$.
Thus we have $M_k YX^{-1} = Y N_k X^{-1} = Y X^{-1} M_k$ which proves our 
claim. Hence $Y X^{-1} X = Y \in \StabMat{\mcM} \cdot X$.
Of course, if $\StabMat{\mcM} X' = \HomMat{\mcM}{\mcN}$, then $X'$ has to 
be invertible since $Z X' = X$ with $X$ being invertible requires that $Z$ 
and $X'$ are invertible (for instance, this follows from the rank 
inequality).
\end{proof}

The above result only gives a sufficient criterion for the 
existence of an invertible matrix in $\HomMat{\mcM}{\mcN}$. 
Indeed, if the module $\HomMat{\mcM}{\mcN}$ is not cyclic, then 
we know that there does not exist an invertible matrix in 
$\HomMat{\mcM}{\mcN}$.
However, if the module is cyclic, then we still have to check whether some
(or, as we know by Lemma~\ref{lemma:suff:sms-solution}, in the positive 
case, 
each) generator is an invertible matrix or not.
In the end we would like to be able to reduce the \sms-problem to 
the module isomorphism problem. The idea is that the cyclicity of a module 
is determined by its isomorphism type. Hence, if, in turn, cyclicity 
would characterise the existence of an invertible matrix, then we
would be done.
But, unfortunately, this last assertion does not hold in general.
However, luckily, for our applications to CFI-structures, it indeed turns out that the 
module $\HomMat{\mcM}{\mcN}$ can only be cyclic if it is generated by an 
invertible matrix.

To sum up, our next aim is to establish sufficient criteria that 
allow us to answer the \sms-problem purely by looking at the isomorphism 
type of $\HomMat{\mcM}{\mcN}$, specifically by considering the cyclicity 
of this module.
Before we proceed, let us explain how we can determine whether a 
module is cyclic or not for the case of a semisimple module.

\begin{lem}
\label{lemma:cyclic-iso}
 Let $A$ be a semisimple $\field{F}$-algebra and let $M$ be an $A$-module.
 Let $A_1, \dots, A_s$ be the simple $A$-submodules of ${}_A A$ and assume 
that 
 $_A{}A \approx A_1^{n_1} \oplus \cdots \oplus A_s^{n_s}$ for some $n_1, 
\dots, 
n_s \geq 1$.
 Since $A$ is semisimple, the $A$-module $M$ is semisimple and we have
 $M \approx A_1^{m_1} \oplus \cdots \oplus A_s^{m_s}$ for some $m_1, 
\dots, 
m_s \geq 0$.
The $A$-module $M$ is cyclic if, and only if, $m_i \leq n_i$ for all $1 
\leq i 
\leq s$.
\end{lem}
\begin{proof}
 That the $A$-module $M$ is cyclic means that for some $m \in M$ we have 
$Am = 
M$. This element $m \in M$ defines an $A$-module homomorphism $\phi : {}_A 
A 
\to M$ via $\phi(a) = am$. 
Such a homomorphism can only map the simple submodules $A_i$ of ${}_A A$ to 
an 
isomorphic copy in $M$ or to $0$ (this fact is known as \emph{Schur's 
Lemma}, 
see~\cite[Section 2.3]{Pierce82} for details). 
Knowing this, the result easily follows.
\end{proof}

Notably, Lemma~\ref{lemma:cyclic-iso} is a key ingredient for the 
polynomial-time algorithm for module isomorphism established 
in~\cite{ChistovIK97} as well (cf.\ the proof of Lemma~7 
of~\cite{ChistovIK97}). 
We will apply Lemma~\ref{lemma:cyclic-iso} in order to determine whether
the $\StabMat{\mcM}$-module $\HomMat{\mcM}{\mcN}$ is cyclic or not. Note 
that 
we need a crucial prerequisite in order to apply this lemma.  Indeed, 
Lemma~\ref{lemma:cyclic-iso} requires that $\StabMat{\mcM}$ is a 
semisimple algebra.

\subsection{Block Matrices}
\label{sec:blockmatrices}
Our next step is to learn more about the special kinds of modules
and matrices that arise in our intended application.
Very roughly, the matrices that we consider are linear combinations of 
``small'' matrices that have their entries only in certain 
canonical blocks (these blocks will bound the orbits under the 
action of the automorphism group).
We make precise what we mean by this later, but, for now, 
we focus on the following important consequence: similarity 
transformations between such matrix families can be chosen to have 
block-diagonal form.
This enables us to decide the \sms-problem only by looking at the 
isomorphism type of the $\StabMat{\mcM}$-module $\HomMat{\mcM}{\mcN}$ 
(see Theorem~\ref{thm:sms:char:fbg}).

A \emph{coloured index pair} $(I, J, \preceq_I, \preceq_J)$, or a \cip 
for short, consists of a pair $(I,J)$ 
of two (finite, non-empty) sets $I$ and $J$ of the same size and of two 
linear 
preorders $\preceq_I$ and $\preceq_J$ which are defined on $I$ and $J$, 
respectively.
The linear preorder $\preceq_I$ linearly orders $I$ up to equivalence 
classes of indices $i$ and $i'$ which are incomparable, that is for which 
it holds that $i \preceq i'$ and $i' \preceq i$.
We denote the ordered partition of $I$ into these equivalence classes by
$I = I_0 \preceq \cdots \preceq I_\nmo$ which are ordered by $\preceq_I$ 
as indicated. We sometimes refer to these 
equivalence classes as \emph{colour classes}.
This comes from the intuition of thinking of the set $I$ as being coloured with 
$n$ 
different colours, which we can order, and such that the elements of 
the same colour cannot be distinguished (i.e.\ elements of the same colour 
are exactly the $\preceq_I$-incomparable elements).
Of course, the same holds for $J$ and $\preceq_J$ and we denote the 
partition of 
$J$ into $\preceq_J$-equivalence classes by
$J = J_0 \preceq \cdots \preceq J_\nmo$.
The reuse of $n$ for the length of the partition of $J$ is intentional:
for $(I, J, \preceq_I, \preceq_J)$ to 
constitute a \cip we require that the number of $\preceq_I$-colour classes 
and $\preceq_J$-colour classes  is the same and that all corresponding 
colour classes $I_k$ and $J_k$, for $k <n$, have the same size.

For the rest of this section, let $(I,J,\preceq_I, \preceq_J)$ be a \cip, 
with
$I = I_0 \preceq \cdots \preceq I_\nmo$ and
$J = J_0 \preceq \cdots \preceq J_\nmo$, and let $M$ be an 
$I\times I$-matrix (with entries in some field $\field{F}$, say). 
Then $M$ is called a \emph{block matrix} if there are two colour 
classes $I_k$ and $I_\ell$ such that $M(i,i') \neq 0$ implies that $i \in 
I_k$ and $i' \in I_\ell$.
In other words, the only non-zero entries of $M$ are in the block $I_k 
\times I_\ell$.
Of course, the same notion is defined for $J \times J$-matrices as 
well.
We say that an $I\times I$-block matrix $M$ and a $J \times J$-block 
matrix $N$ are \emph{compatible} if they are defined over corresponding 
blocks, that is $M$ is non-zero only on block $I_k \times I_\ell$ and 
$N$ is non-zero only on the corresponding block $J_k \times J_\ell$.

Now, let $S$ be an $I \times J$-matrix over $\field{F}$.
We say that $S$ is a \emph{block-diagonal matrix} if
$S(i,j) \neq 0$ implies $i \in I_k$ and $j \in J_k$ for some $k < n$.
Note that by the correspondence between the colour classes $I_k$ and 
$J_k$, and by the requirement that the number of colour classes and their 
sizes coincide, it actually makes sense to call such matrices 
``block-diagonal''  (non-zero entries occur only inside the
diagonal $I_k \times J_k$-blocks, $k < n$).

\begin{defi}
\label{def:projdiag}
 Let $S$ be an $I \times J$-matrix (with entries in some field $\field{F}$).
 For $k < n$ we define $\Diag_k(S)$ to be the projection of $S$ onto the 
 $k$-th diagonal block, that is $\Diag_k(S)$ is the $I \times J$-matrix 
defined as
 \[ \Diag_k(S)(i,j) = \begin{cases}
        S(i,j)& \text{ if } i \in I_k, j \in J_k,\\
        0,& \text{otherwise.}
       \end{cases}
       \]
Moreover, we define $\Diag(S) := \Diag_0(S) + \cdots + \Diag_\nmo(S)$ to 
be the projection of $S$ onto the diagonal blocks.
\end{defi}

\begin{lem}[see also~\cite{DawarHol17}]
\label{lem:block-to-diag}
Let $M$ be an $(I \times I)$-block matrix and let $N$ be a compatible 
$(J\times J)$-block matrix (both matrices having entries in some 
field $\field{F}$).
Moreover, let $S$ be an $(I \times J)$-matrix such that
$M S = S N$. Then $M \cdot \Diag(S) = \Diag(S) \cdot N$.
\end{lem}
\begin{proof} For an illustration see Figure~\ref{fig:lem:block-to-diag}. 
Let $T = \Diag(S)$.
 Let $M$ be a matrix with non-zero entries only in block $I_k \times 
I_\ell$ and, correspondingly, let $N$ have non-zero entries only in block 
$J_k \times J_\ell$.
 We show that $M S = M T$ (and, analogously, it can be shown that $S N = T 
N$).
 Let $i \in I$ and $j \in J$.
 First, if $i \not\in I_k$, then $M S(i,j) = 0 = M T (i,j)$.
 Hence, assume that $i \in I_k$.
 If $j \not\in J_\ell$, then $MS (i,j) = 0 = SN (i,j)$.
 We have $M T (i,j) = 0$, since $M T (i,j) = \sum_{r \in I_\ell}  M(i,r) 
\cdot T(r,j)$ and $T(r,j) = 0$ for $r \in I_\ell$, $j \not\in J_\ell$ by 
 definition.
 The only case that remains is that $i \in I_k$ and $j \in J_\ell$.
 But then $MS(i,j) = \sum_{r \in I_\ell} M(i,r) S(r,j) = \sum_{r \in 
I_\ell} M(i,r) T(r,j) = MT(i,j)$ since $S(r,j) = T(r,j)$ for $r \in 
I_\ell$ and $j \in J_\ell$ by definition.
\end{proof}

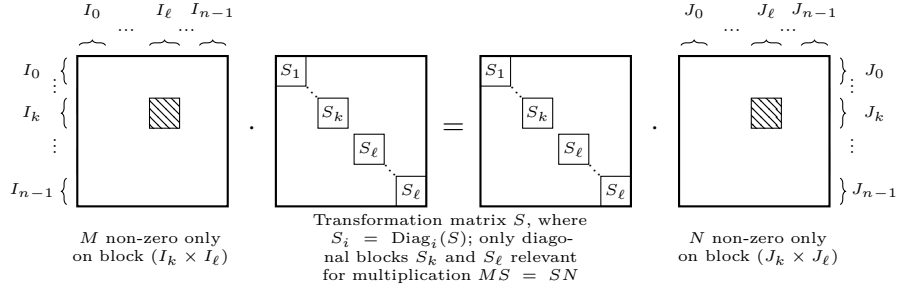
\begin{figure}[h]
\centering
\begin{tikzpicture}[scale=0.8]


\draw[step=1.5,thick] (0,-2.5) rectangle (2.5,0);
\draw[thin, pattern=north west lines] (1.2,-1.2) rectangle (1.7,-0.7);
\foreach \x in {0,...,5}
{   \foreach \y in {0,...,5}
    {   \pgfmathtruncatemacro{\nodelabel}{\x+\y*6}
        \node[inner sep=0cm] (\nodelabel) at (0.5*\x,-0.5*\y) {};
    }
}

\draw [decorate,decoration={brace,amplitude=2pt,raise=4pt}]
(6) -- (0) node [black,midway,xshift=-0.6cm,fontscale=-3] 
{$I_0$};
\draw 
[decorate,decoration={brace,amplitude=2pt,raise=4pt}]
(0,-1.2) -- (0,-0.7) node [black,midway,xshift=-0.6cm,fontscale=-3] 
{$I_k$};
\draw [decorate,decoration={brace,amplitude=2pt,raise=4pt}]
(30) -- (24) node [black,midway,xshift=-0.6cm,,fontscale=-3] 
{$I_\nmo$};
\draw (6) -- (12) node 
[midway,xshift=-0.3cm,yshift=0.2cm,rotate=90,fontscale=-8] 
{$...$};
\draw (18) -- (24) node 
[midway,xshift=-0.3cm,yshift=0.2cm,rotate=90,fontscale=-8]  
{$...$};

\draw [decorate,decoration={brace,amplitude=2pt,raise=4pt}]
(0) -- (1) node [black,midway,yshift=0.6cm,fontscale=-3] 
{$I_0$};
\draw 
[decorate,decoration={brace,amplitude=2pt,raise=4pt}]
(1.2,0) -- (1.7,0) node [black,midway,yshift=0.6cm,fontscale=-3]  
{$I_\ell$};
\draw [decorate,decoration={brace,amplitude=2pt,raise=4pt}]
(4) -- (5) node [black,midway,yshift=0.6cm,fontscale=-3]  
{$I_\nmo$};
\draw (2) -- (3) node [midway,xshift=-0.35cm,yshift=0.35cm,fontscale=-8] 
{$...$};
\draw (4) -- (5) node [midway,xshift=-0.35cm,yshift=0.35cm,fontscale=-8] 
{$...$};


\draw[step=1.5,thick] (3.3,-2.5) rectangle (5.8,0);
\foreach \x in {0,...,5}
{   \foreach \y in {0,...,5}
    {   \pgfmathtruncatemacro{\nodelabel}{\x+\y*6}
        \node[inner sep=0cm] (\nodelabel) at (0.5*\x,-0.5*\y) {};
    }
}
 \draw[thin] (3.3,-0.5) rectangle (3.8,0) node [midway,fontscale=-4] 
{$S_1$};
\draw[thin] (4,-1.2) rectangle (4.5,-0.7) node [midway,fontscale=-4] 
{$S_k$};
\draw[thin] (4.6,-1.8) rectangle (5.1,-1.3) node [midway,fontscale=-4] 
{$S_\ell$};
\draw[thin] (5.3,-2.5) rectangle (5.8,-2) node [midway,fontscale=-4] 
{$S_\ell$};
\node[fontscale=-10,rotate=135] at (3.9,-0.6) {$...$};
\node[fontscale=-10,rotate=135] at (5.2,-1.9) {$...$};

\draw[step=1.5,thick] (6.7,-2.5) rectangle (9.2,0);
\foreach \x in {0,...,5}
{   \foreach \y in {0,...,5}
    {   \pgfmathtruncatemacro{\nodelabel}{\x+\y*6}
        \node[inner sep=0cm] (\nodelabel) at (0.5*\x,-0.5*\y) {};
    }
}
 \draw[thin] (6.7,-0.5) rectangle (7.2,0) node [midway,fontscale=-4] 
{$S_1$};
\draw[thin] (7.4,-1.2) rectangle (7.9,-0.7) node [midway,fontscale=-4] 
{$S_k$};
\draw[thin] (8,-1.8) rectangle (8.5,-1.3) node [midway,fontscale=-4] 
{$S_\ell$};
\draw[thin] (8.7,-2.5) rectangle (9.2,-2) node [midway,fontscale=-4] 
{$S_\ell$};
\node[fontscale=-10,rotate=135] at (7.3,-0.6) {$...$};
\node[fontscale=-10,rotate=135] at (8.6,-1.9) {$...$};


\draw[step=1.5,thick] (10,-2.5) rectangle (12.5,0);
\draw[thin, pattern=north west lines] (11.2,-1.2) rectangle (11.7,-0.7);
\foreach \x in {0,...,5}
{   \foreach \y in {0,...,5}
    {   \pgfmathtruncatemacro{\nodelabel}{\x+\y*6}
        \node[inner sep=0cm] (\nodelabel) at (10+0.5*\x,-0.5*\y) {};
    }
}

\draw [decorate,decoration={brace,amplitude=2pt,raise=4pt,mirror}]
(11) -- (5) node [black,midway,xshift=0.6cm,fontscale=-3]  
{$J_0$};
\draw 
[decorate,decoration={brace,amplitude=2pt,raise=4pt, mirror}]
(12.5,-1.2) -- (12.5,-0.7) node [black,midway,xshift=0.6cm,fontscale=-3]  
{$J_k$};
\draw [decorate,decoration={brace,amplitude=2pt,raise=4pt,mirror}]
(35) -- (29) node [black,midway,xshift=0.6cm,fontscale=-3] 
{$J_\nmo$};
\draw (11) -- (17) node 
[midway,xshift=0.3cm,yshift=0.2cm,fontscale=-8,rotate=90]  
{$...$};
\draw (23) -- (29) node 
[midway,xshift=0.3cm,yshift=0.2cm,fontscale=-8,rotate=90] 
{$...$};

\draw [decorate,decoration={brace,amplitude=2pt,raise=4pt}]
(0) -- (1) node [black,midway,yshift=0.6cm,fontscale=-3]  
{$J_0$};
\draw 
[decorate,decoration={brace,amplitude=2pt,raise=4pt}]
(11.2,0) -- (11.7,0) node [black,midway,yshift=0.6cm,fontscale=-3] 
{$J_\ell$};
\draw [decorate,decoration={brace,amplitude=2pt,raise=4pt}]
(4) -- (5) node [black,midway,yshift=0.6cm,fontscale=-3]  
{$J_\nmo$};
\draw (2) -- (3) node 
[midway,xshift=-0.3cm,yshift=0.35cm,fontscale=-8] 
{$...$};
\draw (4) -- (5) node [midway,xshift=-0.35cm,yshift=0.35cm,fontscale=-8] 
{$...$};



\node[fontscale = -3, text width = 2.2cm, text centered] at (1.2,-3.2) 
{$M$ non-zero only on block  $(I_k \times I_\ell)$};

\node[fontscale = -3, text width = 2.2cm, text centered] at (11.3,-3.2) 
{$N$ 
non-zero only on block  $(J_k \times J_\ell)$};

\node[fontscale = -3, text width = 4.5cm, text centered] at (6.2,-3.2) 
{Transformation matrix $S$, where $S_i = \Diag_i(S)$; only diagonal blocks 
$S_k$ and $S_\ell$ relevant for multiplication $MS = SN$};

\node at (2.9,-1.2) {$\cdot$};
\node at (9.6,-1.2) {$\cdot$};
\node at (6.25,-1.2) {$=$};

\end{tikzpicture}
\caption{Illustration of Lemma~\ref{lem:block-to-diag}}
\label{fig:lem:block-to-diag}
\end{figure}

In relation to the $\sms$-problem, Lemma~\ref{lem:block-to-diag} suggests 
that for matrix families $\mcM$ and $\mcN$ that consist only of pairs of 
compatible block matrices we can restrict ourselves to similarity 
transformations that are block-diagonal. 
There is one obstacle with this approach, as, in general, we don't 
have the guarantee that the projection $D(S)$ of $S$ onto the diagonal 
blocks
preserves the rank of $S$. 
To overcome this, we add a further assumption on $\mcM$ and $\mcN$ that 
ensures that for any suitable transformation $S$ the diagonal blocks 
$D_i(S)$ have to be matrices of full rank.
Ultimately, these assumptions have useful consequences for the 
structure of the module $\HomMat{\mcM}{\mcN}$.
Before we proceed, let us formally summarise our discussion by introducing 
the notion of \emph{(faithfully) block-generated} pairs of matrix families 
$\mcM$ and $\mcN$. This concept captures the important structural 
properties of matrix families that we encounter later in our applications.
As a piece of notation, for two $K$-indexed matrix families $\mcM=\{ M_k 
: k \in K \}$ and $\mcN= \{ N_k : k \in K \}$ as above, we write 
$(\mcM \circ_K \mcN)$ to denote the $K$-synchronised direct product 
between $\mcM$ and $\mcN$, that is the $K$-indexed set consisting of pairs 
of $K$-corresponding matrices $\mcM \circ_K \mcN = \{ (M_k, N_k) : k \in K 
\}$

\begin{defi}[Faithfully block-generated]
\label{def:fblockgen}
We say that a $K$-indexed pair of matrix families $(\mcM, \mcN)$ is 
\emph{block-generated} if there is a set $B \subseteq (\mcM \circ_K \mcN)$ 
consisting of pairs of compatible block matrices that
generates $(\mcM \circ_K \mcN)$ via $\field{F}$-linear combinations.
In this case, $B$ is called a \emph{basis} of $(\mcM, \mcN)$.

Moreover,  $(\mcM, \mcN)$ is \emph{faithfully block-generated}, 
or \fbg for short, if  the set $(\mcM \circ_K \mcN)$ also contains 
all identity matrices on the diagonal blocks, that is, for every $\ell < 
n$, there exists a pair $(M, N) \in (\mcM \circ_K \mcN)$ such 
that $M$ is the identity matrix on block $I_\ell \times I_\ell$ and such 
that $N$ is the identity matrix on block 
$J_\ell \times J_\ell$ (and both matrices are zero on all remaining 
blocks).
\end{defi}

\begin{cor}
 \label{cor:fbgfam:diag}
 Let $(\mcM, \mcN)$ be an \fbg pair of matrix families.
 If $S \in \HomMat{\mcM}{\mcN}$, then $\Diag(S) \in \HomMat{\mcM}{\mcN}$.
 Moreover, if $S$ is invertible, then $\Diag(S)$ is invertible.
\end{cor}
\begin{proof}
 The first claim follows immediately from Lemma~\ref{lem:block-to-diag}.
 For the second claim assume that $S$ is invertible.
 Then we show that each of the diagonal-block matrices $\Diag_\ell(S)$ is 
invertible when considered as an $(I_\ell \times J_\ell)$-matrix, for all 
$\ell 
< 
n$.
To see this, we make use of the fact that $(\mcM, \mcN)$ is faithful.
We choose a pair $(M, N) \in \mcM \circ_K \mcN$ such that $M$ is the 
identity 
matrix on block $I_\ell \times I_\ell$ and $N$ is the identity matrix on 
the
corresponding diagonal block $J_\ell \times J_\ell$.
We have $S^{-1} M S = N$. Since $M S = \Diag_\ell(S)$, it follows that 
$S^{-1} 
\Diag_\ell(S) = N$. Hence, $\Diag_\ell(S^{-1}) \Diag_\ell(S)$ is the 
identity 
matrix on block $J_\ell \times J_\ell$, as claimed.
\end{proof}

Given the preceding result we are now in a position to restrict 
ourselves, for the case of \fbg matrix families, to block-diagonal 
transformation matrices.
Formally, let us denote by $\StabMatDiag{\mcM}$ the subalgebra of 
$\StabMat{\mcM}$ which consists of all $(I \times I)$-matrices $X \in 
\StabMat{\mcM}$ which only have non-zero entries on the diagonal $(I_\ell 
\times I_\ell)$-blocks, $\ell < n$.
Correspondingly, let us denote by $\HomMatDiag{\mcM}{\mcN}$ all matrices 
$S \in \HomMat{\mcM}{\mcN}$ which are non-zero only the diagonal blocks 
$I_\ell \times J_\ell$, $\ell < n$.
Then it is easy to see that $\HomMatDiag{\mcM}{\mcN}$ forms a 
$\StabMatDiag{\mcM}$-module.
Also note that $\StabMatDiag{\mcM} = \Diag(\StabMat{\mcM})$ and
$\HomMatDiag{\mcM}{\mcN}=\Diag(\HomMat{\mcM}{\mcN})$ for the case of \fbg 
pairs of matrix families $\mcM$ and $\mcN$, see 
Corollary~\ref{cor:fbgfam:diag}.

\begin{cor}
\label{cor:fbg:diagmod}
 Let $(\mcM, \mcN)$ be an \fbg pair of matrix families as above.
 Then $\mcM$ and $\mcN$ are simultaneously similar if, and only if, 
 $\HomMatDiag{\mcM}{\mcN}$ contains an invertible matrix.
\end{cor}

\subsection{Locally Invertible Similarity Transformations}
\label{sec:locsim}
We continue to denote by $(I, J, \preceq_I, \preceq_J)$ a \cip 
where the partitions $I = I_0 \preceq \cdots \preceq I_\nmo$ and 
$J = J_0 \preceq \cdots \preceq J_\nmo$ are given as before.
Moreover, we fix an \fbg pair $(\mcM, \mcN)$ of $K$-indexed matrix families
(as before, matrices in $\mcM$ are $I \times I$-matrices and matrices in 
$\mcN$ are $J \times J$-matrices both having entries in some common 
ground field $\field{F}$).
Our aim is to decide the \sms-problem for the pair $(\mcM, \mcN)$ only
by studying the algebraic structure of the $\StabMatDiag \mcM$-module 
$\HomMatDiag{\mcM}{\mcN}$.
As we said earlier, this is not possible in the general case, which is why 
we set out to consider a further property of $\HomMat{\mcM}{\mcN}$ that 
will enable us to follow this approach.

\begin{defi}
\label{def:locsim}
We say that $\mcM$ and $\mcN$ are \emph{locally simultaneously similar}, 
or \locsimsim for short, if for every $\ell < n$, we can find a 
matrix $S \in 
\HomMat{\mcM}{\mcN}$ such that 
$\Diag_\ell(S)$ is invertible (again, we 
consider $\Diag_\ell(S)$ as an $I_\ell \times J_\ell$-matrix).
\end{defi}

To put this definition into words, the families $\mcM$ and $\mcN$ are 
\locsimsim if we can map $\mcM$ to $\mcN$ using (possibly different) 
linear mappings which (individually) are locally, that is on each of the 
diagonal blocks $I_\ell \times J_\ell$, for $\ell < n$, invertible.
For such pairs of matrix families the algebraic structure of 
the $\StabMat{\mcM}$-module $\HomMat{\mcM}{\mcN}$ carries sufficient 
information in order to decide the \sms-problem for input $(\mcM, \mcN)$.

\begin{thm}
\label{thm:criterion:smsproblem}
 Let $(\mcM, \mcN)$ be a block-generated pair of matrix 
families $\mcM$ and $\mcN$ as above, and assume further that $\mcM$ and 
$\mcN$ are locally simultaneously similar.
 Then $\mcM$ and $\mcN$ are simultaneously similar if, and only if, 
 the $\StabMat{\mcM}$-module $\HomMat{\mcM}{\mcN}$ is cyclic.
\end{thm}
\begin{proof}
The direction from left to right was established in 
Lemma~\ref{lemma:suff:sms-solution} for the general case.
Hence, let us focus on the case that $\HomMat{\mcM}{\mcN}$ is cyclic.
We fix a generator $S \in \HomMat{\mcM}{\mcN}$, that is $\StabMat{\mcM} 
\cdot S = \HomMat{\mcM}{\mcN}$.
For $\ell < n$, by our assumption that $\mcM$ and $\mcN$ are locally 
simultaneously similar, we can find a matrix $T_\ell \in 
\HomMat{\mcM}{\mcN}$ such that $\Diag_\ell(T_\ell)$ is invertible 
(considered as an $(I_\ell \times J_\ell)$-matrix).
By Lemma~\ref{lem:block-to-diag} we know that $\Diag(T_\ell) \in 
\HomMat{\mcM}{\mcN}$ (we are using that the pair $(\mcM, \mcN)$ is 
block-generated).
Since $S$ is a generator, we can select $X_\ell \in \StabMat{\mcM}$ such 
that $X_\ell S = \Diag(T_\ell)$.

Now, let $P_\ell$ be the $I \times I$-matrix which is the identity on the 
block $I_\ell \times I_\ell$ and which is zero on all other blocks.
Then $P_\ell \cdot \Diag(T_\ell) = \Diag_\ell(T_\ell)$.
Hence $P_\ell X_\ell S = \Diag_\ell(T_\ell)$. 
We conclude that $(\sum_\ell P_\ell X_\ell) S = \sum_\ell 
\Diag_\ell(T_\ell)$. The right-hand side is a matrix of full rank, hence 
$S$ has full rank as well.
\end{proof}

This result is very useful. It says that for block-generated pairs of 
matrix families $(\mcM, \mcN)$ which are \locsimsim, the isomorphism type 
of the $\StabMat{\mcM}$-module $\HomMat{\mcM}{\mcN}$ determines whether 
$\mcM$ and $\mcN$ are simultaneously similar.
Note that in the proof of Theorem~\ref{thm:criterion:smsproblem} we did 
not require that the pair of matrix families $(\mcM, \mcN)$ is
\emph{faithfully} block-generated. If we add this assumption to our 
criterion, then we obtain a corresponding characterisation with respect to 
the algebra $\StabMatDiag{\mcM}$ and the $\StabMatDiag{\mcM}$-module 
$\HomMatDiag{\mcM}{\mcN}$ consisting of block-diagonal matrices only:
\begin{thm}[\sms-problem over \fbg pairs]
\label{thm:sms:char:fbg}
 Let $(\mcM, \mcN)$ be a \emph{faithfully} block-generated pair of matrix 
families $\mcM$ and $\mcN$ as above, and assume further that $\mcM$ and 
$\mcN$ are locally simultaneously similar.
 Then $\mcM$ and $\mcN$ are simultaneously similar if, and only if, 
 the $\StabMatDiag{\mcM}$-module $\HomMatDiag{\mcM}{\mcN}$ is cyclic.
\end{thm}
\begin{proof}
 In the light of Corollary~\ref{cor:fbg:diagmod}, it suffices to show that 
$\HomMatDiag{\mcM}{\mcN}$ contains an invertible matrix if, and only if, 
$\HomMatDiag{\mcM}{\mcN}$ is cyclic.
Again, the direction from left to right follows as in 
Lemma~\ref{lemma:suff:sms-solution} and we don't need the 
assumption of local simultaneous similarity for this direction.
For the remaining part, assume that $\StabMatDiag{\mcM} \cdot S = 
\HomMatDiag{\mcM}{\mcN}$ 
for some $S \in \HomMatDiag{\mcM}{\mcN}$.
Since $\mcM$ and $\mcN$ are \locsimsim, we can find for every $\ell < n$ a 
matrix $T_\ell \in \HomMatDiag{\mcM}{\mcN}$ such that $\Diag_\ell(T_\ell)$ 
is invertible as an $(I_\ell \times J_\ell)$-matrix.
Moroever, $X_\ell \cdot S = T_\ell$ for some matrix $X_\ell \in 
\StabMatDiag{\mcM}$ by our assumption that $S$ generates 
$\HomMatDiag{\mcM}{\mcN}$. Since $\StabMatDiag{\mcM}$ and 
$\HomMatDiag{\mcM}{\mcN}$ contain block-diagonal matrices only, it follows 
that $\Diag_\ell(X_\ell) \cdot \Diag_\ell(S) = \Diag_\ell(T_\ell)$.
This, in turn, implies that $\Diag_\ell(S)$ is invertible. Since  
$\ell < n$ was chosen arbitrarily, we can conclude that $S$ is invertible.
\end{proof}

\section{Definability of linear-algebraic operators}
\label{sec:homogeneous}

In this section, we delve deeper into the analysis of definable 
linear-algebraic operators in CFI-structures.
Specifically, we establish two key ingredients for proving our main
result in the following Section~\ref{sec:main-result}.
Our first step is to introduce an equivalence relation 
(Definition~\ref{def:linalg:isom}) between structures
that allows us to establish lower bounds for finite-variables logics with general linear-algebraic operators, cf.\ Section~\ref{sec:LALogic}.
This definition is motivated by, and strongly connected to, the definition
of the invertible map equivalence that we introduced 
in Section~\ref{sec:IMEquivdef}. We further discuss relations with
the concept of \emph{coherent configurations}.
Secondly, in Section~\ref{sec:cyclic},  we show that the solvability
problem for certain linear equation systems can be defined in counting logic in the strong functional sense, that is we can not only define the (Boolean) solvability problem, but we can even express full solution spaces of the given system in counting logic, see Theorem~\ref{thm:solvcocylic}. The specific setting for which we can
establish this definability result is that of linear equation systems over a field $\field F$ which are interpreted in
CFI-structures from a class $\CFIclass{\mcF}{p}$ where $\characteristic(\field F) \neq p$. We will make heavy use of this result in our proof of Theorem~\ref{thm:main:cfi}.

\subsection{Algebraic Structure of Equivalence Relations}
\label{sec:alg-structure}

We now want to show how the algebraic machinery that we have developed can
be used to study definability in the logic $\LALogic$.  As a first step we observe that the equivalence relations $\IMequiv{k}{Q}$ induce, in a natural way,
an $\field F$-algebra over any field $\field F$.  Indeed, this is true
of equivalence relations satisfying a natural stability
condition we elaborate below.  In particular, this is satisfied not
only by  $\IMequiv{k}{Q}$, for any $Q$ and sufficiently large $k$ but
also by the partition in a structure into automorphism orbits and
also the $\cequivk$ relations.
We begin by recalling the definition of a \emph{coherent
  configuration} (see~\cite[Chap.~3]{Cameron}).

  \begin{defi}\label{def:coherent}
  A finite set $I$ and an equivalence relation $\sim$ on $I^2$ form
  a \emph{coherent configuration} if the following three conditions
  hold for any $a,b,c,d \in I$.
  \begin{enumerate}
  \item If $(a,a) \sim (b,c)$ then $b=c$.
  \item If $(a,b) \sim (c,d)$ then $(b,a) \sim (d,c)$.
  \item If $(a,b) \sim (c,d)$ and $E$ and $E'$ are $\sim$-equivalence classes
$$|\{ e \mid (a,e) \in E \text{ and } (e,b) \in E'\}| =  |\{ e \mid (c,e) \in E \text{ and } (e,d) \in E'\}|.$$
  \end{enumerate}
\end{defi}

A coherent configuration gives rise for each field $\field F$ to an
$\field F$-algebra.  Such algebras are closely related to \emph{coherent
  algebras} in the literature (see e.g.~\cite{Friedland,Higman}).
Specifically, given a finite set $I$ and $E \subseteq I^2$, we denote
by $M_E$ the $0$-$1$ $(I \times I)$-matrix such that $(M_E)(a,b) =
1$ if, and only if, $(a,b) \in E$.

\begin{defi}
  For any finite set $I$ and an equivalence relation $\sim$ on $I^2$,
  we write $\CA{I}{\sim}{\field F}$ for the collection of matrices
  that are $\field F$-linear combinations of matrices from the set $\{
  M_E \mid E \text{ is an $\sim$-equivalence class} \}$.
\end{defi}

While we have defined this notion for any equivalence relation, the
only interesting case is when $(I,\sim)$ forms a coherent
configuration.  In this case, it can be seen that $\CA{I}{\sim}{\field
  F}$ is an $\field F$-algebra.  Indeed, it is immediate from the
definition that it is an $\field F$-vector space with the collection
of matrices $M_E$ forming a basis.  Thus, to see that it forms an
$\field F$-algebra, it suffices to show that it is closed under matrix
multiplication.  More particularly, it suffices to show that the
product of two basis matrices is itself in $\CA{I}{\sim}{\field F}$.

\begin{lem}\label{lem:coherent}
  If $(I,\sim)$ is a coherent configuration, then
  $\CA{I}{\sim}{\field F}$ is an $\field F$-algebra.
\end{lem}
\begin{proof}
  As noted above, it suffices to prove that if $C = M_E$ and $D =
  M_{E'}$ are two matrices defined from equivalence classes of $\sim$,
  then their product $CD$ is in  $\CA{I}{\sim}{\field F}$.  For this,
  it suffices to show that whenever $(a,b) \sim (c,d)$, we $CD(a,b) =
  CD(c,d)$ since this implies that $CD$ can be expressed as a linear
  combination of the matrices $M_E$.  In other words, we only need to
  show that $CD: I \times I \ra {\field F}$ is constant on each
  equivalence class $E$ of $\sim$.  But, this follows immediately from
  the definition of coherent configurations.
  \begin{align*}
    CD(a,b) & =  \sum_{e \in I} C(a,e)D(e,b) \\
            & =  |\{ e \in I \mid C(a,e) = 1 \text{ and } D(e,b) =
                  1\}| \,\, (\mod \characteristic(\field F)) \\
            & = \sum_{E \in I^2/\sim} |\{e \in E \mid C(a,e) = 1 \text{ and } D(e,b) =
                  1\}| \,\, (\mod \characteristic(\field F)) \\
            & = \sum_{E \in I^2/\sim} |\{e \in E \mid C(c,e) = 1 \text{ and } D(e,d) =
                  1\}| \,\, (\mod \characteristic(\field F)) \\
            & = CD(c,d). 
  \end{align*}
Here the second equality is from the fact that $C$ and $D$ are $0$-$1$
matrices, the third from the fact that the equivalence classes form a
partition of $I \times I$ and the fourth from the definition of a
coherent configuration.
\end{proof}

When, $(I,\sim)$ is a coherent configuration, we call
$\CA{I}{\sim}{\field F}$ its associated $\field F$ algebra.   As an example,
fix a finite structure $\mfA$ and a positive integer $\ell$.  It is
clear that the partition of $A^{\ell}$ into orbits of the automorphism
group of $\mfA$ induces a coherent configuration.  Thus, by
Lemma~\ref{lem:coherent}, we get an $\field F$-algebra.  In the case
when $\field{F}$ is the complex field, this is the centraliser algebra
of the action of the automorphism group of $\mfA$ on $A^l$
(see~\cite{Cameron}).  

Now, fix $k \geq 3\ell$ and consider the
equivalence relation $\cequivk$ on $A^{2\ell}$.  Then,
$(A^{\ell},\cequivk)$ is a coherent configuration.  Indeed, the first
two conditions in Definition~\ref{def:coherent} are easily seen to be
satisfied.  For the third, let $a,b \in A^{\ell}$.  Recall that for each
equivalence class $E \subseteq A^{2\ell}$ of $\equiv^k$ there is
a formula $T_E(\tup{x},\tup{y}) \in \LC^k$ that defines exactly the 
tuples $(a,b)\in E$ in $\mfA$.  Thus, if there are exactly $t$ tuples
$c$ such that $(a,c) \in E$ and $(c,b) \in E'$, the formula 
$$\exists^{=t} \tup{z} T_E(\tup{x},\tup{y}) \land
T_{E'}(\tup{z},\tup{x})$$
of $\LC^k$ is true of $(a,b)$ and hence of any $(c,d)$ with $ab
\cequivk cd$.  The formula is in $\LC^k$ by a standard renaming of
variables (since $k \geq 3\ell$).  As we have written the formula, it
involves a counting quantifier over $\ell$-tuples, but this can be
converted to a formula with ordinary counting
quantifiers, see~\cite{Otto97} for details.  Since
$(A^{\ell},\cequivk)$ is a coherent configuration, for any field
$\field F$, it generates an $\field F$-algebra, which we denote $\clalgebra{\mfA}{\ell;
  \LC^k}{\field F}$.  We also write $\CBMat[\mfA; \ell; \LC^k]$ for
the standard basis of the algebra, i.e.\ the collection of $0$-$1$
matrices given by the $\cequivk$-equivalence classes.  Note that we
did not specify the field $\field F$ in the notation for the basis as
the matrices are same whatever the field.

As a third example, fix a set $Q$ of prime numbers and consider the
equivalence relation $\IMequiv{k}{Q}$ defined on tuples in
$A^{2\ell}$.  Again, $(A^{\ell},\IMequiv{k}{Q})$ is a coherent
configuration by exactly the argument given above, using the fact that
counting quantifiers are expressible in the logic $\LALogic(Q)$ (see
Section~\ref{sec:IMEquiv}).  Thus, for any field $\field F$, this
defines an $\field F$-algebra which we denote $\clalgebra{\mfA}{\ell;
  \LAkLogic(Q)}{\field F}$.  Similarly, we write
$\CBMat[\mfA; \ell; \LAkLogic(Q)]$ for the standard basis of this
algebra. 

We now turn to looking at indistinguishability of a pair of
structures.  The key notion is the following.
\begin{defi}
  \label{def:linalg:isom}
 Let $\ell \geq 1$, let $k \geq 3\ell$, let $\field F$ be a field, let 
$\mfA$ and $\mfB$ be two structures and let $\mcL$ be one of the
logics $\LC^k$ or $\LAkLogic(Q)$ for some $Q$. 
Then $\mfA$ and $\mfB$ are called \emph{$\linisom{\field 
F}{\ell}{\mcL}$} if the following holds:
\begin{enumerate}
 \item $\mfA \cequivk \mfB$, and
 \item if  $M_i\colon I \times I \to \field F \in 
\CBMat[\mfA; \ell; \mcL] \subseteq \clalgebra{\mfA}{\ell; \mcL}{\field F}$ 
 and $N_i\colon J \times J \to \field F \in \CBMat[\mfB; \ell; \mcL] 
\subseteq \clalgebra{\mfB}{\ell; \mcL}{\field F}$ denote the 
corresponding $i$-th basis matrices for $i < s$, where $I = 
A^\ell$ and $J=B^\ell$ and where $s$ denotes the number of $\mcL$-equivalence
classes on $2\ell$-tuples in $\mfA$ (and $\mfB$), 
then we can find an invertible matrix $S\colon J \times I \to \field F$ 
such that 
\[ S \cdot M_i \cdot S^{-1} = N_i \text{ for all } i < s.\]
\end{enumerate}
\end{defi}

In short, $\mfA$ and $\mfB$ are called $\linisom{\field 
F}{\ell}{\mcL}$ if the $\field F$-algebras generated by the partitions of their
$2\ell$-tuples into $\mcL$-equivalence classes are isomorphic (as
algebras) \emph{and} this isomorphism is witnessed by the simultaneous
similarity of their standard bases.

Note that the requirement $\mfA$ and $\mfB$ are $\linisom{\field 
F}{\ell}{\mcL}$ means that not only are the algebras
$\clalgebra{\mfA}{\ell; \mcL}{\field F}$ and $\clalgebra{\mfB}{\ell;
  \mcL}{\field F}$ isomorphic as $\field F$ algebras, but this
isomorphism is witnessed by a simultaneous similarity transform on the
standard bases $\CBMat[\mfA; \ell; \mcL]$ and  $\CBMat[\mfB; \ell;
\mcL]$.  This is analogous to the notion of an \emph{inner
  isomorphism} for coherent algebras~\cite{Friedland}.
The main observation with regard to indistinguishability of structures
is now the following lemma.  
\begin{lem}\label{lem:elem-equiv}
  If $\mfA$ and $\mfB$ are two structures that are $\linisom{\field 
F_q}{\ell}{\LAkLogic(Q)}$ for all $q \in Q$, then they are not
distinguished by any sentence of $\lary{\LAkLogic(Q)}$.
\end{lem}
\begin{proof}
  Suppose towards a contradiction that there is a sentence of
  $\lary{\LAkLogic(Q)}$ that distinguishes $\mfA$ from $\mfB$ and let
  $\phi$ be a minimal such sentence.  We can then assume that $\phi$
  has a linear-algebraic quantifier at its head.  If it did not, it
  would be a Boolean combination of such formulas and one of them
  would distinguish $\mfA$ from $\mfB$, contradicting the minimality
  of $\phi$.  Thus, $\phi$
  is of the form $\mcQ^{t,\ell}_f\tup{x},\tup{y}(\theta_1,\ldots,\theta_m)$
  where each $\theta_i(\tup{x}\tup{y})$ defines a $2\ell$-ary relation
  and $f$ is an $\field F_p$-linear-algebraic operator for some $p \in
  Q$.  Since each $\theta_i$ defines a relation on $\mfA$ (resp.\
  $\mfB$) that is
  closed under $\IMequiv{k}{Q}$, the corresponding matrix $M_i$
  (resp.\ $N_i$) can be expressed as as a linear combination of
  matrices in $\CBMat[\mfA; \ell; \LAkLogic(Q)]$ (resp.\
  $\CBMat[\mfB; \ell; \LAkLogic(Q)]$).  Since we have an algebra
  isomorphism that takes $\CBMat[\mfA; \ell; \LAkLogic(Q)]$ to the
  corresponding matrices in $\CBMat[\mfB; \ell; \LAkLogic(Q)]$, it
  follows that $f(M_1,\ldots,M_m) = f(N_1,\ldots,N_m)$ and we derived a
  contradiction. 
\end{proof}

We conclude this section with an observation about the different
coherent configurations we have introduced along with their associated
algebras.  For any structure $\mfA$, the partition of $A^{2\ell}$ into
its automorphism orbit is the finest partition we are ever interested
in.  The other partitions, given by the equivalence relations
$\cequivk$ and $\IMequiv{k}{Q}$ for various  $k$ and $Q$ are
approximations of this.  In general, because we can define counting in
$\LAkLogic(Q)$, the partition given by $\cequivk$ is the coarsest of
them.  Thus, if for a structure $\mfA$, the partition given by
$\cequivk$ is the same as the partition into automorphism orbits, we
know that all the coherent configurations, and so all the algebras
they generate are, in fact, the same.  The structures we consider in
the remainder of this paper, i.e.\ the CFI structures of the form
$\CFIgraph{G}{p}{\lambda}$ have this property, as we discussed in
Section~\ref{sec:cfi:homogeneous}.  
Thus, we need not consider the algebras $\clalgebra{\mfA}{\ell;
  \LAkLogic(Q)}{\field F}$ explicitly.  We will confine ourselves to
describing $\clalgebra{\mfA}{\ell;
  \LC^k}{\field F}$, which turns out to be the same algebra.

\subsection{Solving Co-cyclic Linear Equation Systems}\label{sec:cyclic}

In the following, we assume some fixed encoding of 
linear equation systems as finite structures. It is an easy exercise to come up with
an appropriate representation for linear equation systems over 
finite fields and over the field of rationals (see e.g.~\cite{Holm10}). In particular, for this setting all 
natural encodings are inter-definable, which is why we refrain from defining 
an encoding explicitly.
On the other hand, linear equation systems over other (infinite) fields may not 
possess an obvious structural encoding or may not even have a finite 
representation at all. For instance, we cannot represent real numbers by finite 
means, so general linear equation system over the reals cannot be represented 
by finite structures for trivial reasons.
To avoid such problems, we will henceforth restrict to linear 
equation systems  over finite fields $\field F_{p^n}$ and over the field of rationals $\mbQ$
which in particular covers all \emph{prime fields}. As we shall see later, 
for our applications it is sufficient to 
solve linear equation systems over prime fields although larger 
fields may be present in the background.

To establish our main technical result (Theorem~\ref{thm:main:cfi})
we need that solution spaces of linear equation 
systems over a field $\field F$ are definable in counting logic if the systems 
are interpreted in (ordered pairs of) CFI-structures from a class
$\CFIclass{\mcF}{p}$ where $p \neq \characteristic(\field F)$.
This has been established in \cite{GraedelGroPagPak19} but our approach here
is somewhat different than the one in that paper. We present the precise result that we
need and a high-level sketch of the proof. For more details, we refer to 
\cite{GraedelGroPagPak19}. 
Technically the definability result depends on the following cyclicity property of CFI-structures.

\begin{defi}[Cyclic Structures]
 An \emph{$\ell$-cyclic} structure $\mfA$ is an $\ell$-homogeneous 
structure with an Abelian automorphism group.
\end{defi}

The following result concerning cyclic structures has been established in \cite{GraedelPak19}. 

\begin{thm}[Counting-Logic-Types in Cyclic Structures]
\label{thm:ctype:cyclic}
Let $k \geq 1$, let $\mfA$ denote an $\ell$-cyclic structure with (Abelian) 
automorphism group $\Gamma$, and let $\ba \in A^k$. 
Then for every $\bb, \bc \in \Gamma(\ba)$ we have 
$(\mfA, \ba, \bb) \equiv^{2\cdot k \cdot \ell} (\mfA, \ba, \bc)$ if, 
and only if, $\bb = \bc$.
Hence, the linear preorder defined by the counting-type formula $\ctype_{2\cdot 
k\cdot \ell}[\ba](\bx,\by) \in \FPC$ (see Section~\ref{sec:counting_logic}) 
defines in the structure $\mfA$ a linear order on the $\Gamma$-orbit 
$\Gamma(\ba)$ of $\ba$.
\end{thm}

\begin{defi}[Co-cyclic linear equation systems]
 A linear equation system  $M \cdot \bx = \bb$ over a prime field~$\field 
F$  is called \emph{co-cyclic} if it is
represented by some $\ell$-cyclic structure with automorphism 
group $\Gamma$ whose order is co-prime with the characteristic of $\field F$.
\end{defi}

\begin{thm}[Solvability of co-cyclic linear equation systems]
\label{thm:solvcocylic}
For every $\ell \geq 1$ there exist formulae of counting logic $\LC^\omega$ 
(actually of $\FPC$) with at most $\mcO(\ell)$ many variables which, 
given a co-cyclic linear equation system 
$M \cdot \bx = \bb$ over a prime field~$\field 
F$, for a coefficient matrix $M\colon I \times J \to \field F$ and a 
vector $\bc\colon I \to \field F$,  
define whether the system is solvable. 
Moreover, in the case that the system is solvable,  the formulae also 
define a solution $\bc \colon 
J \to \field F$ and a $J \times (J \times 
|J|)$-matrix $K$  such that $\im(K) = \ker(M)$.
\end{thm}

\commentout{
Technically, the definability result relies on the following two 
properties of CFI-structures $\mfA \in \CFIclass{\mcF}{p}$:
\begin{enumerate}[(i)]
 \item $\mfA$ is $\ell$-homogeneous (for a constant $\ell \geq 1$), and
 \item the automorphism group $\Gamma = \Aut(\mfA)$ of $\mfA$ is an 
\emph{Abelian} $p$-group (recall that $\characteristic(\field F) \neq p$).
\end{enumerate}
We need to rephrase these properties for the setting of linear equation 
systems in order to obtain our definability result.
But before we can do this we need to agree on an encoding of 
linear 
equation systems as finite structures. It is an easy exercise to come up with
an appropriate representation for linear equation systems over 
finite fields and over the field of rationals (and, indeed, these are the only 
cases that we considered in~\cite{GraedelGroPagPak19}). In particular, for this setting all 
natural encodings will be interdefinable, which is why we refrain from defining 
an encoding explicitly.
On the other hand, linear equation systems over other (infinite) fields may not 
possess an obvious structural encoding or may not even have a finite 
representation at all. For instance, we cannot represent real numbers by finite 
means, so general linear equation system over the reals cannot be represented 
by finite structures for trivial reasons.
To avoid such problems, we will henceforth restrict to linear 
equation systems 
over finite fields $\field F_{p^n}$ and over the field of rationals $\mbQ$.
Note that these fields cover the class of all \emph{prime fields}: recall that 
a finite prime 
field is of the form $\field F_p$, for $p \in \Primes$, and that the only 
infinite prime field is $\mbQ$.
Moreover, as we will see later, for our applications it will be sufficient to 
solve linear equation systems over prime fields although larger 
fields may be present in the background.

\begin{thm}[Solvability of co-cyclic linear equation systems]
\label{thm:solvcocylic}
For every $\ell \geq 1$ there exist formulas of counting logic $\LC^\omega$ 
(actually of $\FPC$) with at most $\mcO(\ell)$ many variables that express the 
following.
Given a linear equation system $M \cdot \bx = \bb$ over a prime field~$\field 
F$, for a coefficient matrix $M\colon I \times J \to \field F$ and a constants 
vector $\bc\colon I \to \field F$, which is represented by a finite structure 
$\mfA$ with the following properties:
\begin{enumerate}[(i)]
 \item $\mfA$ is $\ell$-homogeneous, \label{itm:cc1} and
 \item the automorphism group $\Gamma = \Aut(\mfA)$ of $\mfA$ is an 
\emph{Abelian} $p$-group where $\characteristic(\field F) \neq p$, 
\label{itm:cc2}
\end{enumerate}
the formulas define whether the system $M\cdot \bx = \bb$ is solvable. 
Moreover, in case that the system is solvable, then the formulas also 
define a solution $\bc \colon 
J \to \field F$ to the system $M \cdot \bx = \bb$ and an $J \times (J \times 
|J|)$-matrix $K$ with entries in $\field F$ such that $\im(K) = \ker(M)$.
\end{thm}

We refer to linear equation systems $M \cdot \bx = \bb$ over prime fields 
$\field F$ that are encoded as finite structures $\mfA$ which satisfy the
properties  (\ref{itm:cc1}) and (\ref{itm:cc2}) from above as \emph{cocyclic 
linear equation systems}.

}

\begin{proof}
The proof consists of two 
parts. 
First, we show that a single solution of a (solvable) co-cyclic linear 
equation system is definable in $\FPC$. 
Secondly, we use this result in order to determine a generating set for the 
kernel of the given coefficient matrix $M$. These two results together yield a 
(succinct) representation of the solution space of the given linear equation 
system. 

For the first step we make use of a central idea from~\cite{GraedelPak19} where we 
showed that each solvable co-cyclic linear equation system has a 
\emph{symmetric} 
solution, that is a solution which is fixed by any automorphism of the 
underlying structure $\mfA$.
More precisely, let $M \cdot \bx = \bb$ denote a linear equation system 
over a prime field~$\field F$ encoded by a finite structure $\mfA$ and let 
$\Gamma$ denote the automorphism group of the structure $\mfA$ 
which is an \emph{Abelian} $p$-group with $p \in \Primes$ and $p \neq 
\characteristic(\field F)$.
Then, because of the fact that the coefficient matrix $M$ and the constants 
vector $\bb$ are encoded in $\mfA$, they clearly must be invariant under the 
action of the automorphism group $\Gamma$.
If we write the elements $\pi \in \Gamma$ as permutation matrices~$\Pi$, then 
this translates into saying that for all $\pi \in \Gamma$ we have 
$\Pi \cdot M \cdot  \Pi^{-1} = M$ and $\Pi \cdot  \bb = \bb$.

Let us assume that there exists a solution $\bc$ of the system $M \cdot 
\bx = \bb$, that is $M \cdot \bc = \bb$.
Then, for every $\pi \in \Gamma$ we have $\Pi \cdot M \cdot  \bc = \bb$, which, 
in turn, implies that $M \cdot  \Pi \cdot  \bc = \bb$.
Hence, the solution space of $M \cdot \bx = \bb$ is closed under the action of 
$\Gamma$.
We now make use of the fact that $p \neq \characteristic(\field F)$.
From the above it follows that $(\sum_{\pi \in \Gamma} \Pi) \cdot M 
\cdot \bc = |\Gamma| \cdot \bb$, and thus 
\[ M \cdot (\sum_{\pi \in \Gamma} \frac{1}{|\Gamma|} \cdot \Pi \cdot \bc) =  
\bb.\]
Note that we used $p \neq \characteristic(\field F)$ in the above equation 
when we divided by $|\Gamma|$ (which is a power of $p$).
The new solution $d = (\sum_{\pi \in \Gamma} \frac{1}{|\Gamma|} \cdot \Pi \cdot 
\bc)$ has the remarkable property that it is \emph{symmetric}, that is for 
every $\pi \in \Gamma$ we have $\Pi \cdot d = d$. It follows that $d$ is 
\emph{constant on $\Gamma$-orbits}.
This is the central observation: whenever $M \cdot \bx = \bb$ has a 
solution, then it also has a symmetric solution, that is a solution which is 
completely described by its entries on the individual $\Gamma$-orbits.

Finally, we make use of the $\ell$-homogeneity of $\mfA$. This property tells 
us that we can linearly order the $\Gamma$-orbits of the solution vectors in 
$\FPC$ by uniform formulas that only contain $\mcO(\ell)$ many variables.
Having this \FPC-definable linear order on the $\Gamma$-orbits and knowing that 
a solvable 
system $M \cdot \bx = \bb$ always has a symmetric solution (which is constant 
on the $\Gamma$-orbits) allows us to complete our argument as follows. 
Over ordered inputs, $\FPC$ can simulate all 
polynomial-time algorithms. In particular, $\FPC$ can simulate Gaussian 
elimination over ordered linear equation systems which allows us to find an 
\emph{ordered} (symmetric) solution or to conclude that the system is not 
solvable. 

\medskip
The second step is to define a generating set for the 
kernel $\ker(M) \leq \field F^J$ of the coefficient matrix $M\colon I \times J 
\to \field F$ in $\FPC$. 
\commentout{
We know that the structure $\mfA$ is 
$\ell$-homogeneous, which implies that we can order the $\Gamma$-orbits 
in~$\mfA$ using an $\FPC$-formula with  $\mcO(\ell)$ many variables.
But, indeed, we know more about $\Gamma$, namely we know that $\Gamma$ is an 
\emph{Abelian} ($p$-)group.
This fact has the important consequence that 
each $\Gamma$-orbit splits into \emph{singleton sets} when we fix any of 
the elements in the orbit (and this refined partition into singletons is 
\FPC-definable).
To state this in algebraic terms, let us say that we fix some element $j \in 
J$, and let us consider the corresponding $\Gamma$-orbit $\Gamma(j) \subseteq 
J$. We consider the stabiliser 
subgroup of $j$ in $\Gamma$ which we denote by $\Delta = \Stab(j) \leq \Gamma$.
Then we claim that the $\Delta$-orbits on $\Gamma(j)$ are singletons.
To see this, let $\delta \in \Delta$ and let $j' \in \Gamma(j)$.
Then, by the assumption that $j' \in \Gamma(j)$, we can find $\gamma \in 
\Gamma$ such that $\gamma(j) = j'$.
We then have 
\[ \delta(j') = \delta \gamma (j) = \gamma \delta (j) = \gamma(j) = j'.\]
Here, we used that $\delta \gamma = \gamma \delta$ which holds since $\Gamma$ 
is an Abelian group.
Moreover, by the fact that $\mfA$ is $\ell$-homogeneous, we also know that if 
we fix the element $j \in J$ as a parameter of an $\FPC$-formula, then we 
can obtain an $\FPC$-definable linear 
order on $\Gamma(j)$ (using $\mcO(\ell)$ many variables), because for all $j_1, 
j_2 \in J$ we have
\[ (\mfA, j, j_1) \equiv^{2 \cdot \ell} (\mfA, j, j_1) 
\quad\Longleftrightarrow\quad
 \Delta(j_1) = \Delta(j_2).
\]

In~\cite{GraedelGroPagPak19}, we said that $\ell$-homogeneous structures $\mfA$ with 
Abelian automorphism groups are \emph{cyclic}, for $\FPC$ can linearly 
order objects in $\mfA$ up to symmetries of $\mfA$ and, by fixing single 
elements 
as parameters, $\FPC$ can define a complete 
linear order on each individual orbit; thus, intuitively, each individual orbit 
possesses a \emph{cyclic} structure: if we fix one ``starting point'', then the 
remaining structure becomes a directed path (but, to be completely accurate, 
the different paths do not necessarily correspond to the unfolding of a single 
directed cycle, 
since the Abelian group can be a direct sum of cyclic groups).
In other words, a cyclic structure $\mfA$ can be linearly ordered up to classes 
of structurally indistinguishable elements (the $\Gamma$-orbits), and, 
each individual local class (each $\Gamma$-orbit) can be ordered completely 
using $\FPC$-formulas with $\mcO(\ell)$ many variables. 
Besides our applications in this section, the cyclicity property will play a 
crucial role later in this paper as well, so let us state our results in 
precise terms for future reference.
}
We have already seen how we can define a single 
solution of a cocyclic linear equation system in \FPC. We want to combine this 
result with Theorem~\ref{thm:ctype:cyclic} in order to define a 
generating set for $\ker(M)$ with a particular syntactic form that resembles 
the well-known row-echelon form. 
In order to describe this form, we need some notation.
First of all, we consider the linear order on $\Gamma$-orbits that is defined 
by $\ctype_\ell(x,y)$ in $\mfA$ on $J$. Let us denote the corresponding 
preorder by $\preceq$.
We write $J = J_0 \preceq J_1 \preceq \cdots \preceq J_{\mmo}$ to denote the 
ordered decomposition of $J$ into $\Gamma$-orbits $J_i$, $i < m$.

For $r < m$ we say that a vector $v \colon J \to \field F$ is 
\emph{$r$-homogeneous} if for all $r' < r$ and all $j \in J_{r'}$ we have $v(j) 
= 0$.
That is an $r$-homogeneous vector is zero on all orbits that precede the $r$-th 
one.
We now go one step further and use Theorem~\ref{thm:ctype:cyclic}.
For $r < m$ and $j \in J_r$ let us denote by $<_j$ the linear order on $J_r$ 
that is defined by $\ctype_{2\ell}[j](x,y)$ in $\mfA$.
We write $J_r = 0 <_j 1 <_j \cdots <_j |J_r|-1$ to identify the orbit 
$J_r$ with an initial segment $\inseg{|J_r|}$ of natural numbers according to 
the linear order 
$<_j$.
Let $s < |J_r|$, or equivalently, $s \in J_r = \inseg{|J_r|}$.
Then we say that a vector $v\colon J \to \field F$ is a 
\emph{$(j,s)$-generator} for $\ker(M)$ if:
\begin{itemize}
 \item $v \in \ker(M)$, and
 \item $v$ is $r$-homogeneous, and 
 \item for all $t < s$ we have $v(t) = 0$, and we have $v(s) = 1$.
\end{itemize}
The notion of a \emph{$(j,s)$-generator} very much resembles that of 
generating vectors in row-echelon form: the generating vector is zero on all 
positions that 
precede the $s$-th position in the $r$-th orbit \emph{and} the vector is 
non-zero at this particular position. 
However, what makes our notion different is that the order on the 
$r$-th orbit $J_r$ is \emph{not fixed}, but that it depends on the choice of 
the parameter $j \in J_r$. 
In fact, it can happen that a $(j,s)$-generator is a
$(j',s')$-generator, for $(j,s) \neq (j',s')$, because 
the position $s$ in $<_j$ and the position $s'$ in $<_{j'}$ may point to the 
same element in $J_r$.
This reflects the fact that, due to symmetries, we cannot select a unique $j 
\in J_r$ in a definable way. In particular, there is no canonical generating 
set for $\ker(M)$ in row-echelon form, not least because the row-echelon form 
requires an ordered index set for its definition.
This is why we have to work with $(j,s)$-generators instead.
Let us stress that this notion is well-defined only because of our assumption 
that $\mfA$ is $\ell$-cyclic.
In algebraic terms, note that a \emph{$(j,s)$-generator} is a vector which 
is almost symmetric: it can be defined by means of a single element $j \in J$ 
and, thus, is has a \emph{support} of size one. It is easy to come up with 
examples of families of linear equation systems in which \emph{no} solution has 
a support of sublinear size. Hence, the assumption of $\ell$-cyclicity is 
essential.

Clearly, our notion of $(j,s)$-generators allows us to define generating 
sets for $\ker(M)$ of polynomial size, since all tuples 
$(j,s)$ are contained in the set $J \times |J|$.
We are now prepared to complete our proof (sketch). The only thing that 
remains is to define a \emph{complete} set of $(j,s)$-generators for $\ker(M)$ 
in $\FPC$. 
To this end, we make use of our earlier argument of how we 
can define a single solution of a cocyclic linear equation systems in $\FPC$.
Let $\bx$ be a $J$-vector of variables ranging over $\field F$ and 
let us fix a tuple $(j,s)$ where $j \in J_r$ and $s < |J_r|$ 
(according to $<_j$) as above. In what follows, this tuple $(j,s)$ acts 
as a parameter in our $\FPC$-formula.
Then $(j,s)$-generators precisely correspond to solutions of the  
linear equation system $\textsc{Ker}(j,s)$ with variable set $\bx=(x_j)_{j \in 
J}$ and the following set of equations:
\begin{itemize}
 \item $M \cdot \bx = 0$, and
 \item $\bx(j') = 0$ for all $j' \in \biguplus_{r'<r} J_{r'}$, and
 \item for all $t < s$ we include the equation $\bx(t) = 0$ and for $s \in J_r$ 
the  $\bx(s) = 1$ (again, recall that we use $j \in J_r$ as a parameter to 
define the linear order $<_j$ on $J_r$ which allows us to equate $J_r$ with 
$\inseg{|J_r|}$).
\end{itemize}
Given the original coefficient matrix $M$, it is straightforward to define the 
linear equation systems $\textsc{Ker}(j,s)$ in $\FPC$. 
In particular, it follows that the systems $\textsc{Ker}(j,s)$ are cocyclic. 
Since we can define a single solution of any (solvable) cocyclic linear 
equation systems in $\FPC$ it follows that $\FPC$ can also define a 
$(j,s)$-generator for $\ker(M)$ (if such a generator exists), as claimed.
Note that the number of different $(j,s)$-generators (and, correspondingly, the 
number of different parameter tuples for the above linear system) is bounded by 
the set $J \times |J|$.
In other words, we obtain a generating set for $\ker(M)$ which is indexed by 
$J \times |J|$ as we claimed in.
Finally, it is straightforward to verify that any set of $(j,s)$-generators 
which is \emph{complete} (which means that it contains a $(j,s)$-generator for 
all tuples $(j,s)$ for which such a generator exist) generates $\ker(M)$.
\end{proof}

\section{CFI-Graphs and Linear-Algebraic 
Operators over Fields of Coprime Characterstic}
\label{sec:main-result}

We have derived the necessary background and are now well-prepared in order to 
formulate and prove our main (technical) result of this article.
We are going to show that CFI-structures over a prime field $\field F_p$ 
cannot be distinguished by means of \emph{any}
linear-algebraic operator over a field
$\field F$ with $\characteristic(\field F) \neq p$ if we apply such 
linear-algebraic operators to $\LC^{\Omega(n)}$-definable matrices.
Let us start with a precise statement of our result.
For what follows, recall that we consider CFI-structures over a fixed class of 
expander graphs $\mcF =  \{ G_n : n \in \N \}$ where each graph $G_n$ has $\mcO(n)$ 
vertices and is ordered, connected, and three-regular.

\begin{thm}
\label{thm:main:cfi}
 There is $\epsilon > 0$ such that for all large enough $n > 0$ the 
following holds. 
Let $\mfA = \CFIgraph{G_n}{p}{\lambda}$ and $\mfB = \CFIgraph{G_n}{p}{\sigma}$ 
denote two CFI-structures over $G_n$ and let $\field F$ be any field such that 
$\characteristic(\field F) \neq p$.
Then $\mfA$ and $\mfB$ are $\linisom{\field F}{\ell}{\LC^{k}}$ where $\ell = 
\lfloor \epsilon n\rfloor$ and $k = 3 \ell$.
\end{thm}
 
Of course, the statement of Theorem~\ref{thm:main:cfi} is only interesting 
in the case 
that the CFI-structures $\mfA$ and $\mfB$ are not isomorphic, that is for the 
case where $\sum \lambda \neq \sum \sigma$.
As a first step towards a proof of Theorem~\ref{thm:main:cfi}, let us briefly 
review what it means that $\mfA$ and $\mfB$ 
are  $\linisom{\field F}{\ell}{\LC^{k}}$.
First of all, we assigned to every structure $\mfA$ 
its \emph{Counting-Logic Algebra} 
$\clalgebra{\mfA}{\ell; \LC^{k}}{\field F}$ of 
dimension $\ell$ 
and width $k$ 
that consists of all $\field F$-linear combinations of matrices in 
$\CBMat[\mfA; \ell; \LC^{k}]$. 
The ordered set $\CBMat[\mfA; \ell; \LC^{k}]$, in turn, 
consists of all $\LC^{k}$-basis matrices that correspond to the 
individual $\LC^{k}$-types that are realised in $\mfA$ on 
$2 \ell$-tuples and that we 
view as square adjacency matrices over $\field F$ with entries in $\{0,1\}$ and 
with index set $A^{\ell} \times A^{\ell}$, cf.\ 
Section~\ref{sec:alg-structure}.

Reusing our notation from Definition~\ref{def:linalg:isom}, 
$I = A^\ell$ and $J=B^\ell$, and 
we denote by $M_i\colon I \times I \to \field F \in 
\CBMat[\mfA; \ell; \LC^k] \subseteq \clalgebra{\mfA}{\ell; \LC^k}{\field F}$ 
 and $N_i\colon J \times J \to \field F \in \CBMat[\mfB; \ell; \LC^k] 
\subseteq \clalgebra{\mfB}{\ell; \LC^k}{\field F}$ the $i$-th pair $(M_i, 
N_i)$ of corresponding $\LC^k$-basis matrices for $i < s$ where $s$ denotes 
the total number of realised $\LC^k$-types on $2\ell$-tuples in $\mfA$ (and 
$\mfB$). We let 
\begin{itemize}
 \item $\mcM = \{ M_i : i < s\} = \CBMat[\mfA; \ell; \LC^k]$ and 
 \item $\mcN = \{ N_i : i < s\} = \CBMat[\mfB; \ell; \LC^k]$,
\end{itemize}
and we obtain a $\inseg s$-indexed pair of matrix families $\mcM$ and $\mcN$ 
using the wording from Section~\ref{sec:simmatsim}.
In order to prove Theorem~\ref{thm:main:cfi} we have to show 
that the matrix families $\mcM$ and $\mcN$ are simultaneously similar over 
$\field F$. 

Recall from Section~\ref{sec:simmatsim} that we associated with $\mcM$ the 
$\field F$-algebra $\StabMat{\mcM}$ consisting of all $I \times I$-matrices 
which commute with all matrices in $\mcM$ and, in the analogous way, we defined 
the $\field F$-algebra $\StabMat{\mcN}$ associated with $\mcN$.
Moreover, we saw that the space $\HomMat{\mcM}{\mcN}$ consisting of all $I 
\times J$-matrices $Z$ over $\field F$ that satisfy $M_i\cdot Z = Z \cdot N_i$ 
for all $i < s$ forms a $\StabMat{\mcM}$-module with respect to matrix 
multiplication 
from the left (and it forms a $\StabMat{\mcN}$-module with respect to matrix 
multiplication from the right, but we won't make use of this fact).
Hence, in order to prove Theorem~\ref{thm:main:cfi} we  have to show 
that the $\StabMat{\mcM}$-module $\HomMat{\mcM}{\mcN}$ contains an invertible 
matrix $S \colon I \times J \to \field F$.
Of course, the obvious approach would be to construct such a matrix~$S$.
In fact, in his thesis~\cite{Holm10} Holm describes an explicit construction 
for the special case where $\ell = 1$ and $k > 2$. However, generalising this 
\emph{explicit} construction to higher arities $\ell > 1$ appears to be rather
hard, and, in fact, all of our approaches along these lines failed. Instead, 
we are going to take a completely different approach here.
We show that the \emph{existence} of such a matrix $S$ (but not necessarily the 
matrix $S$ itself) is definable in counting logic using $\mcO(k)$ many variables only. 
The attractive feature of our implicit approach is that we can derive 
the existence of such a matrix $S$ just from the definability of its existence.

\begin{thm}
\label{thm:main:definability}
Let $t \geq 3$ be a constant such that all CFI-structures in 
$\CFIclass{\mcF}{p}$ are $t$-homogeneous for all $p \in \Primes$.
Then there exists a constant $c \geq 1$ such that the following holds.
Let $\ell \geq 1$ and let $k \geq t \ell$. Then for each $p \in \Primes$ there 
exists a $\LC^{ck}$-sentence $\phi$ such that for all pairs of CFI-structures 
$\mfA = \CFIgraph{G_n}{p}{\lambda}$ and $\mfB = \CFIgraph{G_n}{p}{\sigma}$ over 
the same underlying graph $G_n \in \mcF$ we have that
$(\mfA, \mfB) \models \phi$ if, and only if, over every 
field $\field F$ with $\characteristic(\field F) \neq p$, the 
$\StabMat{\mcM}$-module $\HomMat{\mcM}{\mcN}$ contains an invertible matrix 
$S\in \HomMat{\mcM}{\mcN}$
where $\mcM = \CBMat[\mfA, \ell, k]$ and $\mcN = \CBMat[\mfB, \ell, k]$ (and 
where we understand $\StabMat{\mcM}$ as an $\field F$-algebra and 
$\HomMat{\mcM}{\mcN}$ as a $\StabMat{\mcM}$-module as before).
\end{thm}
 
The remainder of this section is devoted to the proof of
Theorem~\ref{thm:main:definability}. 
But before we start, let us see how we can derive 
Theorem~\ref{thm:main:cfi} from Theorem~\ref{thm:main:definability}.
First of all, let $c \geq 1$ and $t \geq 3$ be the constants according to 
Theorem~\ref{thm:main:definability}.  Let $p \in \Primes$.
Then, by Theorem~\ref{thm:cfi92}, we can find $\delta > 0$ 
such that for all large enough $n > 1$ we have 
$\mfA \equiv^{\lfloor \delta n \rfloor} \mfB$ 
where $\mfA = \CFIgraph{G_n}{p}{\lambda}$ and 
$\mfB = \CFIgraph{G_n}{p}{\sigma}$ are two CFI-structures over $\field F_p$ and 
the same underlying expander graph $G_n \in \mcF$ with $\mcO(n)$ many vertices. 
Let $\epsilon = \frac{1}{tc} \delta$. 
Then $(\mfA, \mfA) \equiv^{\lfloor tc\epsilon n \rfloor} 
(\mfA, \mfB)$. Let $\field F$ be a field such that $\characteristic(\field 
F) \neq p$. 
Let $\ell = \lfloor \epsilon n \rfloor$ and $k = \lfloor t\epsilon n \rfloor$.
We consider $\mcM = \CBMat[\mfA, \ell; \LC^{k}]$  and 
$\mcN = \CBMat[\mfB, \ell; \LC^{k}]$.
Since  the formula $\phi$ 
according to 
Theorem~\ref{thm:main:definability} contains at most 
$c k = c \cdot \lfloor t \epsilon n \rfloor \leq  \lfloor \delta n \rfloor$ 
many variables, this formula cannot distinguish between the ordered pairs of 
CFI-structures $(\mfA, \mfA)$ and $(\mfA, \mfB)$.
On the other hand, by its properties stated in 
Theorem~\ref{thm:main:definability}, $\phi$ would need to 
distinguish between $(\mfA, \mfA)$ and $(\mfA,\mfB)$ if no invertible matrix $S 
\in \HomMat{\mcM}{\mcN}$ would exist. Indeed, note that the 
$\StabMat{\mcM}$-module $\HomMat{\mcM}{\mcM}$ contains an 
invertible matrix $S \in \HomMat{\mcM}{\mcM}$ over every field $\field F$ for 
trivial reasons; for 
instance it contains the permutation matrix that corresponds to the identity 
automorphism of $\mfA$.
Hence, we can conclude that $\HomMat{\mcM}{\mcN}$ contains an invertible matrix 
which shows that $\mfA$ and $\mfB$ are $\linisom{\field F}{\ell}{\LC^{k}}$, and 
thus Theorem~\ref{thm:main:cfi} follows, because
$\linisom{\field F}{\ell}{\LC^{k}}$ structures are also
$\linisom{\field F}{\ell}{\LC^{3 \ell}}$ since $k \geq 3\ell$.

\subsubsection*{Proof of Theorem~\ref{thm:main:definability}}
Our proof of Theorem~\ref{thm:main:definability} is structured as 
follows. 
First, we fix a \emph{prime field}~$\field F$ with $\characteristic(\field F) 
\neq p$.
We are going to construct a sentence $\phi_{\field F} \in \LC^{\omega}$,
with at most $c \cdot k$ many variables,
which holds in the ordered pair $(\mfA, \mfB)$ of CFI-structures
$\mfA$ and $\mfB$ if, and only if, $\HomMat{\mcM}{\mcN}$ (considered as a
$\StabMat{\mcM}$-module over the $\field F$-algebra $\StabMat{\mcM}$) contains 
an invertible matrix $S$.
We use these sentences $\phi_{\field F}$ to obtain the desired 
sentence $\phi$ according to Theorem~\ref{thm:main:definability} which talks 
about \emph{all} fields $\field F$ with $\characteristic(F)\neq p$. 
More precisely, $\phi$ is the conjunction over all sentences $\phi_{\field F}$ 
for \emph{prime fields} $\field F$ with $\characteristic(\field F) \neq p$.

\begin{enumerate}
 \item First we show that the final step of the construction is sound. 
 Specifically, we show that it suffices to restrict our considerations to  
\emph{prime fields}. This  observation is important 
because we will frequently apply Theorem~\ref{thm:solvcocylic} in order 
to 
define solution spaces of cocyclic linear equation systems and, indeed, 
we only formulated and proved Theorem~\ref{thm:solvcocylic} for the case of 
prime fields.
\label{proofplan1}

 \item Secondly, we make use of our results from Section~\ref{sec:simmatsim}. 
In 
particular, we recall the notion of a block-generated pair of matrix families 
from Section~\ref{sec:blockmatrices}, and we show that the two bases 
$\mcM=\CBMat[\mfA, \ell, k]$ and
$\mcN=\CBMat[\mfB, \ell, k]$ 
for the  counting logic 
algebras $\clalgebra{\mfA}{\ell; \LC^k}{\field F}$ 
and 
$\clalgebra{\mfB}{\ell; \LC^k}{\field F}$
form such a pair; indeed we show that $(\mcM, \mcN)$
even is a pair of \emph{faithfully block generated} matrix families, see 
Definition~\ref{def:fblockgen}. We further show that the matrix families $\mcM$ 
and $\mcN$ are \emph{locally simultaneously similar} 
(\locsimsim, 
for short), see Section~\ref{sec:locsim}, Definition~\ref{def:locsim}.
This allows us to apply our criterion from
Theorem~\ref{thm:criterion:smsproblem}: in order to check whether 
$\HomMat{\mcM}{\mcN}$ contains an invertible matrix $S$, it suffices to 
check whether the $\StabMatDiag{\mcM}$-module $\HomMatDiag{\mcM}{\mcN}$ is 
cyclic.  
Recall that $\StabMatDiag{\mcM}$ and $\HomMatDiag{\mcM}{\mcN}$ denote the 
diagonal subalgebra and submodule of $\StabMat{\mcM}$ and 
$\HomMatDiag{\mcM}{\mcN}$, respectively, see 
also Corollary~\ref{cor:fbg:diagmod}.  
\label{proofplan3}

 \item The third step is the core of our whole argument. We are going to 
 combine results on the 
$\FPC$-definability of the automorphism groups and orbits of CFI-structures 
with Maschke's 
Theorem, cf.\ Section~\ref{sec:maschke},  Theorem~\ref{thm:maschke}, in order 
to show that the $\field F$-algebra $\StabMatDiag{\mcM}$ is \emph{semisimple}.
It follows that the $\StabMatDiag{\mcM}$-module $\HomMatDiag{\mcM}{\mcN}$ is 
semisimple (Theorem~\ref{thm:semisimple}). \label{proofplan2}

\item Next, we make use of the semi-simplicity of $\HomMatDiag{\mcM}{\mcN}$ in order to decompose the module into ``small'' submodules. 
Moreover, by applying Theorem~\ref{thm:solvcocylic}, we show that we can define 
generating sets for the respective submodules in counting logic by using at 
most 
$c\cdot k$ many variables.
Let us stress that this decomposition only becomes possible due to the 
semi-simplicity of the module $\HomMatDiag{\mcM}{\mcN}$ which follows from our 
application of 
Maschke's Theorem in step~(\ref{proofplan2}).
\label{proofplan4}

\item Finally, we construct the formula $\phi_{\field F}$. By~(\ref{proofplan3}), 
the formula $\phi_{\field F}$ needs to verify that the \emph{semisimple} 
$\StabMatDiag{\mcM}$-module $\HomMatDiag{\mcM}{\mcN}$ is cyclic. 
We approach this problem by expressing a more general query, namely we 
determine the full isomorphism type of the module $\HomMatDiag{\mcM}{\mcN}$ by 
means of a formula of counting logic.
Thanks to our preparation, this becomes possible in the following way.
First of all, we start by determining the isomorphism types of all simple subalgebras of $\StabMatDiag{\mcM}$. 
This we can easily do in counting logic because $\StabMatDiag{\mcM}$ has an (\FPC-definable) ordered basis.
Then we know, from Section~\ref{sec:maschke}, that the isomorphism type 
of $\HomMatDiag{\mcM}{\mcN}$ is (uniquely) determined by the multiplicities of the simple subalgebras of $\StabMatDiag{\mcM}$ as they occur in a decomposition of $\HomMatDiag{\mcM}{\mcN}$ into a direct sum of simple submodules. 
By using our decomposition from step~(\ref{proofplan4}), we can easily determine those multiplicities componentwise, since we can linearly order (again in an $\FPC$-definable way) each of the ``small'' submodules that occur in the decomposition of 
$\HomMatDiag{\mcM}{\mcN}$. In this way we can determine the multiplicities for each individual component which add up to the total multiplicities for the whole module 
$\HomMatDiag{\mcM}{\mcN}$.
Moreover, the number of variables required to express these properties in 
counting logic is, again, bounded by $c  \cdot k$. 
Since the isomorphism type determines the cyclicity of the module, we can  
obtain our desired formula $\phi_{\field F}$ by selecting modules with the 
appropriate isomorphism types.
\label{proofplan5}
\end{enumerate}
 
\paragraph*{Notation}
But before we delve into the details, let us discuss some further notations 
and assumptions.
First of all, the existence of the constant $c \geq 1$ will follow implicitly 
from our proof in which we formulate various requirements on $c \geq 1$ along 
the way. For instance, one important such constraint is that $c \geq 1$ is 
large enough so that we can define the linear preorder on $\ell$-tuples up to 
$\LC^k$-equivalence (in CFI-structures this means up to orbits, since $k \geq 
t\cdot\ell$, and since CFI-structures are $t$-homogeneous by the choice of $t 
\geq 2$) using an $\FPC$-formula with at most $c \cdot k$ many 
variables, recall 
Definition~\ref{def:homogeneity} 
and Theorem~\ref{thm:homogeneity}. 
For the remainder of the proof we are going to assume that the 
given 
CFI-structures $\mfA$ and $\mfB$ are $\LC^{ck}$-equivalent.
This assumption involves no loss of generality.
In fact, it is not hard to see that if one could distinguish 
$\mfA$ and $\mfB$ in counting logic using at most $ck$ many variables, then
one could identify all CFI-structures from $\CFIclass{\mcF}{p}$ over the underlying
graph $G_n \in \mcF$ up to 
isomorphism in $\LC^{ck}$. Hence, we could define any kind of query of the pair
$(\mfA, \mfB)$ in $\LC^{ck}$ (in particular, we could define the query stated
in Theorem~\ref{thm:main:definability}).
Next, we recall from Section~\ref{sec:alg-structure} that, independent of the 
underlying field $\field F$, the counting logic 
($\field F$-)algebras of the CFI-structures $\mfA$ and $\mfB$ of dimension 
$\ell$ and width~$k$, that is the algebras
$\clalgebra{\mfA}{\ell; \LC^k}{\field F}$ and
$\clalgebra{\mfB}{\ell; \LC^k}{\field F}$, are isomorphic.
Recall that these algebras consist of 
all $\field F$-linear combinations of the basic $\LC^k$-type matrices with 
entries in $\{0,1\}$ that is $\field F$-linear combinations of matrices in 
$\mcM=\CBMat[\mfA, \ell, k]$ and
$\mcN=\CBMat[\mfB, \ell, k]$, respectively.
Specifically, the sets $\mcM$ and $\mcN$ are linearly ordered according
to the $\LC^k$-types on $2\ell$-tuples with respect to
the formula $\ctype_k(\tup x, \tup y)$ in both structures $\mfA$ and $\mfB$;
that is
$\CBMat[\mfA; \ell; k] = \{ M_0 < M_1 < \cdots < M_\smo \}$ 
and
$\CBMat[\mfB; \ell; k] = \{ N_0 < N_1 < \cdots < N_\smo \}$ 
and such that:
\begin{itemize}
 \item for $i < s$, the matrices $M_i$ and $N_i$ correspond to the $i$-th 
$\LC^k$-type on $2\ell$-tuples according to $\ctype_k(\tup x, \tup y)$ in 
$\mfA$ and $\mfB$, respectively, and they have entries in $\{ 0, 1 \}$, and
 \item the mapping defined by $M_i \mapsto N_i$, for $i < s$, extends to an 
$\field F$-algebra isomorphism between 
$\clalgebra{\mfA}{\ell; \LC^k}{\field F}$ and 
$\clalgebra{\mfB}{\ell; \LC^k}{\field F}$,  
cf.\ Section~\ref{sec:alg-structure}. 
\end{itemize}
For what follows, we set $I = A^{\ell}$ and $J = B^\ell$. 
Then the matrices in the counting logic algebra of $\mfA$ are $I \times 
I$-matrices and, correspondingly, the matrices in the counting logic algebra 
 of $\mfB$ are square matrices of the form $J \times J$.
 
 \paragraph*{(\ref{proofplan1}) Restriction to prime fields}
 Let us start with a simple, but useful, observation.
 As we said above, we want to argue that it is sufficient 
 to conduct our considerations for prime  fields only.
 In order to verify this, let us assume that for each underlying 
\emph{prime} field~$\field F$, with $\characteristic(\field F) \neq p$, the 
 $\StabMat{\mcM}$-module $\HomMat{\mcM}{\mcN}$ 
 contains an invertible matrix $S \in \HomMat{\mcM}{\mcN}$.
 We then have to show that the same holds for
 \emph{all} underlying fields $\field F$ with $\characteristic(\field F) \neq p$.
 This, however, turns out to be obvious, because 
 the matrix families $\mcM$ and $\mcN$ contain matrices with entries in $\{0, 
 1\} \subseteq \field F$ only. 
 In particular, the matrices in $\mcM$ and $\mcN$ are always matrices whose entries reside
 in the prime field of~$\field F$.
 Formally, let us fix any field $\field F$ with $\characteristic(\field F) \neq p$ and
 let us denote by $\PrimeF(\field F)$ its prime field.
 By our assumption we can find an invertible matrix $S\colon I \times J \to \PrimeF(\field F)$ which simultaneously transforms all matrices 
 $M_i$ to $N_i$, for $i < s$, that is $M_i \cdot S = S \cdot N_i$. In this equation, all operations take place in $\PrimeF(\field F)$. Hence, it readily follows that the same matrix $S$ witnesses that $\mcM$ and $\mcN$ are simultaneously similar over the 
 whole field $\field F$.
 
 From now on, $\field F$ denotes a prime field with $\characteristic(\field F) \neq p$, that is $\field F = \mbQ$ or $\field F = \field F_q$ for $q \in \Primes$, $p \neq q$.

 \paragraph*{(\ref{proofplan3}) Reduction to the cyclicity of the diagonal 
 $\StabMatDiag{\mcM}$-module $\HomMatDiag{\mcM}{\mcN}$}
 The pair of matrix families $(\mcM, \mcN)$ has some special properties 
 that allow us to reduce the question of whether the $\StabMat{\mcM}$-module
 $\HomMat{\mcM}{\mcN}$ contains an invertible matrix to the question of 
 whether the ``diagonal'' $\StabMatDiag{\mcM}$-module $\HomMatDiag{\mcM}{\mcN}$ 
is cyclic.
 Specifically, we are going to show that $(\mcM, \mcN)$ is a 
 \emph{faithfully block generated} pair of matrix families that are 
 \emph{locally simultaneously similar}, see Section~\ref{sec:blockmatrices} and 
Section~\ref{sec:locsim}.
  This allows us to apply our criterion formulated
  as Theorem~\ref{thm:sms:char:fbg}.
 
  First of all, it is not hard to see that $(\mcM, \mcN)$ is faithfully 
  block-generated.
  Recall that the matrices in $\mcM$ are indexed by $I \times I$ and that the matrices in 
  $\mcN$ are $J \times J$-matrices, where $I = A^\ell$ and $J = B^\ell$.
  Of course, in order to talk about block matrices and compatible block 
  matrices 
  at all, we
  require a coloured index pair 
  $(I, J, \preceq_I, \preceq_J)$ that provides us with 
  partitions of the index sets $I$ and $J$
  as
  $I = I_0 \preceq_I \cdots \preceq_I I_\mmo$ 
  and  $J = J_0 \preceq_J \cdots \preceq_J J_\mmo$ 
  into corresponding pairs of colour classes $I_i, J_i$ of the same sizes, 
  see Section~\ref{sec:blockmatrices}.
  These (compatible) ordered partitions are readily provided by the refinements 
  of $I$ and $J$ with respect to $\equiv^k$-equivalence (in both CFI-structures $\mfA$ 
  and $\mfB$, respectively).
  In particular, by our assumption that $\mfA \equiv^k \mfB$ we know that 
  the corresponding $\equiv^k$-classes have the same sizes in both structures.
  We can even say a bit more. Indeed, by our assumption on the 
  constant $t \geq 2$ and the homogeneity of CFI-structures, we know that 
  the partitions of $I$ and $J$ coincide with the respective partitions into 
orbits.
  
  The requirement for $(\mcM, \mcN)$ to be block-generated is the existence of 
  a basis for $\mcM \circ_{\inseg s} \mcN$ that consists of pairs of compatible 
block matrices.
  However, since the matrices in $\mcM$ and $\mcN$ are pairwise disjoint, and 
  because of the fact that all pairs of matrices $M_i \in \mcM$ and $N_i \in 
\mcN$, $i < s $, are compatible block matrices (as they correspond to the same
$\LC^k$-types) we can simply take $(\mcM, \mcN)$ itself as this  basis.
 Moreover, it is not hard to see that $(\mcM, \mcN)$ is also \emph{faithfully} 
block-generated. We only need to show that for each pair of diagonal blocks 
$I_d \times I_d$, and $J_d \times J_d$, $d < m$, the families $\mcM$ and 
$\mcN$ contain a pair of matrices $M_i \in \mcM$ and $N_i \in \mcN$, $i < s$, 
 such that $M_i$ is the identity matrix on the diagonal block $I_d \times I_d$ 
and such that $N_i$ is the identity matrix on the diagonal block $J_d \times 
J_d$.
However, this easily follows since the diagonal types on the $d$-th diagonal 
blocks are particular $\LC^k$-types on $2\ell$-tuples which are determined 
by the $\LC^k$-formula $\phi(\bx, \by)$ which expresses that $\bx$ has 
$\LC^k$-type $I_d$ (or, equivalently, $J_d$) and that $\bx = \by$.
  
Hence, it only remains to show that the faithfully block-generated pair of 
matrix
families $(\mcM, \mcN)$ is also \emph{locally simultaneously 
similar} over $\field F$ (which can be considered as a first step towards our 
overall goal to show that $\mcM$ and $\mcN$ are (globally) simultaneously 
similar). To this end, according to Definition~\ref{def:locsim}, 
we have to show for 
each diagonal block $I_d \times J_d$, $d < m$, that the $\StabMat{\mcM}$-module 
$\HomMat{\mcM}{\mcN}$ contains a matrix $S \in \HomMat{\mcM}{\mcN}$ such that 
$\Diag_d(S)$ is invertible (when we consider $\Diag_d(S)$ as an $I_d \times 
J_d$-matrix).
Recall from Definition~\ref{def:projdiag} that we denote by 
  $\Diag_d(M)\colon I \times J \to \field F$ the projection of a matrix 
  $M\colon I \times J \to \field F$ onto the $d$-th diagonal block, that is 
  $\Diag_d(M)$ coincides with the matrix $M$ for
  all entries in $I_d \times J_d$ and $\Diag_d(M)$ has entry $0$ for all
  other positions.
Similar to our general strategy we will not try to construct such a matrix
$S \in \HomMat{\mcM}{\mcN}$ explicitly. 
Instead we prove its existence by means of the undefinability of the 
CFI-problem in counting logic (Theorem~\ref{thm:cfi92}), the (linear-algebraic) 
structure of $\HomMat{\mcM}{\mcN}$, and
our result about the definability of solution (spaces) of cocyclic linear 
equation systems (Theorem~\ref{thm:solvcocylic}).

First of all, let us recall that $\HomMat{\mcM}{\mcN}$ is a 
(homogeneous)
$\field F$-linear space. In fact, it consists of all matrices 
$S\colon I \times J \to \field F$ that satisfy 
the condition $M_i \cdot S = S \cdot N_i$ for all $i < s$.
If we view the entries of the matrix $S$ as individual variables $S(i,j)$ which 
range over $\field F$, then this condition can easily be written down as a 
system $\mcS$ of linear equations over $\field F$
(the system contains one equation per matrix pair $(M_i,S_i)$ and 
corresponding entry $(i,j) \in I\times J$).
Now, let us fix a diagonal block $I_d \times J_d$ for some $d < m$.
Moreover, let us choose two parameter (tuples) $i \in I = A^\ell$ and $j \in J 
= 
B^\ell$.
As we proved in Theorem~\ref{thm:ctype:cyclic}, with these parameters we can 
define two linear orderings $<_i$ and $<_j$ on the orbits $I_d$ and $J_d$, 
respectively, by using an $\FPC$-formula with at most $c\cdot k$ many variables 
(note that the number of variables of this formula is determined by $\ell$ and 
the homogeneity constant for the class $\CFIclass{\mcF}{p}$; hence, we can 
choose $c \geq 1$ large enough such that $c\cdot k$ variables are sufficient, 
indeed).
Having ordered both orbits $I_d$ and $J_d$ we can then easily define a 
bijection between $I_d$ and $J_d$ by sending elements 
with the corresponding positions according to $<_i$ and $<_j$ to 
each other. Of course, this bijection can also be written as an $(I_d \times 
J_d)$-permutation matrix $P_{i,j}:I_d \times J_d \to \{ 0,1\}$  
(again, we can let $c \geq 1$ be large enough such that $c\cdot k$ variables are 
sufficient to define this matrix in counting logic).
We can now extend our linear equation system $\mcS$ to a linear equation system 
$\mcS^d_{i,j}$ by adding a set of equations that enforces that the $d$-th 
diagonal block of $S\colon I \times J \to \field F$ equals the (invertible) 
permutation matrix $P_{i,j}$. 
Note that if the resulting linear equation system $\mcS^d_{i,j}$ is solvable, 
then we can find a matrix $S \in \HomMat{\mcM}{\mcN}$ which is invertible on the 
$d$-th diagonal block $I_d \times J_d$.

Moreover, according to Theorem~\ref{thm:solvcocylic}, we can select a sentence 
$\phi$ of counting logic which checks, given an ordered pair of CFI-structures 
$(\mfA, \mfB)$, whether for some choice of parameters $(i,j)\in I_d \times J_d$ 
the resulting cocyclic linear equation system $\mcS^d_{i,j}$ is solvable or not
(here we also rely on the simple observation that ordered pairs of 
$r$-cyclic CFI-structures are $r$-cyclic structures as well).
The number of variables of this sentence $\phi$ is determined by $\ell$, the 
homogeneity constant of the CFI-class $\CFIclass{\mcF}{p}$, and other constants 
such as those that are required to formalise the construction 
of Theorem~\ref{thm:solvcocylic} in counting logic.
Hence, once again, we can let $c\geq 1$ be large enough such that $c \cdot k$ 
variables are sufficient to construct this formula $\phi \in \LC^\omega$.
 
We come to our final argument. As
we assumed that $\mfA \equiv^{c\cdot k} \mfB$ we know that the formula $\phi$ 
cannot distinguish between the ordered pairs $(\mfA, \mfA)$ and $(\mfA, \mfB)$.
However, in $(\mfA, \mfA)$ the system $\mcS^d_{i,i}$ is clearly solvable by 
means of the identity automorphism, hence it must be solvable in $(\mfA, \mfB)$ 
as well. Hence, we can conclude that there exists a matrix $S \in 
\HomMat{\mcM}{\mcN}$ such that $\Diag_d(S)$ is an invertible ($I_d 
\times J_d$)-matrix as claimed (in fact, by our argument, we can even choose 
$S$ such that that $\Diag_d(S)$ is a permutation matrix).

  \paragraph*{(\ref{proofplan2}) Application of Maschke's Theorem}
  The next step is to show that the $\field F$-algebra 
  $\StabMatDiag{\mcM}$ is \emph{semisimple}.
  In order to show that $\StabMatDiag{\mcM}$ is semisimple we are going to 
  embed $\StabMatDiag{\mcM}$ into a larger $\field F$-algebra for which
  we can then show, by an application of 
  Maschke's Theorem, that this larger surrounding $\field F$-algebra is 
  semisimple. Since subalgebras of semisimple algebras are semisimple as 
  well, the result follows.
  
  First of all, we make use of the homogeneity property of $\mfA$.
  As above, we write $I = I_0 \preceq \cdots \preceq I_\mmo$ to denote the
  ordered decomposition of the index set $I = A^\ell$ 
  into  $\LC^k$-types.
  By $\Gamma$ we denote the automorphism group of $\mfA$ 
  which is an elementary Abelian $p$-group. 
  Since we chose the constant $t$ large enough so that CFI-structures in 
$\CFIclass{\mcF}{p}$ are $t$-homogeneous, we know that the partition of $I$ 
  into $\LC^k$-types corresponds to the partition of $I$ into $\Gamma$-orbits, 
  that is $\Gamma$ acts transitively on each of the sets $I_i$, $i< m$.
  Moreover, again by our choice of $c\geq 1$, the linear preorder on $I$ is 
  definable in $\mfA$ by using a $\LC^\omega$-formula with at 
  most $c \cdot  k$ many variables.
  
  We now want to take a closer look at one of the diagonal blocks $I_d \times 
  I_d$, for  $d < m$.  More specifically, we want to analyse the structure of 
  matrices in $\StabMatDiag{\mcM}$  when we restrict to 
  this diagonal block $I_d \times I_d$. 
  In particular, we are interested in matrices of the form $\Diag_d(M)$ for 
  $M \in \StabMatDiag{\mcM}$.
  Our motivation comes from the following straightforward observation.
  \newcommand{\DAlg}{\ensuremath{\text{C}}}
  Let us denote by $\DAlg_d$ the $\field F$-algebra which consists of all 
  $I\times I$-matrices   $\Diag_d(M)$ for $M \in \StabMatDiag{\mcM}$, 
  formally \[\DAlg_d = \Diag_d(\StabMatDiag{\mcM}) = \{ \Diag_d(M) : M \in 
  \StabMatDiag{\mcM} \}.\]
  Then, each $\DAlg_d$ forms an $\field F$-algebra, and, for trivial 
reasons, we have
  \[ \StabMatDiag{\mcM} \subseteq \DAlg_0 \oplus \DAlg_1 \oplus \cdots \oplus 
  {\DAlg}_{\mmo}. \] 
  In particular, if we can show that each $\field F$-algebra $\DAlg_d$ 
  is  semisimple, then it follows that the $\field F$-algebra  
$\StabMatDiag{\mcM}$ is semisimple as well (as the class of semisimple algebras is closed under taking subalgebras).
  
  Let us fix one of the $\field F$-algebras $\DAlg_d$, $d < m$.
  Moreover, let us view the automorphisms $\pi \in \Gamma$ of the CFI-structure $\mfA$
  as $I \times I$-permutation matrices $\Pi\colon I\times I \to \{0,1\} 
  \subseteq \field F$ in the usual way, that is $\Pi(i,j) = 1$ if, and only if, 
  $\pi(i)=j$ for $i, j \in I$.
  Note that all matrices $\Pi \in \Gamma$ are diagonal block matrices, that is
  $\Diag(\Pi) = \Pi$. This is simply because 
  the \emph{definable} partition of $I$ into $C^k$-types $I_0, \dots, I_\mmo$ 
is preserved by   all  automorphism $\pi \in \Gamma$.
  Let us set $\Gamma_d = \Diag_d(\Gamma) = \{ \Diag_d(\Pi) : \Pi \in \Gamma \}$.
  Then it is easy to see that each of the sets $\Gamma_d$ forms a group with respect to
  matrix multiplication. 
  We now claim that $\DAlg_d = \field F[\Gamma_d]$, cf.\ Definition~\ref{def:groupalgebra}
  where we defined the notion of a \emph{group algebra} over a field $\field F$.
  
  To verify our claim that $\DAlg_d = \field F[\Gamma_d]$ we show two things.
  First of all, we show that $\Gamma_d \subseteq \DAlg_d$ which, in turn, implies that
  $\field F[\Gamma_d] \subseteq \DAlg_d$.
  Indeed, for all $\Pi \in \Gamma$ we have $M \cdot \Pi = \Pi \cdot M$  (or, 
stated equivalently, $\Pi \cdot M \cdot \Pi^{-1} = M$) for all $M \in \mcM$, 
since each matrix $M \in \mcM$ is \emph{definable} in counting logic and, thus, 
is invariant under automorphisms $\Pi \in \Gamma$.
  Hence, $\Gamma \subseteq \StabMatDiag{\mcM}$ and thus $\Diag_d(\Gamma) = \Gamma_d \subseteq \DAlg_d = \Diag_d(\StabMatDiag{\mcM})$. 

  The remaining direction $\DAlg_d \subseteq \field F[\Gamma_d]$ is more interesting. 
  We have to show that each matrix $X \in \DAlg_d$ can be written as  an $\field 
F$-linear combination of matrices in $\Gamma_d$.
  Our first observation is that each such $X$ itself is invariant under 
  the action of $\Gamma$. Since $X \in \DAlg_d$, this is equivalent to saying
  that for all $\Pi   \in \Gamma_d$ we have $X \cdot \Pi = \Pi \cdot X$.
  In order to verify this we show that each $\Pi \in \Gamma_d$ 
  is contained  in $\mcM$, that is $\Gamma_d \subseteq \mcM$.
  To this end, recall from Section~\ref{sec:generalised_cfi} that  each 
  automorphism $\Pi \in \Gamma$ can uniquely be described by means of
  a vector  $v_\Pi\colon E \to \field F_p$ which determines the action
  of $\Pi$ 
  on every edge class $e \in E$ in  terms of a cyclic $\field F_p$-shift.
  Here,  $E$ denotes the edge relation of the underlying graph $G_n$.
  By our assumption that $G_n$ is ordered, the set  $E$ of edges is ordered
  as well, which implies that we can describe each automorphism $\Pi \in 
  \Gamma$ as an ordered object. 
  In particular, for each fixed $\Pi \in \Gamma$ we can write down a formula
  of counting logic $\phi_\Pi(x_0, \dots, x_\lmo, y_0, \dots, y_\lmo)$
  which defines the matrix $\Pi$.
  This formula $\phi_\Pi$ only needs to express that for
  each position $i < \ell$  the corresponding pair of variables $(x_i, 
  y_i)$ is interpreted  by a pair $(a,b)$ of elements 
  from the (same) edge class $a,b \in e^p$ (for $e \in E$) 
  such that $b$ results from $a$ by a cyclic shift of length $\Pi(e) \in 
\field F_p$.
  This can easily be expressed in counting logic by using the 
  cycle relation of the CFI-structure $\mfA$ and two 
  additional auxiliary variables.
  In particular, $\phi_\Pi$ can be expressed as a $\LC^{ck}$-formula (under the mild 
  assumption that $c \geq 3$).
  This argument already shows that $\Pi$ is contained in the counting logic algebra 
  $\clalgebra{\mfA}{\ell; \LC^k}{\field F}$ of $\mfA$ of dimension $\ell$ and 
  width $k$ over~$\field F$.
  In fact, each $\Pi \in \Gamma_d$ is even contained in the 
  basis $\mcM = \CBMat[\mfA, \ell, k]$ of this algebra.
  For the sake of contradiction, assume that $\Pi \not\in \mcM$. 
  Then, because of the fact that the matrices in $\mcM$ are disjoint matrices 
  with entries in $\{ 0,1 \}$, it  follows that $\mcM$ contains a (non-zero)
  matrix $Y \in \mcM$
  which strictly refines $\Pi$, in the sense that $Y(i,j)=1$ implies 
  $\Pi(i,j)=1$ for all $i, j   \in I_d$, but $Y \neq \Pi$.
  But this would mean that, by using this matrix $Y$ we could refine the set
  $I_d$, because leaving out any entry of $\Pi$ would leave us with a
  partial function $Y$ defined on $I_d$ which means that 
  elements from $I_d$ in the domain of this function $Y$ could be isolated using the
  $\LC^k$-type corresponding to $Y$.
  This, however, is impossible because we have already refined the set $I$ up to $\Gamma$-orbits and, clearly, orbits cannot be broken up in any definable way.
  
  So far we have seen  that $\Gamma_d \subseteq \mcM$. From the definition of 
  $\DAlg_d$ it follows that for
  each matrix $X \in \DAlg_d$ we have $\Pi \cdot X = X \cdot \Pi$ for 
  all $\Pi \in \Gamma_d$. This means that $X$ is invariant on 
  $\Gamma_d$-orbits, that is for each 
  position $(i,j) \in I_d \times I_d$ and each $\pi \in \Gamma_d$ we have $X(i,j) = X(\pi(i),\pi(j))$.
  Let us denote by $\mcP$ the partition of $I_d \times I_d$ into $\Gamma_d$-orbits.
  Then we can identify $X$ with the mapping $X\colon \mcP \to \field F$ which 
  is defined as $X(P) = X(i,j)$ for some $(i,j) \in P$.
  Moreover, we claim that the parts $P \in \mcP$ precisely correspond to the 
  permutation matrices $\Pi \in \Gamma_d$.
  We first observe that each $P \in \mcP$ is the graph of a bijective function 
  $I_d \to I_d$.  To see this, first note that since $\Gamma_d$ acts 
  transitively on $I_d$ each 
  element $i \in I_d$ must have at least one image according to $P$. 
  On the other hand, assume that $(i, j), (i, j') \in P$. Then we can find $\sigma \in \Gamma_d$ such that $\sigma(i,j) = \sigma(i,j')$. This, however, implies that $j = j'$: let $\rho \in \Gamma_d$ be such that $\rho(j) = i$, hence 
  $\rho^{-1} \sigma \rho(j)= j$. Thus, $\rho^{-1} \rho \sigma(j) = j$, because $\Gamma_d$
  is Abelian. Hence, $\sigma(j) = j$.
  With the same arguments, we can see that also each $j \in I_d$ must have a unique preimage according to $P$. 
  All what remains is to show that the bijective functions $P \in \mcP$ correspond to 
  the permutation matrices $\Pi \in \Gamma_d$. 
  To see this, first note that for each pair $(i,j) \in I_d \times I_d$ 
  there exists some $\Pi \in \Gamma_d$ such that $\pi(i) = j$. This follows from the fact
  that $I_d$ is a $\Gamma_d$-orbit.
  But then, the part $P \in \mcP$ which contains $(i,j)$ must coincide with $\Pi$ for 
  $P$ is invariant under the action of the Abelian group $\Gamma_d$. 
  To sum up, if we denote for $P \in \mcP$ the corresponding permutation matrix by $\Pi_P \in \Gamma_d$, then we can express $X \in \DAlg_d$ as the following $\field F$-linear combination of permutation matrices $\Pi \in \Gamma_d$:
  \[  X = \sum_{P \in \mcP} X(P) \cdot \Pi_P.\] 
  Hence, $\DAlg_d \subseteq \field F[\Gamma_d]$. We put everything together and 
conclude that:
  \[ \field F[\Gamma_d] = \DAlg_d.\]
  
  Finally, we apply Maschke's Theorem, cf.\ Theorem~\ref{thm:maschke}. Since $\Gamma_d$ is an elementary Abelian $p$-group, and since $\characteristic(\field F) \neq p$, Maschke's 
  Theorem tells us that $\field F[\Gamma_d]$ is semisimple, and thus, the algebra 
  $\DAlg_d$ is semisimple as well.
  As a consequence, the $\field F$-algebra $\StabMatDiag{\mcM}$ and thus the 
  $\StabMatDiag{\mcM}$-module $\HomMatDiag{\mcM}{\mcN}$ are semisimple, which is what we wanted to show. 
  
  \smallskip
  Before we proceed, let us remark that our choice to focus on 
  the algebra $\StabMatDiag{\mcM}$, and thus on the CFI-structure $\mfA$, 
  is no restriction of generality in the sense that it can easily be 
  shown in the same way that  
  the $\field F$-algebra $\StabMatDiag{\mcN}$ is semisimple.
  However, we have defined $\HomMatDiag{\mcM}{\mcN}$ as a 
  \emph{left} $\StabMat{\mcM}$-module which is why we phrase and present 
  these results only in terms of $\mfA$ and $\StabMatDiag{\mcM}$.

  We make another observation that will become important later on. 
  We saw that 
  $\StabMatDiag{\mcM} \subseteq \bigoplus_{d < m} {\DAlg}_{d}$
  and that each of the $\field F$-algebras $\DAlg_d$, $d < m$, 
  satisfies $\field F[\Gamma_d] = \DAlg_d$.
  This immediately shows that we can define in $\LC^{ck}$, for each of the 
  $\field F$-algebras $\DAlg_d$, a linearly ordered $\field F$-basis.
  Indeed, $\Gamma_d$ forms such a basis and, as we explained before, we can
  easily describe automorphisms $\Pi \in \Gamma$ in counting logic as ordered 
sequences of cyclic shifts on the individual edge classes of the CFI-structure 
  $\mfA$.
  Note that, in this regard, we are again crucially relying on the fact that we 
  are working with  CFI-structures over \emph{ordered} underlying 
  graphs $G_n \in \mcF$.
  In particular, the basis that we obtain is small 
  since $|\Gamma_d| = |I_d| = |J_d|$ which, in turn, follows from the fact that 
  $\Gamma_d$ is an Abelian group that acts transitively on $I_d$.
  Moreover, by employing the embedding 
  $\StabMatDiag{\mcM} \subseteq \bigoplus_{d < m} {\DAlg}_{d}$,
  this $\LC^{ck}$-definable order induces a $\LC^{ck}$-definable order on
  $\StabMatDiag{\mcM}$. 
  This, in turn, means that we can define an ordered $\field F$-basis 
  for $\StabMatDiag{\mcM}$ in $\LC^{ck}$.
  This has the remarkable consequence that we can identify 
  the $\field F$-algebra $\StabMatDiag{\mcM}$ in $\LC^{ck}$
  up to isomorphism. Indeed, since we have access to an ordered $\field F$-basis for
  $\StabMatDiag{\mcM}$, we can also express all products of pairs of 
  basis elements again as $\field F$-linear combinations of these basis elements, 
  and this fully describes the algebra $\StabMatDiag{\mcM}$ in $\LC^{ck}$ 
  up to isomorphism (the coefficients that we obtain when we 
  express all products of pairs of basis elements again as $\field F$-linear 
  combinations of basis elements are known as \emph{structure constants} or
  \emph{structure coefficients} of the $\field F$-algebra, and are also 
  used, for instance, to encode algebras as inputs for algorithms).

 \paragraph*{(\ref{proofplan4}) Decomposition into small submodules} 
   \newcommand{\HDiag}{\ensuremath{\text{H}}}
   \newcommand{\HDecomp}{\ensuremath{\text{R}}}
  The next step is to decompose the
  $\StabMatDiag{\mcM}$-module $\HomMatDiag{\mcM}{\mcN}$
  into small submodules.
  At its core, the decomposition becomes possible due to the semi-simplicity of
  the module $\HomMatDiag{\mcM}{\mcN}$ and the definability of solution spaces
  of cocyclic linear equation systems in counting logic (Theorem~\ref{thm:solvcocylic}).
  
  Let us start by a simple observation. 
  Analogously to our definition of $\DAlg_d$, for each diagonal 
  block $I_d \times J_d$, for $d < m$, let us denote by
  $\HDiag_d$ the projection of the (diagonal) $\StabMatDiag{\mcM}$-module 
  $\HomMatDiag{\mcM}{\mcN}$ to the $d$-th diagonal block $I_d \times J_d$, that 
  is
  $\HDiag_d = \Diag_d(\HomMatDiag{\mcM}{\mcN}) = 
  \{ \Diag_d(S) : S \in \HomMatDiag{\mcM}{\mcN} \}$. 
  Then each $\HDiag_d$ forms a $\StabMatDiag{\mcM}$-module and we have
  \[ \HomMatDiag{\mcM}{\mcN} \subseteq \HDiag_0 \oplus \HDiag_1 \oplus \cdots \oplus \HDiag_\mmo.\]
  
  Before we proceed, let us remark that we can define, for each $d < m$, a 
  generating set (which consists of $I \times J$-matrices with entries in $\field F$) 
  for the $\StabMatDiag{\mcM}$-module $\HDiag_d$ in counting logic 
  (using a formula with at most $c \cdot k$ many variables, for large enough $c \geq 1$). 
  Indeed, we have seen before, cf.\ step~(\ref{proofplan3}), 
  how the $\StabMatDiag{\mcM}$-module
  $\HomMatDiag{\mcM}{\mcN}$ can be described as the solution space of a (cocyclic)
  linear equation system.
  Theorem~\ref{thm:solvcocylic} thus implies that we can find a $\LC^{ck}$-formula
  which defines a generating set for $\HomMatDiag{\mcM}{\mcN}$ in the ordered pair
  $(\mfA, \mfB)$ of CFI-structures $\mfA$ and $\mfB$. 
  Hence, by projecting this generating set to the block $I_d \times J_d$, we can
  obtain a generating set for $\HDiag_d$ in $\LC^{ck}$ as well.
  Moreover, these generating sets for the modules $\HDiag_d$ 
  have an important property that we will exploit frequently: we can linearly order
  these generating sets in $\LC^{ck}$ by fixing a pair $(i,j) \in I_d \times 
  J_d$ as parameter. This easily follows from the observation that the 
  matrices  in $\HDiag_d$ have non-trivial entries only on the diagonal block 
  $I_d \times J_d$. Moreover, by the cyclicity of CFI-structures we 
  know that we can order the (relevant part of the) index set $I_d \times 
J_d$ of such matrices in $\LC^{ck}$, hence we can also order $I_d 
  \times J_d$   matrices 
  (with entries in $\field F$) which leads to an \emph{ordered} generating 
  set for $\HDiag_d$. 
  The important consequence is that we can describe the isomorphism types of 
  the $\StabMatDiag{\mcM}$-modules $\HDiag_d$ in $\LC^{ck}$.
  To see this, recall that we can define an ordered
  basis for the $\field F$-algebra $\StabMatDiag{\mcM}$ in $\LC^{ck}$.
  Since we can, by fixing a parameter $(i,j) \in I_d \times J_d$, also obtain
  an ordered basis for $\HDiag_d$ it follows that wan can define the structure 
  coefficients of the $\StabMatDiag{\mcM}$-module $\HDiag_d$ in $\LC^{ck}$ as 
  well (in particular, the structure coefficients are independent of
 the choice of the parameter $(i,j) \in I_d \times J_d$).
  
 \medskip
 Our key goal in the remainder of our proof for 
Theorem~\ref{thm:main:definability} is to show that we can describe the 
  isomorphism type of the $\StabMatDiag{\mcM}$-module 
  $\HomMatDiag{\mcM}{\mcN}$ by means of a $\LC^{ck}$-formula.
  So far, we saw that we can describe the isomorphism type of the  
  surrounding $\StabMatDiag{\mcM}$-module $\bigoplus_{d < m } \HDiag_d$ in 
  $\LC^{ck}$.
  Unfortunately, we cannot simply transfer our arguments for the module 
$\bigoplus_{d < m } \HDiag_d$ to the case of the module 
$\HomMatDiag{\mcM}{\mcN}$.  
Although we have a $\LC^{ck}$-definable ordered basis for $\StabMatDiag{\mcM}$,
and although we can order the module $\bigoplus_{d < m } \HDiag_d$ locally,  
there is no hope to define a global linear order neither on $\bigoplus_{d < m } 
\HDiag_d$ nor on the submodule $\HomMatDiag{\mcM}{\mcN}$.
The simple reason is that, in general, matrices in $\bigoplus_{d < m } 
\HDiag_d$ or $\HomMatDiag{\mcM}{\mcN}$ are \emph{not} invariant under the 
automorphisms of the ordered pair of CFI-structures $(\mfA, \mfB)$, not even 
if we fix a (sublinear) number of parameters.

On the other hand, as long as we are only interested in determining the 
isomorphism type of a module, it is not necessary to define a
linear order on the module itself.
This is true, in particular, if we have access to a definable decomposition of 
the module into a direct sum of ``small'' submodules, such as in the case of 
$\bigoplus_{d < m} \HDiag_d$, because in this case we can describe the 
isomorphism type of the full module by  means of describing the isomorphism 
types of all small components, such as $\HDiag_d$, $d < m$, in the case 
of $\bigoplus_{d < m} \HDiag_d$. 
Indeed, for the case of $\bigoplus_{d < m} \HDiag_d$ we saw that it is 
straightforward to describe the isomorphism types of the submodules 
$\HDiag_d$, $d < m$, in  $\LC^{ck}$ simply because we can define a 
linearly ordered basis for $\HDiag_d$, $d < m$, in $\LC^{ck}$ (but this 
requires the choice of a parameter $(i,j) \in I_d \times J_d$).

The preceding discussion motivates our following strategy. 
We aim to apply the
idea of decomposing the module into ``small'' submodules, that we can control 
easily in $\LC^{ck}$, to the case of the module $\HomMatDiag{\mcM}{\mcN}$.
But, of course, the difficulty here is that we don't have access to a 
(definable) decomposition of the module $\HomMatDiag{\mcM}{\mcN}$ into ``small'' 
submodules as in the case of $\bigoplus_{d < m} \HDiag_d$.
Still, there now appears to be an obvious path: since $\HomMatDiag{\mcM}{\mcN}$ 
is a submodule of $\bigoplus_{d <m } \HDiag_d$, and since we have a nice 
decomposition of $\bigoplus_{d <m } \HDiag_d$ into the small modules 
$\HDiag_d$, we could try to transfer this decomposition to the submodule 
$\HomMatDiag{\mcM}{\mcN}$.
For $d < m$ let us define $\HDecomp_d$ as the 
$\StabMatDiag{\mcM}$-submodule of $\HomMatDiag{\mcM}{\mcN}$ that only contains 
those elements from  $\HomMatDiag{\mcM}{\mcN}$ which are non-trivial on 
summands $\HDiag_{d'}$ with $d' \geq d$, that is
\[ \HDecomp_d = \HomMatDiag{\mcM}{\mcN} \cap (\{0\} \oplus \cdots \oplus \{0\} 
\oplus \HDiag_d \oplus \cdots \oplus \HDiag_\mmo). \]
In particular, $\HDecomp_0 = \HomMatDiag{\mcM}{\mcN}$.
Moreover, $\HDecomp_{d} / 
\HDecomp_{d+1}$ is isomorphic to a $\StabMatDiag{\mcM}$-submodule of 
$\HDiag_{d}$, $d < m$, where we agree 
that $\HDecomp_{m} = \bigoplus_{d < m }\{0\}$, 
and we obtain a chain of submodules as
\[ \HomMatDiag{\mcM}{\mcN} = \HDecomp_0 \supseteq \HDecomp_1 \supseteq \cdots 
\supseteq \HDecomp_\mmo \supseteq R_m = 0.    \]
Furthermore, in an analogous way as for the full module 
$\HomMatDiag{\mcM}{\mcN}$, we can construct, for each $d < m$, a (cocyclic) 
linear equation system whose solution space is $\HDecomp_d$. Hence, by another 
application of Theorem~\ref{thm:solvcocylic} it follows that we can define in 
$\LC^{ck}$ a generating set for each of the submodules $\HDecomp_d$, $d < m$.

The final step is to use the chain of submodules 
$R_0 \supseteq R_1 \supseteq \cdots \supseteq R_\mmo$ in order to decompose 
$\HomMatDiag{\mcM}{\mcN}$ into a direct sum of ``small'' submodules.
We proceed recursively, so let us assume that we already know how to decompose 
the $\StabMatDiag{\mcM}$-module $\HomMatDiag{\mcM}{\mcN}$, in 
a $\LC^{ck}$-definable way,  as a direct sum 
\[ \bigoplus_{d < t} T_d \oplus \HDecomp_{t},\]
where $t \leq m$ and where each of the $\StabMatDiag{\mcM}$-submodules 
$T_d$, $d < t$, of $\HomMatDiag{\mcM}{\mcN}$ is ``small'' in the sense that we 
can define a linearly ordered basis for $T_d$ in $\LC^{ck}$ by only using a 
constant number of parameters from $(\mfA, \mfB)$ (in fact, a single parameter 
tuple $(i,j)\in I_d \times J_d$ will be sufficient).
Then we only need to explain how we can express, by means of a 
$\LC^{ck}$-formula, a 
decomposition of $\HDecomp_{t}$ as a direct sum $T \oplus \HDecomp_{t+1}$ 
together with a linearly ordered basis for $T$ (where we can use a 
constant number of parameters to define the basis, but \emph{not} the 
decomposition).

The crucial ingredient for our argument is the semi-simplicity of the  
$\StabMatDiag{\mcM}$-module $\HomMatDiag{\mcM}{\mcN}$ 
which we proved in step~(\ref{proofplan2}) by an application of 
Maschke's Theorem.
Indeed, this result already implies the existence of a complement for 
$\HDecomp_{t+1}$ in $\HDecomp_t$, that is it proves the existence of a
$\StabMatDiag{\mcM}$-submodule $T$ of $\HDecomp_t$ such that
\[ \HDecomp_{t} = T \oplus \HDecomp_{t+1}. \]
Still, the immediate question is: why should the pure existence 
of such a submodule $T$ say anything about the definability of a linearly 
ordered basis in $\LC^{ck}$?

In order to approach this question, we first need to recall one of our 
earlier observations, namely that the $\field F$-algebra $\StabMatDiag{\mcM}$ 
contains the automorphism group $\Gamma$ of the CFI-structure $\mfA$ (in the 
sense that we view automorphisms $\pi \in \Gamma$ as $I \times I$-permutation 
matrices $\Pi\colon I \times I \to \{0, 1\} \subseteq \field F$ as above).
Secondly, we observe that matrices in $T$ are unique when projected onto
the $t$-th diagonal block $I_t \times J_t$ (that is onto the module $\HDiag_t$).
Indeed, assume that $X, Y \in T$. Then we claim that either $X = Y$ or that 
$\Diag_t(X) \neq \Diag_t(Y)$.
In fact, if $\Diag_t(X) = \Diag_t(Y)$, then $X - Y \in 
\HDecomp_{t+1}$ (because, $\Diag_t(X-Y) = 0$) and $X - Y \in T$ (because 
modules are closed under differences). Hence, $X - Y \in 
\HDecomp_{t+1} \cap T = \{ 0\}$, which implies $X = Y$.
As we see next, these two facts together allow us to show that we can define a 
linearly ordered basis of $T$ in $\LC^{ck}$.

Let us denote by $\Delta_\mfA = \Stab(i)$, $i \in I_t$, the stabiliser 
group of the  orbit $I_t$ in the CFI-structure $\mfA$, that is the group 
consisting of all automorphisms 
$\pi$ of $\mfA$ which fix some (and therefore all) $i \in I_t$ 
(since the automorphism group of $\mfA$ is Abelian, the stabiliser groups 
for all elements $i \in I_t$ are identical).
Then, obviously, for each $\Pi \in \Delta_\mfA$ we have that $\Diag_t(\Pi)$ is 
the identity matrix.
Hence, if we let $X \in T$ be arbitrary, then $\Diag_t(\Pi \cdot X) = 
\Diag_t(X)$ for all $\Pi \in \Delta_\mfA$. But then, because matrices in $T$ are 
unique on the $t$-th diagonal block, we can actually conclude that $\Pi \cdot X 
= X$ for all $\Pi \in X$.
It can be shown, in precisely the same way, that also $X \cdot \Pi = X$ holds 
for every permutation matrix $\Pi \colon J \times J \to \{ 0,1 \}$ which 
corresponds to an automorphism of the CFI-structure $\mfB$ that stabilises the 
$t$-th diagonal block $J_t \times J_t$, that is $X \cdot \Pi = X$ holds for 
every $\Pi \in \Delta_\mfB$ where $\Delta_\mfB$ denotes the set of automorphisms 
$\pi$ of $\mfB$ which fix some (any) $j \in J_t$.
Altogether, if we denote by $\Delta$ the set of automorphisms $\pi$ of the 
ordered pair $(\mfA, \mfB)$ which pointwise fixes some (any) parameter tuple
$(i,j) \in I_t \times J_t$, that is $\pi(i,j) = (i,j)$, then we have that 
$\Delta = \Delta_\mfA \times \Delta_\mfB$ and $\Pi \cdot X \cdot \Pi^{-1} = X$ 
for every $X \in T$.
In other words, every matrix $X \in T$ in a complement $T$ of the module 
$\HDecomp_{t+1}$ in $\HDecomp_t$ has to be invariant on $\Delta$-orbits.
Of course, this also implies that $T$ itself must be invariant under $\Delta$ 
(but not necessarily under the action of the full automorphism group of 
$(\mfA, \mfB)$; recall that complements do not need to be unique in 
general).

Now, again because of the homogeneity of CFI-structures, we can easily define a 
linear order on the set of all matrices $X\colon I \times J \to \field F$ which 
are invariant under the action of $\Delta$ in counting logic by using at most 
$ck$ many variables and by fixing only a single parameter tuple $(i,j) \in I_t 
\times J_t$. This means that a $\LC^{ck}$-formula that uses a single parameter 
tuple $(i,j) \in I_t \times J_t$ can quantify over all possible 
submodules of $\HDecomp_{t+1}$ that consist of such $\Delta$-invariant matrices
$X \colon I \times J \to \field F$ only.
In particular, this means that a $\LC^{ck}$-formula can fix a complement $T$ of 
$\HDecomp_{t+1}$ in $\HDecomp_t$ together with an ordered basis by fixing a
single parameter tuple $(i,j) \in I_t \times J_t$ only.

This concludes our argument and we obtain the desired decomposition of 
$\HomMatDiag{\mcM}{\mcN}$ as $\bigoplus_{d < m} T_d$ where each of the 
submodules $T_d$ denotes a complement of $\HDecomp_{t+1}$ in 
$\HDecomp_t$ which we can fix, together with an ordered basis, in $\LC^{ck}$ by 
selecting a single parameter tuple $(i,j) \in I_t \times J_t$ only.
Let us stress that these complements $T_d$ may not be unique and, indeed, 
depending on the choice of parameters $(i,j) \in I_t \times J_t$ we may end up
with different complements $T_d$ of $\HDecomp_{t+1}$ in $\HDecomp_t$.
However, this clearly doesn't cause any harm as long as we are only interested 
in the isomorphism types of these complements (the reader should think of this 
as a canonisation procedure: we only need to express the isomorphism type of 
the module $\HomMatDiag{\mcM}{\mcN}$ in $\LC^{ck}$, but we don't need to 
define an explicit isomorphism from the abstract module into an ordered copy).

\paragraph*{(\ref{proofplan5}) Determining the isomorphism type}
We can finally complete our proof of Theorem~\ref{thm:main:definability} by 
putting everything together.
First of all, we saw that in order to show that the $\StabMat{\mcM}$-module
$\HomMat{\mcM}{\mcN}$ contains an invertible matrix $S \in \HomMat{\mcM}{\mcN}$ 
we can equivalently verify that the ``diagonal'' 
$\StabMatDiag{\mcM}$-module $\HomMatDiag{\mcM}{\mcN}$ is cyclic 
(Step~(\ref{proofplan3})).
We then proved as a second step that we can define, by means of a 
$\LC^{ck}$-formula, an ordered basis for the $\field F$-algebra 
$\StabMatDiag{\mcM}$ and we showed, by an application of Maschke's Theorem, that 
this algebra is semisimple (Step~(\ref{proofplan2})).
Using the semi-simplicity of $\StabMatDiag{\mcM}$, we further explained how 
one can decompose the $\StabMatDiag{\mcM}$-module $\HomMatDiag{\mcM}{\mcN}$, in 
a $\LC^{ck}$-definable way, into ``small'' submodules $T_d$, $d < m$, for which 
we can, furthermore, define an ordered basis in $\LC^{ck}$ by fixing a single 
parameter tuple $(i,j) \in I_d \times J_d$.

From this decomposition of $\HomMatDiag{\mcM}{\mcN}$ as $\bigoplus_{d < m} T_d$ 
we can now easily extract the isomorphism type of $\HomMatDiag{\mcM}{\mcN}$ in 
$\LC^{ck}$.
The reason is that we can determine the isomorphism type of each 
submodule $T_d$, $d < m$, individually. This is because we have a 
$\LC^{ck}$-definable ordered basis for each submodule $T_d$, $d < m$, and also 
for the $\field F$-algebra $\StabMatDiag{\mcM}$, which means that we can 
determine the structure coefficients for each submodule $T_d$, $d < m$, in 
$\LC^{ck}$. These structure constants clearly determine the isomorphism types of 
the submodules $T_d$, $d <m$. Having this we know how to  decompose each $T_d$ 
into a direct sum of simple $\StabMatDiag{\mcM}$-modules and thus we know the 
isomorphism type of the full module $\HomMatDiag{\mcM}{\mcN}$ as it is 
determined by the multiplicities of simple $\StabMatDiag{\mcM}$-modules as 
they occur in any decomposition of $\HomMatDiag{\mcM}{\mcN}$ into a direct sum 
of simple modules.
Finally, since the isomorphism type of $\HomMatDiag{\mcM}{\mcN}$ 
determines whether $\HomMatDiag{\mcM}{\mcN}$ is cyclic or not, we have  
completed our proof of Theorem~\ref{thm:main:definability}.

\section{Main results}
In this section we spell out the consequences of the main technical
result, Theorem~\ref{thm:main:cfi}, for approximations of isomorphism
and for logics with linear-algebraic operators.

With regard to the relations $\IMequiv{k}{Q}$ as approximations of
isomorphism, it follows immediately that as long as $Q \neq \Primes$,
i.e.\ $Q$ is not the set of all primes, there is no $k$ for which
$\IMequiv{k}{Q}$ coincides with isomorphism on all structures.

\begin{cor}
  If $Q\neq \Primes$, there is no fixed $k$ such that $\IMequiv{k}{Q}$
  coincides with isomorphism on all structures.
\end{cor}
\begin{proof}
Fix a prime $p \not\in Q$.  Then, for each $k$, we have, by
Theorem~\ref{thm:main:cfi} a pair of structures  $\mfA =
\CFIgraph{G_n}{p}{\lambda}$ and $\mfB = \CFIgraph{G_n}{p}{\sigma}$
that are $\linisom{\field F_q}{\ell}{\LC^{k}}$, for all $q\neq p$,
though $\sum \lambda \neq \sum \sigma$.  It follows that $\mfA
\IMequiv{k}{Q} \mfB$, but $\mfA \not\cong \mfB$, by Theorem~\ref{thm:cfi-isomorphism}.
\end{proof}

It should be noted that this was proved in a special case by
Holm~\cite{Holm10}.  To be precise, we can further parameterise the
equivalence relations $\IMequiv{k}{Q}$ by the parameter $\ell$, as in Definition~\ref{def:logic}.  That
is, in the iterative definition of $\IMequiv{k}{Q}$, we only ever
consider $A^{\ell} \times A^{\ell}$ for some fixed $\ell$.  Then, Holm
shows that in the case when $\ell = 1$, the resulting equivalence
relation does not capture isomorphism whenever $Q \neq \Primes$.  It
was left as an open question whether this could be proved in general.
Our result establishes this, and required substantial new algebraic
machinery.  The interesting open question remaining, of course, is to
establish such a result in the case when $Q = \Primes$.

The consequences for the expressive power of the logic $\LALogic$ are
also immediate.
\begin{cor}\label{cor:undefinable}
   If $Q\neq \Primes$, there is a class of structures that is not
   definable in $\LALogic(Q)$.
\end{cor}
\begin{proof}
  Fix a prime $p \not\in Q$ and consider the class $\mcC$ of structures of
  the form $\CFIgraph{G_n}{p}{\lambda}$ where $\sum \lambda = 0$
  (i.e.\ what we called the \emph{CFI-problem}.
  This is an isomorphism-closed class of structures by
  Theorem~\ref{thm:cfi-isomorphism}.  Suppose it were defined by a
  sentence $\phi$ of $\LALogic(Q)$.  Let $\ell$ the maximum dimension of an
  interpretation used with any quantifier in $\phi$ and choose $k$
  such that $k \geq 3\ell$ and $k$ is greater than the number of
  variables in $\phi$.  Then, by Theorem~\ref{thm:main:cfi}, we have
  a structure  $\mfA = \CFIgraph{G_n}{p}{\lambda} \in \mcC$ which is
  $\linisom{\field F_q}{\ell}{\LC^{k}}$ to every structure
  $\CFIgraph{G_n}{p}{\sigma}$.  Letting $\mfB$ be such a structure where $\sigma \neq
  0$, we have, by Lemma~\ref{lem:elem-equiv}, that $\mfB \models \phi$,
  contradicting the assumption that $\phi$ defines~$\mcC$.
\end{proof}

It should be noted that the class of structures $\mcC$ defined in the
proof of Corollary~\ref{cor:undefinable} is decidable in polynomial
time.  This is because the class can be decided by solving systems of
linear equations, for example by Gaussian elimination.  Thus, we know
that $\LALogic(Q)$ cannot express some PTIME property as long as $Q
\neq \Primes$.  Since this logic subsumes any extension of fixed-point
logic with $Q$-linear algebraic operators, we also have the following
conclusion.
\begin{cor}
   If $Q\neq \Primes$, no extension of fixed-point logic with $Q$-linear algebraic
  operators captures PTIME.
\end{cor}

We can say more.  The class $\mcC$ is not just decidable in PTIME, but
also definable in \emph{choiceless polynomial time} (CPT) (see~\cite{Pakusa16}).
We do not define the class CPT here but details may be found in~\cite{BlassGurevichShelah}
  Thus, the following corollary is immediate.
\begin{cor}
   If $Q\neq \Primes$, no extension of fixed-point logic with $Q$-linear algebraic
  operators captures CPT.
\end{cor}

On the other hand it remains an intriguing open question whether CPT captures all of rank logic, for example.

\bibliographystyle{plain}
\bibliography{matrix-bib}

\end{document}